\documentclass[journal]{IEEEtran}

\usepackage{hyperref}
\hypersetup{hidelinks}
\usepackage{xcolor,soul,framed} %,caption

\colorlet{shadecolor}{yellow}
\usepackage[pdftex]{graphicx}
\DeclareGraphicsExtensions{.pdf,.jpeg,.png}

\usepackage{amsfonts,amssymb}
\usepackage{mathrsfs}
\usepackage{amsthm}
\usepackage[cmex10]{amsmath}
\usepackage{mathtools}
\usepackage{bbm}
\DeclareMathAlphabet{\mathbbb}{U}{bbold}{m}{n}
\usepackage{array}
\usepackage{aligned-overset}
\usepackage{flushend}
\usepackage{physics}
\usepackage{cite}
\usepackage{makecell}
\usepackage{threeparttable}
\usepackage{comment}
\usepackage{colortbl}
\let\labelindent\relax
\usepackage{enumitem}

\theoremstyle{definition}

\newtheorem{lemma}{Lemma}
\newtheorem{definition}{Definition}

\newtheorem{proposition}{Proposition}

\theoremstyle{remark}
\newtheorem{remark}{Remark}

\newcommand{\pha}[1]    {\underline{#1}}

\newcommand{\tabitem}{~~\llap{\textbullet}~~}
\newcommand{\tabitemm}{~~\llap{}~~}

\definecolor{color_voltage_forming}{RGB}{89,169,90}
\definecolor{color_current_forming}{RGB}{77,133,189}
\definecolor{color_cross_forming}{RGB}{247,144,71}

\begin{document}
\bstctlcite{IEEEexample:BSTcontrol}
    \title{Cross-Forming Control and Fault Current Limiting for Grid-Forming Inverters}
    \author{Xiuqiang~He,~\IEEEmembership{Member,~IEEE,}                         
            Maitraya~Avadhut~Desai,~\IEEEmembership{Graduate Student Member,~IEEE,} \\
            Linbin~Huang,~\IEEEmembership{Member,~IEEE,}
          and~Florian~Dörfler,~\IEEEmembership{Senior Member,~IEEE}% <-this % stops a space
  \thanks{The authors have filed a patent associated with this work, “Method and a controller for controlling a grid forming converter”, pending to ETH Zurich.}
  \thanks{X. He, L. Huang, and F. Dörfler are with the Automatic Control Laboratory, ETH Zurich, 8092 Zurich, Switzerland. M.A. Desai is with the Power System Laboratory, ETH Zurich, 8092 Zurich, Switzerland. Emails: xiuqhe@ethz.ch, desai@eeh.ee.ethz.ch, linhuang@ethz.ch, dorfler@ethz.ch.}}

% ====================================================================
\maketitle

%\pagenumbering{gobble} % no page number

% === ABSTRACT ====================================================================
\begin{abstract}

This article proposes a ``cross-forming" control concept for grid-forming inverters operating against grid faults. Cross-forming refers to \textit{voltage angle forming} and \textit{current magnitude forming}. It differs from classical grid-forming and grid-following paradigms that feature voltage magnitude-and-angle forming and voltage magnitude-and-angle following (or current magnitude-and-angle forming), respectively. The cross-forming concept addresses the need for inverters to remain grid-forming (particularly \textit{voltage angle forming}, as required by grid codes) while managing \textit{fault current limitation}. Simple and feasible cross-forming control implementations are proposed, enabling inverters to quickly limit fault currents to a prescribed level while preserving voltage angle forming for grid-forming synchronization and providing dynamic ancillary services, during symmetrical or asymmetrical fault ride-through. Moreover, the cross-forming control yields an equivalent system featuring a \textit{constant} virtual impedance and a ``normal form" representation, allowing for the extension of previously established transient stability results to include scenarios involving current saturation. Simulations and experiments validate the efficacy of the proposed cross-forming control implementations.
\end{abstract}

% === KEYWORDS ====================================================================

% === I. INTRODUCTION =============================================================
\section{Introduction}

\IEEEPARstart{G}{rid}-forming inverters play a crucial role in future power systems in autonomously regulating grid frequency and voltage. While a grid-forming inverter operates like a voltage source, limiting its current during grid disturbances is critical to prevent potential overcurrent damage. Moreover, grid-forming inverters should sustain transient stability during grid faults, ensuring synchronization while transitioning from one operating state to another. Transient stability is crucial for successful fault ride-through (FRT) and ancillary services during FRT, such as fault current injection and phase jump power provision. These are requirements outlined in recent grid-forming specifications, e.g., the Great Britain, Australian, and European grid codes \cite{gb2024,aemo2023voluntary,entso2024}; see a survey in \cite{bahrani2024grid,ghimire2024functional}. To satisfy these requirements, grid-forming inverters should maintain grid-forming synchronization and provide FRT ancillary services as continuously as possible, even when the current reaches the limit \cite{gb2024,aemo2023voluntary,entso2024,bahrani2024grid,ghimire2024functional}. These requirements involve technical challenges in limiting fault current, maintaining transient stability (synchronization), and providing FRT ancillary services simultaneously \cite{baeckeland2024overcurrent,ordono2024current}.

\subsection{Related Work}

When grid-forming inverters operate under normal grid conditions, i.e., without current saturation, grid-forming synchronization and ancillary services provisions are well understood. More specifically, the transient stability of current-unsaturated grid-forming inverters has been widely investigated in the literature; see \cite{pan2020transient} for a comparative study and \cite{rosso2021gridforming} and \cite{zhang2021gridforming} for a review. In parallel, the provision of dynamic ancillary services for grid-forming inverters under normal operating conditions has also been extensively explored in the literature; see \cite{dorfler2023control} for a survey. In contrast to normal operating conditions, the critical challenge under grid fault conditions arises from current limiting. In the literature, current limiting for grid-forming inverters is addressed with three typical strategies: 1) adaptive/threshold virtual impedance \cite{paquette2015virtual,qoria2020current}; 2) current limiter cascaded with virtual admittance \cite{rosso2021implementation,kkuni2024effects,fan2022equivalent,zhang2023current,saffar2023impacts}; and 3) current-forming voltage-following control \cite{huang2019transient,rokrok2022transient,liu2023dynamic,li2023transient,wang2023transient,xin2021dual,schweizer2022grid}. Their different merits and shortcoming are described below.

\textit{1) Adaptive/Threshold Virtual Impedance:} This strategy uses current feedback to increase the virtual impedance magnitude adaptively, thus reducing the voltage reference and ultimately reducing the overcurrent \cite{paquette2015virtual,qoria2020current}. This strategy can preserve the original grid-forming synchronization and ancillary services provision. However, the parameter tuning thereof is quite complicated due to the design being tuned to ensure reliable current limiting in the face of the worse-case fault \cite{paquette2015virtual}. In most cases, the strategy underutilizes the overcurrent capability due to the inherent performance limitations of proportional feedback regulation. Moreover, the equivalent virtual impedance varies with state-dependent overcurrent severity, distorting the power output characteristics \cite{liu2022current,qoria2020critical} and thus deteriorating transient stability and complicating its analysis.

\textit{2) Current Limiter Cascaded With Virtual Admittance:} This strategy uses a current limiter along with a virtual admittance. The virtual admittance acts as a voltage proportional regulator \cite{fan2022equivalent,zhang2023current,saffar2023impacts}, avoiding limiter-induced windup issues inherent in integrator-included regulators (e.g., classical proportional-integral regulators), and meanwhile preserving grid-forming synchronization and ancillary services provision. This strategy is simple to implement, easy to tune, and fully utilizes the overcurrent capability. However, the current saturation still leads to a varying equivalent impedance \cite{fan2022equivalent,zhang2023current}, similar to the virtual impedance strategy. Likewise, the varying impedance distorts the power-angle characteristics (from an ideal sine function to a highly nonlinear one) \cite{fan2022equivalent,kkuni2024effects,zhang2023current,saffar2023impacts}, jeopardizing transient stability and complicating its analysis.

\textit{3) Saturated Current-Forming Control:} This strategy deactivates the voltage control during current saturation while preserving the current control \cite{huang2019transient,rokrok2022transient,liu2023dynamic,li2023transient,wang2023transient}. The reference angle of the current vector control is generated by a frequency droop control. The strategy falls into the category of current-forming (grid-following) controls \cite{li2022revisiting} since the controlled variable is current rather than voltage. Hence, while effective for current limiting, it cannot provide voltage-forming behaviors independently, not fulfilling the requirements of grid codes.

Concisely, the current-limiting prior arts are affected by multiple shortcomings, rendering them insufficient for fulfilling the requirements for current limiting, transient stability guarantees, and grid-forming ancillary services during FRT. A thorough review and comparison of current-limiting strategies can be found in recent review articles \cite{baeckeland2024overcurrent} and \cite{ordono2024current}.

% =======
% FIG
% =======
\begin{figure*}
  \begin{center}
  \includegraphics{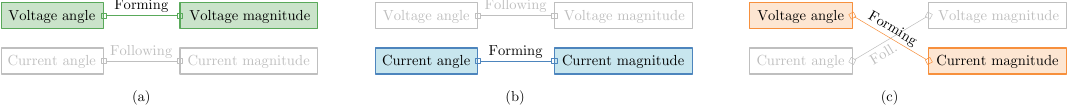}
  \caption{Illustrations of (a) grid-forming (voltage-forming current-following), (b) grid-following (voltage-following current-forming), and (c) cross-forming modes.}
  \label{fig:forming-illustration}
  \end{center}
\end{figure*}

\subsection{Motivations}

We aim to develop a new control concept, termed cross-forming, to address fault current limiting in grid-forming inverters. The concept of cross-forming integrates both voltage angle forming and current magnitude forming characteristics, driven by three critical observations. First, the preservation of the voltage angle forming is critical to ensure grid-forming synchronization and provide essential frequency- and angle-sensitive ancillary services, such as phase jump \textit{active power} delivery and fault \textit{reactive current} injection \cite{gb2024,aemo2023voluntary,entso2024}. Second, during periods of current saturation, the inverter terminal voltage magnitude has been observed to passively follow the grid operation, a phenomenon recently described as ``voltage decline" in \cite{zhang2024control}, thus preventing the need for explicit voltage magnitude forming under these conditions. Third, the current-limiting constraint suggests that the current magnitude forming is required. These facts/observations motivate us to enforce the \textit{voltage angle forming} and \textit{current magnitude forming} characteristics. Therefore, this type of control exhibits cross-forming behaviors during current saturation, as illustrated in Fig.~\ref{fig:forming-illustration}.

\subsection{Contributions}

The contributions of this article are summarized as follows:
\begin{itemize}
    \item We clarify that voltage angle forming is crucial for the inverters to satisfy grid code requirements, current magnitude limiting is also necessary, and voltage magnitude forming is lost under current saturation. Accordingly, we present the concept of \textit{cross-forming control} and show that it naturally achieves fault current limitation, grid-forming synchronization, FRT services provision, etc.
    \item We develop two simple and feasible cross-forming control implementations that apply to symmetrical and asymmetrical grid fault conditions. The cross-forming control yields an equivalent circuit featuring a \textit{constant virtual impedance} and a current-dependent virtual internal voltage magnitude, differing from prior results. 
    \item Based on the equivalent circuit, we present an equivalent normal form of the cross-forming system, which represents structurally identical synchronization dynamics to the normal form of current-unsaturated grid-forming systems. Therefore, the existing transient stability analysis approaches and results for current unsaturation are readily extended to current-saturated conditions.
    \item As a minor contribution, we provide a survey of prior arts in appendices, including grid-forming controls, negative-sequence controls, and typical current-limiting strategies, facilitating the understanding of our contributions, stimulating further research of grid-forming technologies, and offering useful references for practical applications.
\end{itemize}
We highlight that the proposed cross-forming control offers a fast response, fully utilizes overcurrent capabilities, adapts to various disturbances, is simple to implement, easy to tune, and robust in stability performance. Hence, it stands as a promising solution for future grid-forming product development.

\subsection{Organization and Notation}

The remainder of the article is organized as follows. Section~\ref{sec:grid-forming-modes} outlines grid-forming control objectives in the presence of grid faults, defines the cross-forming concept, and discusses its capabilities to serve grid code requirements. Section~\ref{sec:cross-forming-controls} presents the cross-forming control design and a comparison with the existing strategies. Section~\ref{sec:transient-stability} derives the equivalent normal form and extends the prior transient stability results. Simulation and experimental validations are provided in Section~\ref{sec:validations}. Section~\ref{sec:conclusion} concludes the article. Finally, a self-contained review of relevant prior arts is provided in Appendices~\ref{appendix:standard-gfm} to \ref{appendix:current-limiting}. 

\textit{Notation:} We use $\pha v \coloneqq v_{\alpha} +j v_{\beta}$ and $\pha i \coloneqq i_{\alpha} +j i_{\beta}$ to denote voltage and current vectors in the stationary reference frame, respectively, and $\pha v_{\mathrm{dq}} \coloneqq v_{\mathrm{d}} +j v_{\mathrm{q}}$ and $\pha i_{\mathrm{dq}} \coloneqq i_{\mathrm{d}} +j i_{\mathrm{q}}$ to denote voltage and current vectors in the rotational reference frame, respectively. The underlines throughout the article indicate complex variables. $\Re{\cdot}$, $\Im{\cdot}$, $\abs{\cdot}$, and $\angle(\cdot)$ denote the real part, the imaginary part, the modulus, and the angle of a complex variable, respectively, and $\mathrm{conj(\cdot)}$ denotes the associated complex conjugate. Superscripts $^+$ and $^-$ indicate positive- and negative-sequence components, respectively.

\section{Grid-Forming Control Objectives, Operation Modes, and FRT Services Provision}
\label{sec:grid-forming-modes}

The aim of this section is three-fold. We first outline grid-forming control objectives under grid faults from the perspective of stability and grid code requirements. Then, we define and categorize typical operation modes of power inverters, and finally, we investigate their FRT services provision capabilities under current saturation to meet grid code requirements.

\subsection{Grid-Forming Control Objectives Under Grid Faults}
\label{sec:control-objectives}

The control objectives of grid-forming inverters under grid faults should ideally remain the same as under normal conditions. However, grid-forming inverters must limit their output current to prevent overcurrent damage. Moreover, under asymmetrical grid faults, the control of negative-sequence components should also be of concern. We summarize the grid-forming control objectives under grid faults in the following.

\subsubsection{Current Limiting} The current limit of the inverters must be respected regardless of the control objectives or strategies used. This requires that the maximum phase current magnitude remains within or at the limit. Moreover, voltage limits preventing overmodulation are also of concern in practice, but mostly relevant for high-/over-voltage ride-through.

\subsubsection{Grid-Forming Synchronization} Synchronization of angular frequency is a necessary condition for closed-loop control and stability, typically achieved by grid-forming-type feedback control architecture (i.e., grid-forming synchronization or more specifically \textit{voltage angle forming} synchronization).

\subsubsection{Positive-Sequence FRT Services Provision} The regular dynamic ancillary services under normal grid conditions typically include frequency and voltage regulation. Concerning grid faults, grid-forming devices are required by grid codes to maintain a constant (positive-sequence) internal voltage phasor in a short time frame (before current saturation) and to be able to supply fast fault reactive current or phase jump active power after grid voltage dips or phase jumps, respectively \cite{gb2024,aemo2023voluntary,entso2024} (for instance, as specified in the GB grid code \cite{gb2024}, with an initial delay less than $5$ ms and full activation time less than $30$ ms).

\subsubsection{Negative-Sequence FRT Services Provision} In Fig.~\ref{fig:phase-sequence-circuit}, we show generic equivalent circuits in the phase domain and the sequence domain under asymmetrical grid faults. The FRT services in the negative-sequence domain are also of concern, and admit typically multiple options, depending on specific requirements; see Appendix~\ref{appendix:negative-sequence-current}. The options include: i) negative-sequence grid-forming (behaving as a voltage source), symmetric to positive-sequence grid-forming \cite{awal2023double}; ii) negative-sequence current injection (behaving as a current source) for power oscillation suppression \cite{zheng2018flexible}; and iii) negative-sequence voltage mitigation \cite{nasr2023controlling} (behaving as a virtual impedance), a requirement of grid codes; see IEEE Std. 2800 \cite{ieee2800}).

% =======
% FIG
% =======
\begin{figure}
  \begin{center}
  \includegraphics{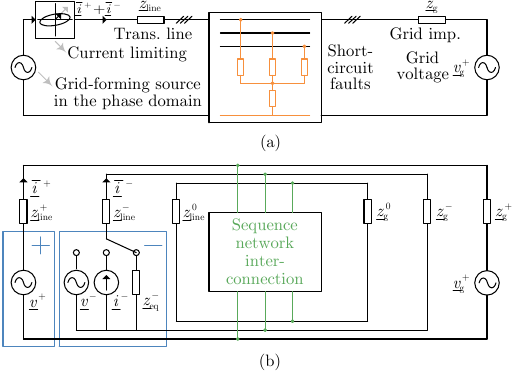}
  \caption{(a) Phase-domain and (b) sequence-domain generic equivalent circuits for grid-forming inverters under grid short-circuit faults, where the grid-forming operation is maintained in the positive sequence while the negative sequence admits different options among negative-sequence grid-forming \cite{awal2023double}, current injection \cite{zheng2018flexible}, or impedance emulation \cite{nasr2023controlling}.}
  \label{fig:phase-sequence-circuit}
  \end{center}
\end{figure}

\subsection{Multiple Forming Control Modes}
\label{sec:three_types_def}

While extensive research on grid-forming inverters is ongoing, there is no universally accepted definition of grid-forming. In this work, we follow the definition in \cite{li2022revisiting}, which is based on the structural duality perspective, to introduce three types of forming control modes, as displayed in Fig.~\ref{fig:forming-illustration}. The first two modes are well known, while the last is novel, which inspires new solutions to address current limiting and transient stability challenges in grid-forming inverters under grid faults.

\begin{definition}
\label{def:voltage-forming}
\textit{Grid-Forming Mode:} A voltage-forming current-following mode whose voltage angle and magnitude are independently controlled and whose current depends on an exogenous grid, as illustrated in Fig.~\ref{fig:forming-illustration}(a).
\end{definition}

\begin{definition}
\label{def:current-forming}
\textit{Grid-Following Mode:} A current-forming voltage-following mode whose current angle and magnitude are independently controlled and whose voltage depends on an exogenous grid, as shown in Fig.~\ref{fig:forming-illustration}(b).
\end{definition}

\begin{definition}
\label{def:cross-forming}
\textit{Cross-Forming Mode:}\footnote{From a combinatorial point of view, there exist other cross-forming modes, such as controlling the current angle and the voltage amplitude, but they may not be relevant to the grid-forming control objectives under grid faults.} A cross-forming mode whose voltage angle and current magnitude are independently controlled and whose voltage magnitude and current angle depend on an exogenous grid, as illustrated in Fig.~\ref{fig:forming-illustration}(c).
\end{definition}

The comparison of these three forming control modes is shown in Table~\ref{tab:gfm-operation-mode}. Since power grids provide voltage source behaviors, the concept of grid-forming and grid-following is more precisely \textit{voltage-forming} and \textit{voltage-following}, respectively \cite{li2022revisiting}. In particular, a voltage-forming source and a current-forming source share a dual relationship \cite{li2022revisiting}. Voltage-forming control modes typically include droop control, virtual synchronous machine (VSM), dispatchable virtual oscillator control (dVOC), and SM-matching, with forming characteristics in voltage angle and magnitude. In contrast, current-forming modes feature forming characteristics in current angle and magnitude. Typical examples are phase-locked loop (PLL)-based current control and the control schemes that are structurally dual to voltage-forming controls \cite{li2022revisiting,xin2021dual}.

\begin{table*}
\centering
\caption{Three Forming Control Modes, Characteristics, and FRT Services Provision Capabilities}
\begin{threeparttable}
\arrayrulecolor{black!10}
\begin{tabular}{llll}
\arrayrulecolor{black}
\hline\hline
 & Voltage-forming mode  & Current-forming mode  & Cross-forming mode \\
\hline
\makecell[l]{Illustrations} & \makecell[l]{\includegraphics[width=4.4cm]{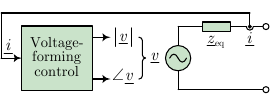}} & \makecell[l]{\includegraphics[width=4.4cm]{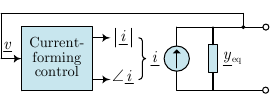}} & \makecell[l]{\includegraphics[width=4.4cm]{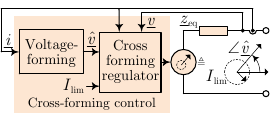}}\\
\arrayrulecolor{black!10}
\hline
Examples & \makecell[l]{Droop, VSM, \\ dVOC, complex droop control, \\ SM-matching, and dual-port control\tnote{1}} & \makecell[l]{PLL-based current control, and \\ the dual counterparts of voltage \\ forming controls} & \makecell[l]{Two types; see Section~\ref{sec:cross-forming-controls}} \\
\hline
References & \makecell[l]{\cite{rosso2021gridforming,dorfler2023control}} & \makecell[l]{\cite{he2019resynchronization,huang2019transient}} & This article \\
\hline
Voltage behavior & {\color{color_voltage_forming}Angle and magnitude forming} & Angle and magnitude following & {\color{color_cross_forming}Angle forming}, magnitude following \\
\hline
Current behavior & Angle and magnitude following & {\color{color_current_forming}Angle and magnitude forming} & Angle following, {\color{color_cross_forming}magnitude forming} \\
\hline
Synchronization & {Voltage angular frequency} & {Current angular frequency} & {Voltage angular frequency} \\
\hline
Fault current injection & Yes, natural response & Yes, but need setpoint adjust & Yes, natural response \\
\hline
Phase jump power delivery & Yes, natural response & Yes, but need setpoint adjust & Yes, natural response \\
\hline
Negative-sequence services\tnote{2} & Yes, flexible & Yes, flexible & Yes, flexible \\
\hline
Current limiting & No, additional remedies required\tnote{3} & Yes, inherent & Yes, inherent \\
\arrayrulecolor{black}
\hline \hline
\end{tabular}
\label{tab:gfm-operation-mode}
\begin{tablenotes}
    \item[1] A comprehensive review of grid-forming (voltage-forming) controls is provided in Appendix~\ref{appendix:standard-gfm}.
    \item[2] A comprehensive review and rigorous proof of multiple control modes for negative-sequence components are presented in Appendix~\ref{appendix:negative-sequence-current}.
    \item[3] A categorized review of typical existing current-limiting strategies for the voltage-forming control mode is presented in Appendix~\ref{appendix:current-limiting}.
\end{tablenotes}
\end{threeparttable}
\end{table*}

The control objectives outlined in Section~\ref{sec:control-objectives} require us to enforce both current limiting (necessary) and voltage forming (as much as possible) under grid faults. The objectives cannot be achieved simply using either voltage-forming mode or current-forming mode, since the former requires additional remedies to limit current, while the latter is opposed to voltage forming. To naturally fulfill the control objectives, we propose the concept of cross-forming in Definition~\ref{def:cross-forming}. The cross-forming mode features the \textit{voltage angle forming} and \textit{current magnitude forming} characteristics, therefore preserving the voltage angle forming capability and enabling an inherent current magnitude limiting capability. The capability of the \textit{voltage magnitude forming}, established as under normal operating conditions, is lost in case of current saturation \cite{zhang2024control}. We formulate and prove this finding in Proposition~\ref{prop:voltage-following}.

\begin{proposition}
\label{prop:voltage-following}\textit{Voltage Magnitude Following Under Current Saturation:}
Consider the three-phase balanced\footnote{Proposition~\ref{prop:voltage-following} can be extended to unbalanced cases, where the Thevenin equivalent grid in Fig.~\ref{fig:voltage-following-proof}(a) models a positive-sequence aggregated circuit.} circuit system in Fig.~\ref{fig:voltage-following-proof}(a) and assume that the current magnitude of the inverter is saturated, $\abs{\pha i} = I_{\lim}$, (it can therefore only be a current-forming source or a cross-forming source). Assume that the circuit admits a steady state. The voltage magnitude $\abs{\pha v}$ of the inverter terminal in the steady state is determined by the exogenous conditions, i.e., it depends on the grid voltage $\pha{v}_{\mathrm{g}}$ and cannot be independently specified by the source itself.
\end{proposition}

% =======
% FIG
% =======
\begin{figure}
  \begin{center}
  \includegraphics{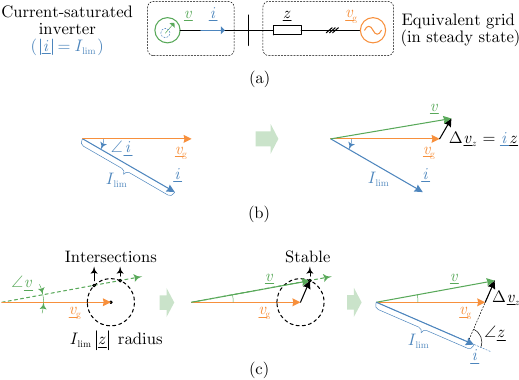}
  \caption{(a) A generic circuit representing a current-saturated inverter ($\abs{\pha i} = I_{\lim}$) connected to a grid. In both (b) and (c), where $\angle{\pha i}$ and $\angle{\pha v}$ are specified, respectively, the voltage magnitude $\abs{\pha v}$ passively follows the circuit law.}
  \label{fig:voltage-following-proof}
  \end{center}
\end{figure}

\begin{proof}
The proof for the case of a current-forming source is trivial. As illustrated in Fig.~\ref{fig:voltage-following-proof}(b), the voltage drop, $\Delta \pha{v}_{z} = \pha v - \pha{v}_{\mathrm{g}}$, on the grid impedance is determined as $\Delta \pha{v}_{z} = \pha {i}\, \pha{z}$. Then, the voltage $\pha v$ is determined by the triangle in Fig.~\ref{fig:voltage-following-proof}(b).

In the case of a cross-forming source, since the current magnitude is specified, $\abs{\Delta \pha{v}_{z}} = \lvert \pha {i}\,\pha{z} \rvert = I_{\lim} \vert \pha z \vert$ is also determined. Consider the voltage equation in the steady state,
\begin{equation}
\label{eq:voltage-equation-law}
    \pha v = \pha{v}_{\mathrm{g}} + \Delta \pha{v}_{z} = \pha{v}_{\mathrm{g}} + \pha i\,\pha z.
\end{equation}
If the circuit admits a steady state, we can see from Fig.~\ref{fig:voltage-following-proof}(c) that the green dashed arrow, indicating the specified $\angle{\pha v}$, must cross the dashed circle, which has a radius $\abs{\Delta \pha{v}_{z}} = I_{\lim} \vert \pha z \vert$ and is centered at the vertex of $\pha{v}_{\mathrm{g}}$. This would imply that there exists $\pha v$ and $\Delta \pha{v}_{z}$ satisfying \eqref{eq:voltage-equation-law}. Then, the direction of $\pha i$ results from $\Delta \pha{v}_{z}$. There may exist two intersections, between which the one with a larger voltage magnitude (corresponding to inductive reactive current provision) will prove to be stable and the other is unstable, depending on a specific control; see Section \ref{sec:cross-forming-controls} later. For either intersection, the voltage $\pha v$ of the source follows from the given grid conditions. Hence, it cannot be independently specified by the source.
\end{proof}

Proposition~\ref{prop:voltage-following} reveals that for a current-saturated voltage-forming inverter (irrespective of the control method used), the terminal voltage magnitude can no longer be controlled independently (this phenomenon is described as ``voltage decline" in \cite{zhang2024control}). In other words, a current-saturated voltage-forming inverter loses the capability to impose the terminal voltage and, instead, exhibits a voltage magnitude following behavior to respect the circuit law. However, this voltage magnitude following behavior does not imply a challenge for grid-forming inverters in fulfilling grid code requirements. This is because maintaining a constant voltage phasor under a grid disturbance is necessary only for the virtual internal voltage and only for a short period before current saturation occurs, as required by grid codes \cite{gb2024,aemo2023voluntary,entso2024}. After current saturation, the focus shifts to supplying reactive or active current rather than maintaining the voltage phasor. Moreover, even when the voltage magnitude forming capability is lost, the voltage angle forming capability can still be preserved. This is crucial to provide an autonomous voltage angle reference for the fault reactive current provision during low-voltage ride-through, phase jump active power provision during phase jump ride-through, or active inertia power provision during frequency drifts, as required in grid codes \cite{gb2024,aemo2023voluntary,entso2024}. In this sense, Proposition~\ref{prop:voltage-following} motivates the notion of cross-forming behavior, as explained in Section~\ref{sec:cross-forming-controls} below.

\subsection{Capabilities to Serve Grid Code Requirements}
\label{sec:gfm-capability}

We now examine the inherent capabilities of the three control modes to provide FRT services. These services include fast fault current injection, phase jump active power delivery, and sink for imbalances---three key requirements commonly specified in grid codes concerning grid fault scenarios \cite{ghimire2024functional}.

\textit{Voltage-Forming Mode:}
This mode can naturally provide fast fault (reactive) current injection and phase jump (active) power as it acts as a voltage-forming source. This is illustrated in Fig.~\ref{fig:fault-services}(a). The reactive current $i_Q$ increases naturally when the grid voltage $\vert \pha{v}_{\mathrm{g}} \vert \downarrow$ dips, providing a fast reactive current to counteract the voltage dip. The active current $i_P$ increases naturally when the grid phase $\angle \pha{v}_{\mathrm{g}} \downarrow$ jumps backward, providing fast active power to counteract the phase jump. Since the voltage-forming mode requires the output current to naturally respond to disturbances, it does not provide any current-limiting capabilities on its own. Numerous remedies have been developed to address this infamous current-limiting challenge; see Appendix~\ref{appendix:current-limiting} for a categorized review and discussions about their limitations.

% =======
% FIG
% =======
\begin{figure}
  \begin{center}
  \includegraphics{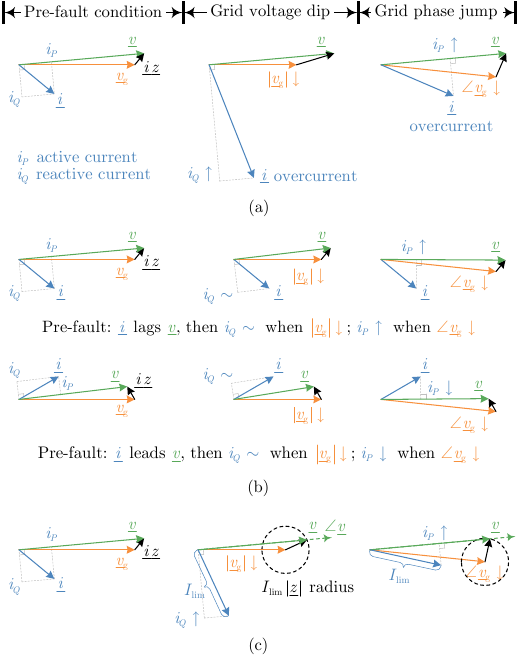}
  \caption{Provisions of fault reactive current $i_Q$ and phase jump active current $i_P$ at the moment of a grid voltage dip and phase jump, respectively, resulting from the circuit equation of $\pha v = \pha{v}_{\mathrm{g}} + \pha i\,\pha z$. (a) Voltage-forming mode ($\pha{v}$ fixed), (b) Current-forming mode ($\pha{i}$ fixed), and (c) Cross-forming mode ($\angle \pha{v}$ fixed and $\lvert \pha{i} \rvert$ limited to $I_{\mathrm{lim}}$).}
  \label{fig:fault-services}
  \end{center}
\end{figure}

\textit{Current-Forming Mode:}
The response of this mode to a voltage dip $\vert \pha{v}_{\mathrm{g}} \vert \downarrow$ or phase jump $\angle \pha{v}_{\mathrm{g}} \downarrow$ is depicted in Fig.~\ref{fig:fault-services}(b). This mode cannot provide any voltage-forming source behaviors independently. When the grid voltage dips, if the current reference is fixed, the reactive current $i_Q$ will remain almost the same as the pre-fault condition (assuming that the current-forming mode is also used before the fault). To provide a required reactive current, one needs to adjust the current reference \cite{xin2021dual,schweizer2022grid}. When the grid phase jumps, the active current $i_P$ may increase or decrease, depending on the pre-fault condition. In other words, the natural response of $i_P$ may not satisfy grid code specifications. Likewise, one needs to adjust the current reference to ensure a compliant active power provision \cite{schweizer2022grid}. On the other hand, the merit of the current-forming mode lies in its inherent current-limiting capability.

\textit{Cross-Forming Mode:}
Since the voltage angle is controlled autonomously in the cross-forming mode, this control mode provides natural reactive current $i_Q$ under a grid voltage dip and active current $i_P$ under a phase jump, as illustrated in Fig.~\ref{fig:fault-services}(c). This occurs without the need for voltage angle detection, as required in PLL-based grid-following control. Hence, the voltage angle forming characteristics facilitate both fault reactive current injection and phase jump active power provision. Meanwhile, the current magnitude forming characteristics inherently enable current limiting. It is noted that the provision of $i_Q$ and $i_P$ is not as large as in the (current-unlimited) voltage-forming mode, since the current magnitude is limited within $I_{\lim}$ in the cross-forming mode. This reflects the inherent capabilities of a current-limited grid-forming inverter to provide FRT services. Concerning response speed, the cross-forming mode is comparable to the voltage-forming mode, as it builds on a similar implementation (see the control design in the next section). Therefore, the cross-forming mode can meet grid code requirements for fast response speed.

Regarding negative-sequence services under asymmetrical grid faults, typical requirements include twice-fundamental-frequency power oscillation suppression \cite{jia2017review} and voltage imbalance mitigation \cite{ieee2800}. Our review and rigorous proof in Appendix~\ref{appendix:negative-sequence-current} indicates that the achievement of negative-sequence services is not difficult, and the control of negative-sequence components can be readily implemented in any of the control architectures, remaining flexible for all three control modes.

\section{Cross-Forming Control}
\label{sec:cross-forming-controls}

This section aims to develop simple and feasible control designs for the cross-forming mode. We first consider three-phase balanced conditions and then extend the results to unbalanced conditions. Our control designs will also apply to single-phase inverters (in which orthogonal signal generators are required for controlling single-phase signals in a two-axis reference frame \cite{rezazadeh2023single}).

A generic block diagram of the cross-forming control architecture is shown in Fig.~\ref{fig:desired-equivalent-circuit}(a), which comprises four modules. The standard voltage-forming reference module aims to provide a voltage-forming reference $\pha {\hat v}$, e.g., droop control or dVOC. The current limiter and current tracking control modules are also standard modules. The cross-forming regulator is a novel and crucial module, which takes the voltage reference $\pha {\hat v}$ to generate a current reference $\pha {\hat i}$. The cross-forming regulation is activated only during current saturation. This regulator operates as a standard virtual admittance control both before current saturation and after saturation recovery, maintaining the original voltage-forming mode and thus also meeting grid code requirements for current-unsaturated grid-forming operation. We first present a desired equivalent circuit of cross-forming behaviors and then use it to design the control implementations of the cross-forming regulator.

\subsection{Desired Equivalent Circuit of Cross-Forming Mode}

A desired equivalent circuit of the cross-forming control is shown in Fig.~\ref{fig:desired-equivalent-circuit}(b), where the output current has a specified magnitude $I_{\lim}$, and the virtual internal voltage has a specified angle $\angle \pha {\hat v}$. Moreover, a (constant) virtual impedance $\pha z_{\mathrm{v}}$ can be added to enhance control flexibility and performance, the existence of which aligns with the grid code specifications for grid-forming inverters \cite{gb2024,aemo2023voluntary,entso2024}. Given that the magnitude of the virtual internal voltage cannot be independently specified by the source itself (cf. Proposition~\ref{prop:voltage-following}) and that the virtual internal voltage is required to have the same angle as $\angle \pha{\hat v}$ to preserve voltage angle forming, we denote the virtual internal voltage of the desired circuit in Fig.~\ref{fig:desired-equivalent-circuit}(b) as
\begin{equation}
\label{eq:virtual-internal-voltage}
    \pha{\hat v}_{\lambda} \coloneqq \lambda \pha{\hat v},
\end{equation}
where $\lambda \in \mathbb{R}_{>0}$ denotes the ratio between the virtual internal voltage $\pha{\hat v}_{\lambda}$ and the reference voltage $\pha {\hat v}$ (given by the voltage-forming reference module). The voltage equation of the desired equivalent circuit in Fig.~\ref{fig:desired-equivalent-circuit}(b) is thus given as
\begin{equation}
\label{eq:desired-voltage-equation}
    \pha{\hat v}_{\lambda} = \lambda \pha{\hat v} = \pha v + \pha z_{\mathrm{v}} \pha i, \quad \vert \pha i \vert = I_{\lim}.
\end{equation}
Since the voltage reference $\pha{\hat v}$ and the current magnitude $I_{\lim}$ are independent given variables, the voltage angle forming is preserved in the virtual internal voltage, and the current-limiting behavior is enforced. In contrast to these independent variables, the internal voltage magnitude $\lvert \pha{\hat v}_{\lambda} \rvert$ (and thus $\lambda$) as well as the resulting current angle $\angle {\pha{\hat i}}$ are dependent variables, i.e., control outputs. These dependent variables passively follow the circuit law, as revealed in Proposition~\ref{prop:voltage-following}. The cross-forming signal causality is illustrated in Fig.~\ref{eq:desired-voltage-equation}(c).

% =======
% FIG
% =======
\begin{figure}
  \begin{center}
  \includegraphics{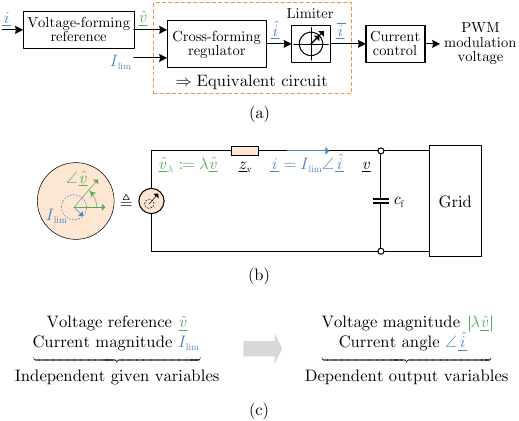}
  \caption{(a) Cross-forming control architecture, (b) desired equivalent circuit, and (c) cross-forming signal causality for inverters to perform voltage angle forming and current magnitude forming behaviors under current saturation.}
  \label{fig:desired-equivalent-circuit}
  \end{center}
\end{figure}

\begin{remark}
\textit{Differences From Prior Equivalent Circuits:}
We highlight that the desired equivalent circuit proposed in this work differs from previous results in the literature. Typically, prior current-limiting strategies use an adaptive virtual impedance or a current limiter with a virtual admittance. These strategies lead to a variable (specifically, current-dependent) virtual impedance in their equivalent circuits, as noted in Appendix~\ref{appendix:current-limiting}. In contrast, the virtual impedance resulting from our cross-forming control remains constant, while the virtual internal voltage magnitude becomes current-dependent, following the exogenous circuit's response to adapt to grid faults of different severity. The constant virtual impedance enables the formation of a constant equivalent power network, allowing us to recover an \textit{equivalent normal/canonical form} similar to the normal form of voltage-forming systems. Hence, aside from its charming simplicity with a constant impedance, our equivalent circuit facilitates transient stability analysis, since we can extend the existing transient stability results from current-unsaturated to current-saturated conditions, a benefit that remains unestablished with previous current-limiting strategies.
\end{remark}

\subsection{Cross-Forming Control Designs}

We propose two simple, viable, and robust feedback control implementations of the cross-forming regulator to ensure that inverters fulfill the desired voltage equation \eqref{eq:desired-voltage-equation}. As indicated in the desired voltage equation, the cross-forming regulator aims to match the angle of the virtual internal voltage angle with the angle of the voltage reference, i.e., $\angle \pha{\hat v}_{\lambda} = \angle \pha{\hat v}$, and generate a constant virtual impedance $\pha{z}_{\mathrm{v}}$. The cross-forming regulator thus differs from the classical inner voltage-tracking controller. To facilitate the design of the cross-forming regulator, we rewrite the desired circuit equation in \eqref{eq:desired-voltage-equation} as
\begin{equation}
\label{eq:desired-voltage-equation-re1}
    \pha { i} = \frac{1}{\pha z_{\mathrm{v}}} \left(\pha{\hat v}_{\lambda} - \pha v \right).
\end{equation}
We further define the degree of saturation (DoS, also known as current-limiting factor \cite{sadeghkhani2017current}) as the ratio between the saturated current reference $\pha {\bar{i}}$ and the unsaturated current reference $\pha{\hat i}$, i.e.,
\begin{gather}
\label{eq:dos-definition}
    \mathrm{DoS\, (balanced)\text{:}}\ \mu \coloneqq \frac{\, \pha {\bar{i}}\, }{\pha {\hat{i}}} \overset{\mathrm{(saturated)}}{=} \frac{I_{\lim}}{\lvert \pha{\hat i} \rvert},
\end{gather}
where $\mu \in (0,\, 1] \subset \mathbb{R}$ and $\pha {\bar{i}}$ and $\pha {\hat{i}}$ have the same angle. Based on the definition of $\mu$ and assuming that the output current $\pha i$ tracks the saturated current reference $\pha {\bar{i}}$, i.e., $\pha i = \pha {\bar{i}}$, the circuit equation in \eqref{eq:desired-voltage-equation-re1} is reformulated as
\begin{equation}
\label{eq:desired-voltage-equation-re2}
    \pha {\hat i} = \frac{1}{\pha z_{\mathrm{v}}} \left(\frac{\pha{\hat v}_{\lambda}}{\mu} - \frac{\pha v}{\mu} \right).
\end{equation}
Based on \eqref{eq:desired-voltage-equation-re1} or \eqref{eq:desired-voltage-equation-re2}, there are two distinct implementations of the cross-forming regulator. Both implementations have different pros and cons, and their relative merits will be discussed later in Remark~\ref{rem:comparison}.

% =======
% FIG
% =======
\begin{figure}
  \begin{center}
  \includegraphics{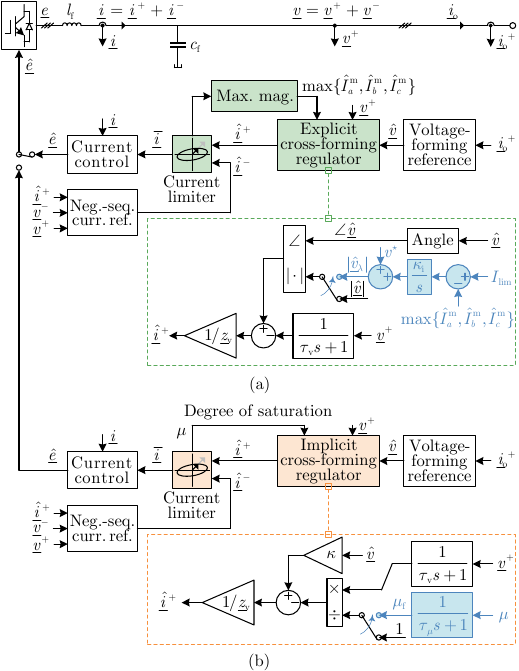}
  \caption{Proposed cross-forming implementations for grid-forming inverters operating against grid faults. (a) Explicit implementation in \eqref{eq:explicit-extended}. (b) Implicit implementation in \eqref{eq:implicit-extended}. With the incorporation of current information feedback (highlighted in blue color), cross-forming control distinguishes itself from classical virtual admittance control while remaining backward compatible, enabling simple switching between the two control modes.}
  \label{fig:cross-forming-controls}
  \end{center}
\end{figure}

\subsubsection{Explicit Cross-Forming Implementation}

In the first implementation, based on \eqref{eq:desired-voltage-equation-re1}, we first aim to design a suitable virtual internal voltage magnitude reference $\lvert \pha{\hat v}_{\lambda} \rvert$ so that the reference current magnitude $\lvert \pha{\hat i} \rvert$ (and thus the current magnitude $\lvert \pha{i} \rvert$) tracks the prescribed limit $I_{\lim}$. As given in \eqref{eq:explicit-cross-forming-a} below, we use an integral regulator to enforce the current limit and generate $\lvert \pha{\hat v}_{\lambda} \rvert$. Furthermore, we use the angle reference $\angle \pha{\hat v}$ given by the voltage-forming reference module to generate the virtual internal voltage vector reference as $\pha{\hat v}_{\lambda} = \lvert \pha{\hat v}_{\lambda} \rvert \angle \pha{\hat v}$, such that the angle forming characteristic is preserved. Based on this, the desired voltage equation in \eqref{eq:desired-voltage-equation-re1} is fulfilled using the virtual admittance control in \eqref{eq:explicit-cross-forming-b},
\begin{subequations}
\label{eq:explicit-cross-forming}
\begin{align}
\label{eq:explicit-cross-forming-a}
    \lvert \pha{\hat v}_{\lambda} \rvert &= v^{\star} + \frac{\kappa_{\mathrm{i}}}{s} \bigl(I_{\lim} - \lvert \pha{\hat i} \rvert \bigr),\\
\label{eq:explicit-cross-forming-b}
    \pha {\hat{i}} &= \frac{1}{\pha z_{\mathrm{v}}}\left(\lvert \pha{\hat v}_{\lambda} \rvert \angle \pha{\hat v} - \pha v \right),
\end{align}    
\end{subequations}
where $v^{\star}$ denotes the initial value of the voltage magnitude $\lvert \pha{\hat v}_{\lambda} \rvert$. The integrator aims to track the given current limit when overcurrent occurs, which is achieved by explicitly reducing the virtual internal voltage magnitude $\pha{\hat v}_{\lambda}$ whenever the current reference exceeds the limit. In other words, the current limiting is accomplished by dropping the voltage level without altering the original angle reference. In a current-saturated steady state, the current settles at the limit, and the virtual internal voltage magnitude settles at the value determined by the exogenous grid conditions. The control diagram is shown in Fig.~\ref{fig:cross-forming-controls}(a) (already extended to unbalanced cases, shown later). It is noted that current limiting is mainly tackled by the limit tracking in \eqref{eq:explicit-cross-forming-a}. The current limiter is not mandatory, but it can still be enabled to limit fast overcurrent transients. Moreover, if only the current limiter is used for current limiting, we would need additional feedback control to achieve the desired equivalent circuit, as will be shown in the implicit design below.

We note that the explicit cross-forming control allows for alternatives to \eqref{eq:explicit-cross-forming}. For instance, one can use more advanced controls, e.g., nonlinear or predictive control, to replace the linear integral control in \eqref{eq:explicit-cross-forming-a}. Alternatively, one can directly use the measured current $\abs{\pha i}$ as the feedback in \eqref{eq:explicit-cross-forming-a}, bypass \eqref{eq:explicit-cross-forming-b}, and then treat the voltage vector reference as the modulation voltage without using a current limiter and a current controller. In doing so, it is possible to tune the regulator in \eqref{eq:explicit-cross-forming-a} to be sufficiently fast to suppress overcurrents while preserving good small-signal stability performance \cite{liu2022current}. This is analogous to the classical single-loop voltage-magnitude control presented in \cite{liu2022current}, where, in contrast, a proportional current feedback control is employed to emulate a virtual resistor to perform current limiting.

\subsubsection{Implicit Cross-Forming Implementation}
The second implementation is based on \eqref{eq:desired-voltage-equation-re2}. Since $\vert \pha{\hat v}_{\lambda} \vert$ is an unknown dependent variable, \eqref{eq:desired-voltage-equation-re2} cannot be implemented directly. By substituting the virtual (unknown) variable ${\pha{\hat v}_{\lambda}}/{\mu}$ in \eqref{eq:desired-voltage-equation-re2} with the given voltage vector reference $\pha{\hat v}$ and a control gain $\kappa$, we present an implicit cross-forming regulator as
\begin{subequations}
\label{eq:implicit-cross-forming}
\begin{align}
\label{eq:implicit-cross-forming-a}
    \mathrm{DoS\, (balanced)\text{:}}\ \mu \overset{\mathrm{(saturated)}}&{=} \frac{I_{\lim}}{\lvert \pha{\hat i} \rvert},\\
\label{eq:implicit-cross-forming-b}
    \pha {\hat{i}} &= \frac{1}{\pha z_{\mathrm{v}}}\left(\kappa \pha{\hat v} - \frac{\pha v}{\mu} \right),
\end{align}
\end{subequations}
where $\kappa \in \mathbb{R}_{>0}$ is an independent control gain. Using the implicit cross-forming regulator in \eqref{eq:implicit-cross-forming}, we arrive at the following dependence of the virtual internal voltage in the desired circuit,
\begin{equation}
\label{eq:internal-voltage}
    \pha{\hat v}_{\lambda} = \kappa \mu \pha{\hat {v}} = \pha z_{\mathrm{v}} \pha {{i}} + {\pha v},\quad \Rightarrow \quad \lambda = \kappa \mu,
\end{equation}
which suggests that the internal voltage $\pha{\hat v}_{\lambda}$ preserves the same angle as the voltage reference $\pha{\hat {v}}$, thus preserving the angle forming behavior. In addition, since the magnitude of the virtual internal voltage is governed by the DoS $\mu$ indirectly, it is implicitly made to follow the circuit law. For example, if we choose $\kappa = 1$ simply, it will follow from \eqref{eq:internal-voltage} that $\lambda = \mu$, which implies that $\mu$ will serve the same role as $\lambda$ in the virtual internal voltage $\pha{\hat v}_{\lambda}$ in \eqref{eq:virtual-internal-voltage}.

% =======
% FIG
% =======
\begin{figure}
  \begin{center}
  \includegraphics{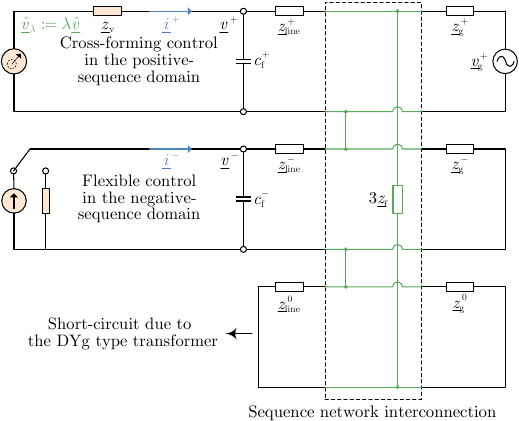}
  \caption{Sequence-domain equivalent circuit for cross-forming control under an asymmetrical grid fault (a single line-to-ground fault is exemplified \cite{he2022synchronization}).}
  \label{fig:equivalent-circuit-sequence}
  \end{center}
\end{figure}

The diagram of the implicit control is shown in Fig.~\ref{fig:cross-forming-controls}(b) (already extended to unbalanced cases). The proportional-like feedback control in \eqref{eq:implicit-cross-forming-b} subsumes a classical virtual admittance voltage control as given by \eqref{eq:virtual-admittance} in Appendix~\ref{appendix:standard-gfm}, with the addition of the control gain $\kappa$ and the auxiliary factor $1/\mu$ scaling up the voltage feedback. The feedback DoS signal $\mu$ can be easily extracted from the current limiter. A low-pass filter (LPF) is added to the feedback of $\mu$ to make the cross-forming regulating response slower than the current control, and to avoid the appearance of an algebraic loop associated with $\mu$. Moreover, an LPF (for implementation in $dq$ coordinates) or a band-pass filter (for implementation in $\alpha\beta$ coordinates) is added to the voltage feedback to enhance the small-signal stability of the virtual admittance control \cite{wu2022smallsignal}, particularly in the case of a large $X/R$ ratio in $\pha z_{\mathrm{v}}$.

We remark that the implicit cross-forming control also allows for alternatives to \eqref{eq:implicit-cross-forming}. For example, one alternative is $\pha {\hat{i}} = \kappa \frac{\pha{\hat v}}{{\pha v}/\pha i + \pha z_{\mathrm{v}}}$,
which is approximately equivalent to \eqref{eq:implicit-cross-forming}. However, it is not compatible with the standard virtual admittance control architecture.

\subsubsection{Extensions to Unbalanced Conditions}
Both implementations of cross-forming controls apply to not only balanced but also unbalanced conditions, as shown in Fig.~\ref{fig:cross-forming-controls}(a) and (b). Particularly, regarding asymmetrical grid faults, the desired equivalent circuit in \eqref{eq:desired-voltage-equation} refers to the positive-sequence-domain circuit. Therefore, the variables involved in the above design refer to corresponding positive-sequence components. Specifically, concerning asymmetrical grid faults, the explicit cross-forming regulator in \eqref{eq:explicit-cross-forming} extends to
\begin{subequations}
\label{eq:explicit-extended}
\begin{align}
\label{eq:explicit-extended-a}
    \lvert \pha{\hat v}_{\lambda} \rvert &= v^{\star} + \frac{\kappa_{\mathrm{i}}}{s} \bigl(I_{\lim} - \max \{\hat{I}_{a}^{\mathrm{m}},\hat{I}_{b}^{\mathrm{m}},\hat{I}_{c}^{\mathrm{m}} \} \bigr),\\
\label{eq:explicit-extended-b}
    \pha {\hat{i}}^{+} &= \frac{1}{\pha z_{\mathrm{v}}}\left(\lvert \pha{\hat v}_{\lambda} \rvert \angle \pha{\hat v} - \pha v^{+} \right),
\end{align}
\end{subequations}
with phase current magnitude references $\hat{I}_{x}^{\mathrm{m}}$, $x \in \{a,b,c\}$. We refer the reader to \eqref{eq:phase-current-mag} or \eqref{eq:phase-current-mag-dq} in Appendix~\ref{appendix:current-limiting} for how to calculate $\hat{I}_{x}^{\mathrm{m}}$ with the positive- and negative-sequence current references in $\alpha\beta$ or $dq$ coordinates, respectively. Moreover, the implicit cross-forming regulator in \eqref{eq:implicit-cross-forming} extends to
\begin{subequations}
\label{eq:implicit-extended}
\begin{align}
\label{eq:implicit-extended-a}
    \mathrm{DoS\, (unbalanced)\text{:}}\ \mu \overset{\mathrm{(saturated)}}&{=} \frac{I_{\lim}}{\max \{\hat{I}_{a}^{\mathrm{m}},\hat{I}_{b}^{\mathrm{m}},\hat{I}_{c}^{\mathrm{m}} \}},\\
\label{eq:implicit-extended-b}
    \pha {\hat{i}}^{+} &= \frac{1}{\pha z_{\mathrm{v}}}\left(\kappa \pha{\hat v} - \frac{\pha v^{+}}{\mu} \right),
\end{align}
\end{subequations}
where the DoS variable is extended from \eqref{eq:implicit-cross-forming-a} to \eqref{eq:implicit-extended-a}, with the phase current magnitude references extracted from the current limiter; see \eqref{eq:phase-current-mag} and \eqref{eq:phase-current-mag-dq}.

Consider a single line-to-ground fault as an example. The sequence-domain equivalent circuit is shown in Fig.~\ref{fig:equivalent-circuit-sequence}, which refines the generic circuit in Fig.~\ref{fig:phase-sequence-circuit}(b). In the positive-sequence domain, the cross-forming control creates a cross-forming source (with voltage angle forming and voltage magnitude following behaviors). Therefore, it provides the required fault reactive current and phase jump active power in the positive-sequence domain. In the negative-sequence domain, multiple flexible control options are available, as detailed in Appendix~\ref{appendix:negative-sequence-current}. When an asymmetrical fault occurs, the cross-forming regulator in \eqref{eq:explicit-extended} or \eqref{eq:implicit-extended} reduces the voltage magnitude of the positive-sequence cross-forming source, $\lvert \pha{\hat v}_{\lambda}\rvert$ or $\kappa \mu \lvert \pha{\hat v} \rvert$, ensuring that the maximum phase current magnitude, $\max \{{I}_{a}^{\mathrm{m}},{I}_{b}^{\mathrm{m}},{I}_{c}^{\mathrm{m}} \}$, remains at the specified limit.

% =======
% FIG
% =======
\begin{figure*}
  \begin{center}
  \includegraphics{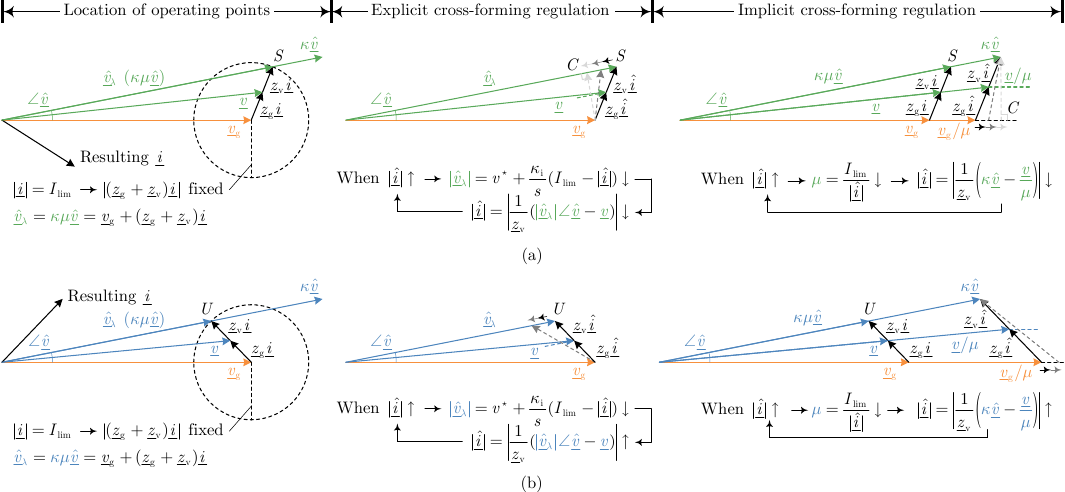}
  \caption{Stability analysis of operating points resulting from cross-forming regulation. (a) The operating point $S$ is stable due to negative feedback regulation. (b) The operating point $U$ is unstable due to positive feedback regulation.}
  \label{fig:stability-analysis}
  \end{center}
\end{figure*}

\subsubsection{Parameter Tuning}
\label{sec:parameter-tuning}
The parameter tuning of both cross-forming regulators is simple. First, the virtual impedance $\pha z_{\mathrm{v}}$ should be tuned such that the lumped impedance voltage drop $\abs{\Delta \pha{v}_{z}}$ is large enough to ensure the existence of feasible operating points during current saturation (particularly in response to phase jumps), as shown in Fig.~\ref{fig:voltage-following-proof}(c). However, it should not be too large, to avoid causing a high virtual voltage drop and reactive power consumption. Furthermore, the choice of the $X/R$ impacts both small-signal and transient stability \cite{taoufik2022variable}. Next, the bandwidth of both cross-forming regulators should be tuned to be slower than the inner current loop. For the explicit one in \eqref{eq:explicit-cross-forming}, the integral gain $\kappa_{\mathrm{i}}$ yields a time constant $\lvert \pha{z}_{\mathrm{v}} \rvert / \kappa_{\mathrm{i}}$. For the implicit one, its bandwidth is determined by the time constant $\tau_{\mu}$. The selection of the LPF time constants $\tau_{\mathrm{\mu}}$ and $\tau_{\mathrm{v}}$ should ensure a clear time-scale separation between the cross-forming regulator and the inner current loop. We recommend performing a quantitative small-signal analysis across typical operating points in tuning these parameters to guarantee a well-damped dynamic response. Ill-tuned parameters may lead to suboptimal performance or even instability \cite{wu2022smallsignal,taoufik2022variable}. Finally, the initial voltage magnitude $v^{\star}$ in \eqref{eq:explicit-cross-forming-a} can take the nominal value simply. The feedforward gain $\kappa \geq 1$ is preferable to render $\mu < 1$, implying that the inverter stays at the current limit and thus fully utilizes the overcurrent capability. The value of $\kappa$ does not affect the system's steady-state operating point provided that the current is saturated.

\subsubsection{Stability of Cross-Forming Regulators}
As can be seen in Fig.~\ref{fig:cross-forming-controls}, a feedback control loop is introduced in both cross-forming implementations. The stable operation of both regulators is based on negative feedback of this control loop. To show this, we consider a cross-forming controlled inverter, which is connected to a Thevenin equivalent grid with impedance $\pha z_{\mathrm{g}}$ and voltage $\pha{v}_{\mathrm{g}}$. It follows from \eqref{eq:internal-voltage} that, in the steady state,
\begin{equation}
    \pha{\hat v}_{\lambda} = \kappa \mu \pha{\hat {v}} = \pha z_{\mathrm{v}} \pha {{i}} + {\pha v} = (\pha z_{\mathrm{v}} + \pha z_{\mathrm{g}}) \pha {{i}} + \pha{v}_{\mathrm{g}}.
\end{equation}
Given the voltage reference $\pha{\hat v}$ and the current magnitude limit $I_{\lim}$, there exist two possible operating points, c.f. Fig.~\ref{fig:voltage-following-proof}(c). Fig.~\ref{fig:stability-analysis} details the location of the two operating points, i.e., the two intersection points $S$ and $U$. Typically, the operating point $S$ is stable, while the other is unstable. The stability around the operating point can be analyzed qualitatively by investigating the response to a disturbance leading to an increase in $\lvert \pha {\hat{i}} \rvert$, which is illustrated in Fig.~\ref{fig:stability-analysis}.
\begin{itemize}
    \item For explicit cross-forming regulation: An increase in $\lvert \pha {\hat{i}} \rvert$ will cause $\lvert \pha{\hat v}_{\lambda} \rvert = v^{\star} + \frac{\kappa_{\mathrm{i}}}{s} \bigl(I_{\lim} - \lvert \pha{\hat i} \rvert \bigr)$ in \eqref{eq:explicit-cross-forming-a} to decrease. For this case, the regulation in \eqref{eq:explicit-cross-forming-b} should be such that $\lvert \pha {\hat{i}} \rvert$ decreases. This is only possible if the projection of the vector ${\pha v}_{\mathrm{g}}$ onto the direction of $\pha{\hat v}_{\lambda}$ is smaller than $\lvert \pha{\hat v}_{\lambda} \rvert$, as shown in Fig.~\ref{fig:stability-analysis}(a), since $\lvert \pha {\hat{i}} \rvert = \bigl \lvert \frac{1}{\pha z_{\mathrm{v}}} \bigr\rvert \bigl \lvert \lvert \pha{\hat v}_{\lambda} \rvert \angle \pha{\hat v} - \pha v \bigr \rvert = \bigl \lvert \frac{1}{\pha z_{\mathrm{v}} + \pha z_{\mathrm{g}}} \bigr \rvert \lvert \pha{v}_{\mathrm{g}} - \pha v \rvert$. As $\lvert \pha{\hat v}_{\lambda} \rvert$ decreases, the critical case appears when the operating point reaches $C$, where the projection equals $\lvert \pha{\hat v}_{\lambda} \rvert$. The operating point $U$ in Fig.~\ref{fig:stability-analysis}(b) is unstable since the projection is larger than $\lvert \pha{\hat v}_{\lambda} \rvert$, leading to positive feedback regulation.
    \item For implicit cross-forming regulation: An increase in $\lvert \pha {\hat{i}} \rvert$ will cause $\mu = \frac{I_{\lim}}{\lvert \pha{\hat i} \rvert}$ in \eqref{eq:implicit-cross-forming-a} to decrease. We require that the feedback regulation in \eqref{eq:implicit-cross-forming-b} be such that the decrease in $\mu$ would reduce $\lvert \pha {\hat{i}} \rvert$. Since $\lvert \pha {\hat{i}} \rvert = \bigl \lvert \frac{1}{\pha z_{\mathrm{v}}} \bigr \rvert \bigl \lvert \kappa \pha{\hat v} - \frac{\pha v}{\mu} \bigr \rvert = \bigl \lvert \frac{1}{\pha z_{\mathrm{v}} + \pha z_{\mathrm{g}}} \bigr \rvert \bigl \lvert \kappa \pha{\hat v} - \frac{\pha v_{\mathrm{g}}}{\mu} \bigr \rvert$, this is only possible if the projection of the vector $\kappa \pha{\hat v} $ onto the direction of ${\pha v_{\mathrm{g}}}$ is greater than $\bigl \lvert \frac{\pha v_{\mathrm{g}}}{\mu} \bigr \rvert$. Likewise, the critical case appears at $C$ in Fig.~\ref{fig:stability-analysis}(a), and the operating point $U$ in Fig.~\ref{fig:stability-analysis}(b) is unstable due to positive feedback regulation.
\end{itemize}
Moreover, corresponding to the stable operating point $S$, the resulting output current vector lags the internal voltage vector, implying reactive current injection and therefore satisfying grid code requirements. As indicated in Fig.~\ref{fig:stability-analysis}(a), the provision of reactive current stems from the preserved voltage angle forming $\angle \pha{\hat v}$ and the inductive property of $\pha{z}_{\mathrm{v}}$ and $\pha{z}_{\mathrm{g}}$. This suggests that the voltage angle forming matters for not only active power regulation but also fault reactive current injection.

We note that if $\angle \pha{\hat v}$ is large or $\lvert (\pha z_{\mathrm{v}} + \pha z_{\mathrm{g}}) \pha {{i}} \rvert$) is small to the point where both $S$ and $U$ are located in the left half plane of the dashed circle, both operating points will be unstable for the implicit cross-forming regulator. If $\lvert (\pha z_{\mathrm{v}} + \pha z_{\mathrm{g}}) \pha {{i}} \rvert$ is so small that $\angle \pha{\hat v}$ is beyond the tangent direction of the dashed circle, there will be no feasible operating points. Such situations should be avoided. Our future research will explore how to ensure the existence of operating points under certain parametric conditions and quantify their stability region.

\begin{remark}
\label{rem:recovery}
\textit{Switching Between Voltage-Forming and Cross-Forming Modes:}
When a grid-forming inverter is free of current saturation, it can maintain the voltage-forming mode. The cross-forming mode is activated only when grid faults occur and lead the current to saturate. The control should switch back to the original voltage-forming mode after grid fault recovery or when the current reference exits saturation. Since the voltage-forming control architecture with virtual admittance control is already subsumed in both cross-forming implementations, the switching between the two control modes is straightforward. For example, switching from the explicit cross-forming regulator to the virtual admittance control can be accomplished by disabling the integrator in \eqref{eq:explicit-cross-forming-a} and replacing $v^{\ast}$ with the original reference $\lvert \pha{\hat v} \rvert$. Setting $\kappa = 1$ the implicit cross-forming regulator in \eqref{eq:implicit-cross-forming} reduces to the virtual admittance control in case of current unsaturation, as $\pha {\hat{i}} = \pha {{i}}$ implies $\mu = 1$. In case of mode switching failures, the control will either remain in the voltage-forming mode or stay in the cross-forming mode, where the inverter remains under control and the current is limited. However, stability and performance may not meet expectations. A complete state machine should be in place to properly deal with the switching between different control modes and the subsequent operation in case of switching failures. Moreover, reliable grid fault detection can help improve the robustness of mode transitions.
\end{remark}

\begin{remark}
\label{rem:comparison}
\textit{Comparison and Selection Between Both Implementations:} Both the explicit and implicit cross-forming implementations exhibit similarly superior dynamic performance (which will be validated later), provided that their parameters are well-tuned. However, each has distinct advantages and disadvantages in the following aspects, explaining why we develop two different types of implementations.
\begin{itemize}
    \item The implementation of the implicit cross-forming regulator is closer to the original virtual admittance control with $\mu = 1$ by default.
    \item The tuning of the implicit cross-forming regulator is simpler, where $\kappa = 1$ is a straightforward choice.
    \item The implicit implementation can act faster since no additional control dynamics are introduced, apart from the LPF. By contrast, the explicit one has limited control bandwidth due to involving additional integral dynamics in tracking the current limit.
    \item For the implicit cross-forming regulator, there exist cases where both operating points are unstable. For the explicit one, provided that operating points exist, the one with the larger voltage must be stable.
\end{itemize}
\end{remark}

% =======
% FIG
% =======
\begin{figure}
  \begin{center}
  \includegraphics{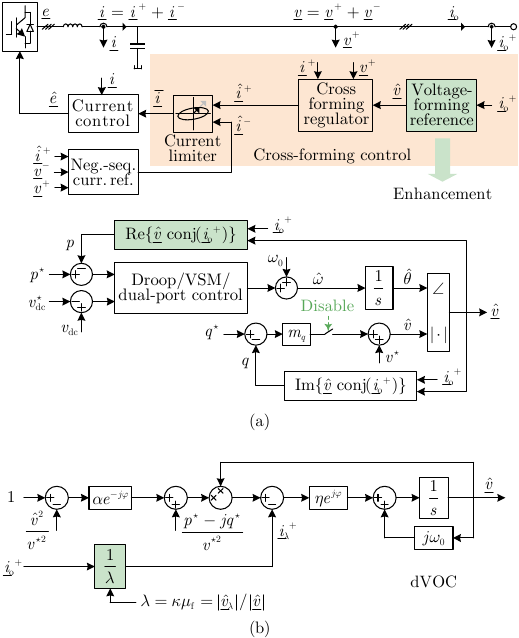}
  \caption{Enhanced voltage-forming reference module embedded in the cross-forming control, with (a) droop control, VSM, AC-DC dual-port control, etc. in polar coordinates and (b) dVOC in $\alpha\beta$ coordinates.}
  \label{fig:enhanced-voltage-forming}
  \end{center}
\end{figure}

\subsection{Enhanced Voltage-Forming References}

The voltage-forming reference is part of the cross-forming control, as noted in Fig.~\ref{fig:enhanced-voltage-forming}. In this section, we enhance the voltage-forming reference module to improve the transient stability of the cross-forming control. Based on that, we show later in Section~\ref{sec:transient-stability} that the proposed enhancement renders the closed-loop system structurally conform to the normal/canonical form of the standard voltage-forming system. Thus, we can extend the existing transient stability results in the literature to current-saturated conditions.

Three types of voltage-forming controls have been reviewed in Appendix~\ref{appendix:standard-gfm}: 1) single-input single-output linear controls, e.g., droop control and VSM; 2) multivariable nonlinear controls, e.g., complex droop control and dVOC; and 3) AC-DC dual-port controls. The enhancement of these voltage-forming controls is based on the coordinates they employ. The most typical coordinates are the $\hat v$ and $\hat \theta$ polar coordinates, and the other typical ones are $\hat v_{\alpha}$ and $\hat v_{\beta}$ rectangular coordinates, as used by the recently prevalent dVOC. We show different enhancements concerning different coordinates of voltage-forming controls.

\begin{table*}
\centering
\caption{Benchmarking of the Proposed Cross-Forming Strategies Against Existing Current-Limiting Strategies.}
\arrayrulecolor{black!10}
\begin{tabular}{llll}
\arrayrulecolor{black}
\hline\hline
& Voltage-forming strategies & Current-forming strategies & Cross-forming strategies \\
\hline
\makecell[l]{Current \\ limiting \\ strategies} & 
\makecell[l]{
\tabitem Type-A: Adaptive virtual impedance \\
\tabitemm \cite{paquette2015virtual,qoria2020current} \\ 
\tabitem Type-B: Current limiter with virtual \\
\tabitemm  admittance \cite{rosso2021implementation,kkuni2024effects,fan2022equivalent,zhang2023current,saffar2023impacts}} & 
\makecell[l]{
 Current reference direct \\ 
 specification \cite{huang2019transient,rokrok2022transient,liu2023dynamic,li2023transient} \\ or in a drooped way \cite{xin2021dual,schweizer2022grid}} & 
\makecell[l]{
\tabitem Explicit cross-forming  \\ 
\tabitem Implicit cross-forming} \\
\arrayrulecolor{black!10}\hline
\makecell[l]{Forming behaviors} & 
\makecell[l]{
\tabitem Internal voltage forming \\ 
\tabitem Current magnitude constrained} & 
\makecell[l]{
\tabitem Current forming \\ 
\tabitem Voltage following (undesired)} & 
\makecell[l]{
\tabitem Voltage angle forming \\ 
\tabitem Current magnitude forming} \\
\hline
\makecell[l]{Control switching} & No switching needed & Yes and incompatible & Yes but {\color{color_cross_forming}compatible} \\
\hline
\makecell[l]{Tuning complexity} & \makecell[l]{\tabitem {\color{color_voltage_forming}Complicated if using Type-A} \\ \tabitem Simple if using Type-B} & \makecell[l]{{\color{color_current_forming}Simple for current limiting, but} \\ {\color{color_current_forming}complicated for FRT services}} & {\color{color_cross_forming}Simple} \\
\hline
\makecell[l]{Current-limiting \\ speed} & 
\makecell[l]{\tabitem {\color{color_voltage_forming}Slow if using Type-A} \\ \tabitem Fast if using Type-B} & Fast & {\color{color_cross_forming}Fast} \\
\hline
\makecell[l]{Overcurrent \\ utilization} & \makecell[l]{\tabitem {\color{color_voltage_forming}No, for Type-A, limited within $[I_{\mathrm{th}},\, I_{\mathrm{lim}}]$} \\ \tabitem Yes, for Type-B, limited at $I_{\mathrm{lim}}$} & Yes, limited at $I_{\mathrm{lim}}$ & {\color{color_cross_forming}Yes, limited at $I_{\mathrm{lim}}$} \\
\hline
\makecell[l]{FRT $i_Q$ provision} & 
\makecell[l]{ $i_Q$ naturally provided, with high priority \\ by reducing $p^{\star}$} & 
\makecell[l]{\tabitem Adjust $i_Q^{\star}$ if using PLL, but slow \\ \tabitem Reduce $p^{\star}$ if using frequency droop \\ \tabitemm \cite{huang2019transient}, but $i_Q$ provision may be slow} & 
\makecell[l]{ $i_Q$ naturally provided, with \\ high priority by reducing $p^{\star}$} \\ 
\hline
\makecell[l]{Phase jump $i_P$ \\ provision} & $i_P$ naturally provided & \makecell[l]{$i_P$ may not be compliant due \\ to behaving as a current source} & $i_P$ naturally provided \\
\hline
\makecell[l]{Resulting impedance} & {\color{color_voltage_forming}Current-dependent for both Type-A and -B} & Optional & {\color{color_cross_forming}Constant} \\
\hline
\makecell[l]{Transient stability \\ enhancements } & 
\makecell[l]{ Numerous strategies available, e.g., \\ 
\tabitem Virtual power feedback \cite{kkuni2024effects} \\ 
\tabitem Alternating virtual inertia \cite{alipoor2014power} \\ 
\tabitem Mode-adaptive control \cite{wu2020mode} \\
\tabitem Power setpoint adjusting \cite{taul2020current} \\
\tabitem Transient active power control \cite{wang2024active}} & 
\makecell[l]{\tabitem Q-axis voltage feedback \cite{huang2019transient} \\ 
\tabitem Current reference angle adjusting \\ 
\tabitemm \cite{rokrok2022transient,liu2023dynamic,li2023transient}, but may conflict with \\ 
\tabitemm FRT services requirements } & 
\makecell[l]{
\tabitem Enhanced voltage-forming \\ \tabitemm references} \\
\hline
\makecell[l]{Transient stability \\ analysis} & 
\makecell[l]{{\color{color_voltage_forming} Difficult since the resulting virtual impedance} \\ {\color{color_voltage_forming} is current-dependent} } & \makecell[l]{{\color{color_current_forming}Relatively difficult due to involving} \\ {\color{color_current_forming}control architecture switching}} & 
\makecell[l]{{\color{color_cross_forming}Equivalent normal forms allow} \\ {\color{color_cross_forming}extending existing methods}} \\
\arrayrulecolor{black}
\hline \hline
\end{tabular}
\label{tab:comparison-current-limiting-strategies}
\end{table*}

\subsubsection{Enhanced Voltage-Forming Reference in Polar Coordinates} For a voltage-forming control in polar coordinates, we enhance it by utilizing virtual active power feedback as
\begin{equation}
\label{eq:power-feedback}
    p = \Re{\hat{\pha v}\, \mathrm{conj}(\pha{i}_{\mathrm{o}}^{+})},
\end{equation}
where $\hat{\pha v} = \hat v \angle \hat \theta$ is the voltage vector reference, $\mathrm{conj}(\cdot)$ represents a conjugate operation, and $\pha{i}_{\mathrm{o}}^{+}$ is the measured positive-sequence current (saturated). In the power computation in \eqref{eq:power-feedback}, we use the \textit{voltage reference} $\hat{\pha v}$ instead of the measured terminal voltage $\pha v^{+}$. This power feedback enhancement is depicted in Fig.~\ref{fig:enhanced-voltage-forming}(a). The enhancement enables us to arrive at a standard form of power-angle relationship, thus recovering an equivalent normal form (shown later in Section~\ref{sec:transient-stability-polar}). We note that the concept of virtual power feedback, such as that based on an unsaturated \textit{current reference}, has been used in \cite{kkuni2024effects} (specifically for the ``current limiter with virtual impedance" strategy) to recover a standard power-angle relationship. Our virtual power feedback uses a \textit{voltage reference}, as cross-forming control renders the actual voltage magnitude to follow accordingly. Although virtual power may not accurately reflect actual power, its use can enhance transient stability, support dynamic frequency regulation, and likewise achieve steady-state active power sharing. Moreover, since the terminal voltage magnitude is not independently controlled during current saturation, it is preferable to disable the original voltage magnitude droop control. Instead, we can directly specify a fixed voltage magnitude reference \cite{zhang2024control} to further improve transient stability. Disabling the voltage droop control does not impact the parallel operation of multiple inverters, since the current-saturated cross-forming sources behave like voltage magnitude-following current sources, with their reactive current and power determined by the current limit and network impedance.

\subsubsection{Enhanced Voltage-Forming Reference in Rectangular Coordinates} For a voltage-forming dVOC control in $\alpha\beta$ coordinates, we enhance it by scaling the current feedback as
\begin{equation}
\label{eq:current-feedback}
    \pha{i}_{\lambda}^{+} = \frac{\pha{i}_{\mathrm{o}}^{+}}{\lambda},
\end{equation}
where $\lambda$ takes $\lvert \pha{\hat v}_{\lambda} \rvert / \lvert \pha{\hat v} \rvert$ for the explicit cross-forming control, with $\lvert \pha{\hat v}_{\lambda} \rvert$ given from \eqref{eq:explicit-cross-forming-a}, whereas $\lambda$ takes $\kappa \mu_{\mathrm{f}}$ as given from \eqref{eq:internal-voltage} for the implicit cross-forming control. The enhanced dVOC is depicted in Fig.~\ref{fig:enhanced-voltage-forming}(b). In this manner, we can correspondingly recover an equivalent normal form of dVOC (shown later in Section~\ref{sec:transient-stability-rect}).

Similar to the recovery of the virtual admittance control as discussed in Remark~\ref{rem:recovery}, when the current exits saturation after grid fault recovery, we can easily recover the normal voltage-forming reference from the enhanced one in Fig.~\ref{fig:enhanced-voltage-forming}(a) by restarting the voltage droop control, and likewise we can restore the normal dVOC from the enhanced dVOC in Fig.~\ref{fig:enhanced-voltage-forming}(b) by resetting $\lambda = 1$.

\subsection{Benchmarking Against Existing Strategies}
\label{sec:benchmarking}

We present a comparative analysis of the proposed cross-forming control versus typical existing strategies regarding the current limiting in grid-forming inverters. Table~\ref{tab:comparison-current-limiting-strategies} summarizes the key features and performance metrics of these strategies, categorizing them into voltage-forming-based, current-forming-based, and cross-forming-based strategies.

The current-limiting strategies previously investigated for grid-forming inverters encompass two main types: voltage-forming-based and current-forming-based strategies. As reviewed in Appendix~\ref{appendix:current-limiting}, voltage-forming-based strategies include two types: Type-A, characterized by adaptive virtual impedance mechanisms (also known as threshold virtual impedance) \cite{paquette2015virtual,qoria2020current}, and Type-B, characterized by a current limiter cascaded with a virtual admittance \cite{rosso2021implementation,kkuni2024effects,fan2022equivalent,zhang2023current,saffar2023impacts}. Conversely, current-forming strategies typically specify current references directly in the current-forming control setup \cite{rokrok2022transient,liu2023dynamic,li2023transient}. The developed cross-forming strategies can be implemented in an explicit or implicit manner. The merits of both the voltage-forming and current-forming strategies are preserved as the proposed cross-forming strategy enables forming characteristics for voltage angle and current magnitude.

A critical aspect of the difference among the strategies is the switching between control architectures. Voltage-forming-based strategies do not necessitate control switching, while current-forming-based strategies typically require us to switch the control mode. In contrast, cross-forming-based strategies feature backward compatibility of control architecture (more specifically, \textit{compatibility} means that it is straightforward to recover a standard virtual admittance control by reducing the cross-forming control architecture). Thus, we can avoid a complete switch between control architectures. Moreover, cross-forming strategies feature simple parameter tuning similar to Type-B voltage-forming strategies. In contrast, the tuning of Type-A voltage-forming strategies is complicated since it involves the consideration of the worst case \cite{paquette2015virtual} and it is also not straightforward to take a negative-sequence current into account in the case of asymmetrical grid faults.

In addition, we evaluate the speed of current limiting, indicating that Type-A voltage-forming-based strategies demonstrate a relatively slow response since a virtual impedance control is typically located outside the voltage control loop. The other strategies exhibit fast rapidity because of the use of a current limiter. Moreover, Type-A voltage-forming-based strategies cannot guarantee full overcurrent utilization, while the other ones fully utilize the overcurrent limit. Furthermore, we investigate FRT services (fault reactive current $i_Q$ and phase jump active current $i_P$ provision), highlighting the differing responses between strategies. In particular, cross-forming-based strategies offer natural $i_Q$ and $i_P$ provisions within the current limit, similar to voltage-forming-based strategies. In contrast, current-forming-based strategies may not be able to provide a natural response as fast as a voltage source to fulfill the requirements of FRT services. These differences in dynamic response performance are evident in principle, which have also been validated through simulations.

Last but not least, we discuss the resulting equivalent circuits and transient stability, emphasizing the challenges posed by a current-dependent virtual impedance in voltage-forming-based strategies \cite{fan2022equivalent,kkuni2024effects,zhang2023current,saffar2023impacts} and switching systems in current-forming-based strategies \cite{huang2019transient,rokrok2022transient,liu2023dynamic,li2023transient,wang2023transient}. Particularly, the current-dependent virtual impedance is influenced by grid fault severity and is thus unpredictable, bringing significant difficulties in transient stability analysis and insufficient robustness to unforeseen fault disturbances. In contrast, cross-forming-based strategies result in a constant equivalent impedance, demonstrating compatibility with the existing analysis methods, as presented in the next section.

\section{Equivalent Normal Forms and Transient Stability Analysis Extension}
\label{sec:transient-stability}

The transient stability of voltage-forming systems is rooted in the closed-loop relationship between the voltage-forming control and the grid network. Numerous transient stability results have been reported in the literature, but these results are primarily concentrated on the normal form of voltage-forming systems. The \textit{normal form} means that the system has normal structure and nominal dynamics, where the current is not saturated, e.g., when operating under normal grid conditions. In this work, we extend the existing stability analysis methods to the case of current saturation. To achieve this, we first consider the cross-forming control along with the enhanced voltage-forming reference to recover an equivalent normal form. Then, we extend the existing stability results concerning the normal form to the system under current saturation.

% =======
% FIG
% =======
\begin{figure}
  \begin{center}
  \includegraphics{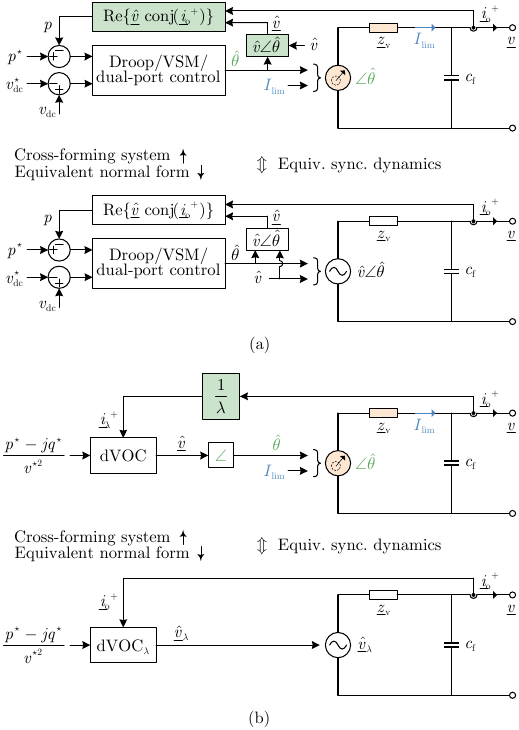}
  \caption{Cross-forming controlled closed-loop systems, resulting from the cross-forming regulator and the enhanced voltage-forming reference, maintain equivalent synchronization dynamics to the normal-form system, allowing us to extend the existing transient stability results to current-saturated conditions. (a) For droop control, VSM, AC-DC dual-port control, etc. in polar coordinates. (b) For dVOC in $\alpha\beta$ coordinates.}
  \label{fig:resulting-equivalent-circuit}
  \end{center}
\end{figure}

\subsection{Equivalent Normal Form in Polar Coordinates}
\label{sec:transient-stability-polar}

Consider the cross-forming control architecture that embeds an enhanced voltage-forming reference in polar coordinates. The cross-forming system shown in Fig.~\ref{fig:resulting-equivalent-circuit}(a) depicts a closed-loop connection between the enhanced voltage-forming reference and the equivalent circuit given earlier in Fig.~\ref{fig:desired-equivalent-circuit}(b). The equivalent circuit incorporates a cross-forming source that results from the cross-forming control. Fig.~\ref{fig:resulting-equivalent-circuit}(a) also shows the normal form representation of the system. The transient stability of the normal-form system has been widely investigated with the angle forming control in a closed loop with the so-called power-angle relationship, where the voltage magnitude is typically assumed to be constant. We show below that the cross-forming system maintains equivalent closed-loop nonlinear synchronization dynamics (i.e., transient stability characteristics) to the normal-form system. To show this, we simply need to show that the active power feedback (i.e., the power-angle relationship) in Fig.~\ref{fig:resulting-equivalent-circuit}(a) is equivalent, given that the angle forming control dynamics remain consistent.

Consider a typical assumption that the grid network is dominantly inductive, i.e., $\pha{z}_{\mathrm{v}} = j x_{\mathrm{v}}$. The virtual internal voltage is denoted by $\pha{\hat v}_{\lambda}$ as in the equivalent circuit in Fig.~\ref{fig:desired-equivalent-circuit}(b). We consider that the inverter is connected to a Thevenin equivalent grid with voltage $\pha{v}_{\mathrm{g}} \coloneqq v_{\mathrm{g}} \angle \theta_{\mathrm{g}}$ and grid impedance $\pha z_{\mathrm{g}} \coloneqq j x_{\mathrm{g}}$. The power feedback $p$ of the cross-forming system with saturated current $\pha{i}_{\mathrm{o}}^{+}$ is expanded from \eqref{eq:power-feedback} to (neglecting $c_{\mathrm{f}}$)
\begin{equation}
\label{eq:power-angle-equation}
\begin{aligned}
    p &= \Re{\hat{\pha v}\, \mathrm{conj}(\pha{i}_{\mathrm{o}}^{+})} \\
    &= \Re \biggl\{\hat v \angle \hat{\theta}\, \mathrm{conj}\biggl(\frac{\lvert \pha{\hat v}_{\lambda} \rvert \angle \hat{\theta} - v_{\mathrm{g}} \angle \theta_{\mathrm{g}}}{j x_{\mathrm{v}} + j x_{\mathrm{g}}} \biggr) \biggr\} \\
    &= \Re \biggl\{j \frac{\hat v \abs{\pha{\hat v}_\lambda} - \hat v v_{\mathrm{g}} \angle (\hat{\theta} - \theta_{\mathrm{g}})}{x_{\mathrm{v}} + x_{\mathrm{g}}} \biggr\} \\
    & = \frac{\hat v v_{\mathrm{g}}}{x_{\mathrm{v}} + x_{\mathrm{g}}} \sin{ (\hat{\theta} - \theta_{\mathrm{g}})},
\end{aligned}
\end{equation}
where it is seen that the current-saturated power-angle relationship conforms to the \textit{normal form} of a sine function. Hence, the cross-forming system shares equivalent transient stability characteristics to the normal-form system. If we further consider $\hat v =  v^{\star}$ [i.e., disabling the voltage droop control as in Fig.~\ref{fig:enhanced-voltage-forming}(a)], the power-angle relationship will be independent of the voltage dynamics. We note that the result in \eqref{eq:power-angle-equation} can be extended to a more general case where the network is not necessarily inductive but has a uniform $X/R$ ratio. Moreover, the result may be extended to multi-inverter networks, in which the coefficients of the sine terms in the active power flow equation will rely on the virtual internal voltage magnitude of the current-saturated inverters and the terminal voltage magnitude of the current-unsaturated inverters.

We note that the voltage dynamics of the cross-forming system and that of the normal-form system are of course different, since the former features voltage magnitude following (mainly determined by the exogenous circuit, cf. Proposition~\ref{prop:voltage-following}) while the latter features voltage forming. Hence, we only focus on the synchronization dynamics of the equivalent normal form while the voltage dynamics are not relevant.

\subsection{Equivalent Normal Form in Rectangular Coordinates}
\label{sec:transient-stability-rect}

Consider the cross-forming control architecture that embeds an enhanced voltage-forming reference in rectangular coordinates [specifically, the enhanced dVOC using \eqref{eq:current-feedback}]. The cross-forming system under current saturation is shown in Fig.~\ref{fig:resulting-equivalent-circuit}(b), which depicts a closed-loop connection between the enhanced dVOC using \eqref{eq:current-feedback} and an equivalent circuit. The normal form representation of the system is also shown in Fig.~\ref{fig:resulting-equivalent-circuit}(b). The transient stability of the normal form of dVOC has been explored based on the dVOC dynamics in a closed loop with the network voltage-current relationship, where the dynamics in both rectangular coordinates are included \cite{he2023quantitative,he2024passivity,colombino2019global}. We indicate in the following that the cross-forming system is comparable to the normal-form system in terms of the closed-loop nonlinear synchronization dynamics during current saturation.

Consider the enhanced dVOC in Fig.~\ref{fig:enhanced-voltage-forming}(b),
\begin{equation}
\label{eq:enhanced-dvoc}
    \dot{\pha{\hat v}} = j\omega_0\pha{\hat v} + \eta e^{j\varphi} \left( \frac{p^{\star} - jq^{\star}}{v^{\star 2}} \pha{\hat v} - \frac{\pha{i}_\mathrm{o}^{+}}{\lambda} \right) + \eta \alpha \frac{{v^{\star 2} - \hat v^2}}{v^{\star 2}} \pha{\hat v},
\end{equation}
which further gives rise to the equivalent \textit{normal form} as
\begin{equation}
\label{eq:equiv-dvoc}
    \dot{\pha{\hat v}}_{\lambda} = j\omega_0\pha{\hat v}_{\lambda} + \eta e^{j\varphi} \left( \frac{p^{\star} - jq^{\star}}{v^{\star 2}} \pha{\hat v}_{\lambda} - \pha{i}_\mathrm{o}^{+} \right) + \eta \alpha \frac{{v^{\star 2}_{\lambda} - \hat v^2_{\lambda}}}{v^{\star 2}_{\lambda}} \pha{\hat v}_{\lambda},
\end{equation}
where $\pha{\hat v}_{\lambda} = \lambda \pha{\hat v}$ and $v^{\star}_{\lambda} \coloneqq \lambda v^{\star}$, and the dynamics of $\lambda$ is ignored for ease of analysis. This indicates that \eqref{eq:equiv-dvoc}, shaped by the enhanced dVOC, describes the virtual internal voltage dynamics. It can be seen that the dynamics in \eqref{eq:equiv-dvoc} structurally conform to the normal form of dVOC. Thus, the cross-forming system shares the same transient stability characteristics (mainly in synchronization) as the normal-form system.

Likewise, we only care about the frequency synchronization dynamics of \eqref{eq:equiv-dvoc}. More specifically, the representation of the equivalent voltage dynamics in \eqref{eq:equiv-dvoc} is only aimed to analyze stability; it does not imply that the internal voltage $\pha{\hat v}_{\lambda}$ can be imposed independently, since the internal voltage magnitude $\lambda \abs{\pha{\hat {v}}}$ passively follows the circuit law, cf. Proposition~\ref{prop:voltage-following}.

\subsection{Transient Stability Results Extended to Current Saturation}
\label{sec:transient-stability-results}

Based on the equivalent normal form obtained earlier, we now extend the typical existing transient stability results to the case of current saturation. We mainly show the extension concerning a single-inverter infinite-bus system for ease of understanding. We refer the reader to \cite{schiffer2014conditions,colombino2019global,he2024passivity,choopani2020newmulti} for more comprehensive results for multi-inverter systems.

% =======
% FIG
% =======
\begin{figure}
  \begin{center}
  \includegraphics{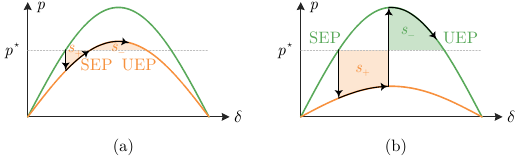}
  \caption{Transient stability of the cross-forming system in Fig.~\ref{fig:resulting-equivalent-circuit}(a) during (a) the fault-on period and (b) the post-fault period. The fault-on power-angle curve (orange) remains sinusoidal, as given in \eqref{eq:power-angle-equation}.}
  \label{fig:power-angle-curves}
  \end{center}
\end{figure}

As noted earlier, the transient stability of a voltage-forming control in polar coordinates is analyzed with the so-called power-angle relationship \cite{shuai2019transient}. For the equivalent normal form of the cross-forming system, Fig.~\ref{fig:power-angle-curves} illustrates typical power-angle curves as in \eqref{eq:power-angle-equation} with a normal grid voltage, a slight voltage dip in Fig.~\ref{fig:power-angle-curves}(a), and a severe voltage dip in Fig.~\ref{fig:power-angle-curves}(b). It can be seen that the power curves share the general sinusoidal characteristic. If stable equilibrium points (SEPs) exist during a fault, i.e., $p^{\star}$ intersects with the fault-on power-angle curve, as in Fig.~\ref{fig:power-angle-curves}(a), the system may achieve (fault-on) transient stability during the grid voltage dip. More specifically, for a dominantly first-order angle forming control (e.g., droop control), the transient stability can be maintained, provided that the equilibrium point is present. For a dominantly second-order angle forming control (e.g., VSMs), the transient stability during the grid voltage dip is determined by the difference between the accelerating area $s_{+}$ and the decelerating area $s_{-}$; see Fig.~\ref{fig:power-angle-curves}(a). If there are no equilibrium points during the voltage dip, as in Fig.~\ref{fig:power-angle-curves}(b), the system cannot achieve synchronization but can achieve (post-fault) synchronization after grid voltage recovery. Likewise, the transient stability is determined by the sufficiency of the decelerating area $s_{-}$ to counteract the accelerating area $s_{+}$. The equal-area criterion, a simple stability criterion inherited from conventional power systems, has been widely used to analyze and improve the transient stability of grid-forming inverters. Alternatively, one may resort to the energy function method for assessing transient stability, for example, a typical energy function of VSMs is given as \cite{kabalan2017largesignal}
\begin{equation}
\label{eq:energy-function}
    V(\omega,\delta) = \underbrace{\frac{1}{2} T_{J} \omega^2}_{\mathrm{kinetic\ energy}} \underbrace{- p_{\max} (\cos{\delta} - \cos{\delta_0}) - p^{\star} (\delta - \delta_0)}_{\mathrm{potential\ energy}},
\end{equation}
with inertia time constant $T_J$, power transfer limit $p_{\max}$, and stable equilibrium point $\delta_0$. For a VSM, the kinetic energy term represents the virtual energy stored in the virtual inertial mass $T_J$ rotating in the forming frequency $\omega$, and the potential energy term represents the virtual energy stored in the angle displacement $\delta$ relative to the equilibrium point $\delta_0$. The kinetic and potential energy are exchanged during acceleration and deceleration, as shown in Fig.~\ref{fig:power-angle-curves}. This energy conversion is driven by the power imbalance between the power setpoint $p^{\star}$ and the output $p$, occurring virtually in the VSM controller, unlike the real energy conversion in a synchronous machine (with a physical rotating mass) driven by mechanical and electromagnetic power; see \cite{fu2021large} for further details on the correspondence between grid-forming inverters and synchronous machines. The energy function in \eqref{eq:energy-function} and the information on the unstable equilibrium point (UEP) can be used to approximate the stability region and identify the critical clearance time for grid faults\cite{kundur1994power}.

It should be noted that both the equal-area criterion shown in Fig.~\ref{fig:power-angle-curves} and the energy function in \eqref{eq:energy-function} neglect damping effects, therefore being conservative. More advanced and less conservative variants are available in the literature; e.g., see the results in \cite{li2023iterative,moon2000estimating}. In actual applications, e.g., grid-forming inverters connecting renewable energy sources such as wind or photovoltaic to the grid, the active power setpoint $p^{\star}$ of the inverter, as well as the power from the side of the primary energy source, should be reduced during grid voltage dips to respect the power transfer limit of the transmission line and prioritize FRT services such as fault reactive current provision.

The transient stability of dVOC is typically analyzed in rectangular coordinates \cite{he2023quantitative,he2024passivity,colombino2019global}. A standard Lyapunov stability analysis was presented in our earlier work \cite{he2023quantitative}. For the equivalent normal form recovered in \eqref{eq:equiv-dvoc} connected to an infinite bus, a sufficient condition for transient stability is \cite{he2023quantitative},
\begin{equation}
\label{eq:stability-condition}
    \Re \bigl\{e^{j\varphi} \frac{p^{\star} - jq^{\star}}{{v}^{\star 2}} \bigr\} + \alpha < \frac{\alpha}{2} \frac{\hat v_{{\lambda}\mathrm{s}}^2}{v_{\lambda}^{\star 2}}  + \Re \bigl\{e^{j\varphi} \pha{y} \bigr\},
\end{equation}
where $\hat v_{{\lambda}\mathrm{s}}^2$ is the steady-state magnitude of the internal voltage $\pha{\hat v}_{\lambda}$, and $\pha y = 1/(\pha z_{\mathrm{v}} + \pha z_{\mathrm{g}})$ denotes the admittance of the lumped impedance seen from the virtual internal voltage to the infinite bus. This stability condition provides quantitative insights for parameter tuning, system operation, etc. Our recent results in \cite{desai2023saturation} show further extensions to multi-inverter systems.

% =================================================================================

\section{Simulation and Experimental Validations}
\label{sec:validations}

We present case studies to validate the performance of the proposed cross-forming control. For the validations in this work, we consider VSM as the voltage-forming reference. For validations where dVOC is considered as the voltage-forming reference, readers are referred to our recent work \cite{desai2023saturation}.  In Case Study~I, we validate the performance of both cross-forming implementations (explicit and implicit) under a symmetrical grid fault. In Case Study~II, we compare the performance of the cross-forming control with three typical strategies. In Case Study~III, we show the result of the implicit cross-forming control under an asymmetrical grid fault, where four different negative-sequence control modes are encompassed. In Case Study~IV, the performance of the cross-forming control is validated in multi-inverter grid-connected and islanded systems. Finally, experimental validations are presented.

% =======
% FIG
% =======
\begin{figure}
  \begin{center}
  \includegraphics{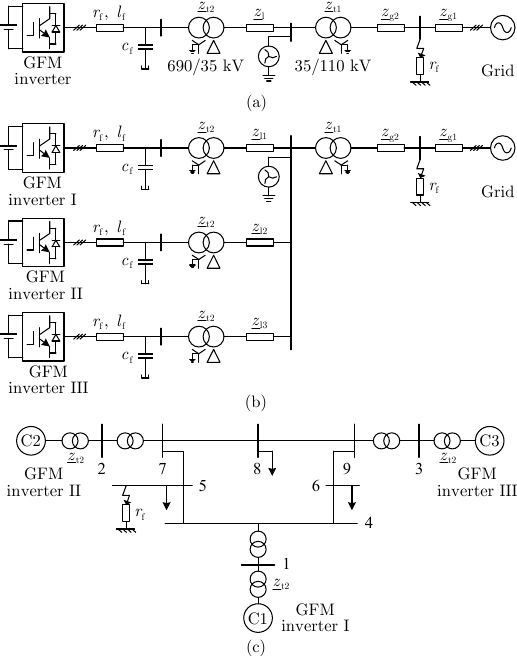}
  \caption{Illustration of the system models in Case Studies. (a) A single-inverter system in Case Studies I--III. (b) A multi-inverter system in Case Study~IVa. (c) IEEE 9-bus system in Case Study~IVb.}
  \label{fig:case-study-system}
  \end{center}
\end{figure}

\begin{table}
\centering
\caption{Parameters in Simulation Studies and Experiments}
\begin{tabular}[t]{lll}
\arrayrulecolor{black}
\hline \hline
Symbol                 & Description       & Value               \\
\hline
\multicolumn{3}{l}{(a) Parameters in simulation Case Studies I--IV} \\
\hline
$S_N$       & Nominal capacity  & \makecell[tl]{I--III:\,$200$\,MVA \\ IVa:\,$200/3$\,MVA\,for\,each \\ IVb:\,$247.5$,\,$192$,\,$128$\,MVA} \\
$\omega_0$       & Fundamental frequency & $100\pi$\,rad/s       \\
$\pha z_{\mathrm{g1}}$ & Grid impedance  & $0.01 + j0.1$\,pu\\
$\pha z_{\mathrm{g2}}$ & Grid impedance  & $0.003 + j0.03$\,pu \\
$\pha z_{\mathrm{l}}$  & Line impedance  & $0.01 + j0.05$\,pu \\
$\pha z_{\mathrm{l1}}$  & Line impedance  & $0.01 + j0.05$\,pu \\
$\pha z_{\mathrm{l2}}$  & Line impedance  & $0.02 + j0.10$\,pu \\
$\pha z_{\mathrm{l3}}$  & Line impedance  & $0.03 + j0.15$\,pu \\
$\pha z_{\mathrm{t1}}$ & Transformer impedance & $0.16/30 + j0.16$\,pu \\
$\pha z_{\mathrm{t2}}$ & Transformer impedance & $0.06/30 + j0.06$\,pu \\
$r_{\mathrm{f}}$ & Fault grounding resistance & $1.0\,\Omega$ \\
$l_\mathrm{f}$      & Filter inductance & $0.05$\,pu \\
$r_\mathrm{f}$      & Filter inductance & $0.05/10$\,pu\\
$c_\mathrm{f}$      & Filter inductance & $0.05$\,pu\\
$p^{\star}$ & Active power setpoint & \makecell[tl]{I:\,$0.2$\,pu \\ II:\,$1.0$\,pu \\ III:\,$0.2$\,pu \\ IVa,b:\,$0.5$,\,$0.7$,\,$0.9$\,pu} \\
$q^{\star}$ & Reactive power setpoint & \makecell[tl]{I--III: $0.0$\,pu \\ IVa,b: $0.5$,\,$0.3$,\,$0.1$\,pu} \\
$v^{\star}$ & Voltage setpoint & \makecell[tl]{I--IVa:\,$1.0$\,pu \\ IVb:\,$1.1$\,pu} \\
$T_J$ & Inertia time constant & $5$\,s \\
$D$ & Damping coefficient & $25$ \\
$m_q$ & Reactive power droop gain & $0.2$ \\

$\pha {z}_\mathrm{v}$      & Virtual impedance & \makecell[tl]{I--IVa:\,$j0.2$\,pu \\ IVb:\,$0.2e^{j5\pi/12}$\,pu} \\
$\tau_{\mathrm{v}}$ & LPF time const. in $\pha v$ feedback & $0.01$\,s \\
$\tau_{\mu}$ & LPF time const. in $\mu$ feedback & $0.01$\,s \\
$\kappa$ & Cross-forming feedforward gain & $1$ \\
$\kappa_{\mathrm{i}}$ & Cross-forming integral gain & $50$ \\
$\kappa_{\mathrm{vi}}$ & Feedback gain in adaptive VI & $0.91$ \\
$\sigma_{\mathrm{vi}}$ & $X/R$ ratio in adaptive VI & $10$ \\

$k^{-}$ & $K$-factor in neg.-seq mode IV & $6$ \\
$I_{\lim}$ & Current limit & $1.1$\,pu \\
\hline
\multicolumn{3}{l}{(b) Parameters in experiments} \\
\hline
$U_{\mathrm{rms}}$       & Line-to-line voltage level     & $200$\,V \\
$S_N$       & Nominal capacity  & $1$\,kVA                  \\
$\omega_0$       & Fundamental frequency & $100\pi$\,rad/s       \\
$L_\mathrm{f}$      & Filter inductance & $1.5$\,mH ($0.012$\,pu) \\
$R_\mathrm{f}$      & Filter inductance & $1.0\,\Omega$ ($0.025$\,pu)\\
$C_\mathrm{f}$      & Filter inductance & $3.5\,\mu$F ($0.044$\,pu)\\
$G_\mathrm{load}$    & Resistor load conductance & $0.124$\,pu \\
$\pha {z}_\mathrm{g}$      & Grid impedance (emulated) & $j0.2$\,pu\\
$p^{\star}$ & Active power setpoint & $0.1$\,pu \\
$q^{\star}$ & Reactive power setpoint & $0.0$\,pu \\
$v^{\star}$ & Voltage setpoint & $1.1$\,pu \\
$\pha {z}_\mathrm{v}$ & Virtual impedance & $j0.1,\,j0.6$\,pu \\
$f_{\mathrm{sw}}$ & Switching frequency & $32$\,kHz \\
$f_{\mathrm{s}}$ & Sampling and control frequency & $8$\,kHz \\
\hline \hline
\end{tabular}
\label{tab:system-parameters}
\end{table}

The system models used in the simulations are depicted in Fig.~\ref{fig:case-study-system}, where a single-inverter system, a multi-inverter grid-connected system, and the IEEE 9-bus system with three inverters are included. Grid faults are simulated on the high-voltage network through a grounding resistor. Grid-forming inverters are configured with cascaded control loops, as shown in Fig.~\ref{fig:generic-control-diagram} in Appendix~\ref{appendix:standard-gfm}. Since the inner voltage controller is arranged as a static virtual admittance element, in contrast to conventional PI/PR dynamic regulators, the controller parameter tuning is simple, allowing a wide range of parameter value choices. The guidelines outlined in Section~\ref{sec:parameter-tuning} have been adopted to choose specific parameters, which are summarized in Table~\ref{tab:system-parameters}. In particular, we choose the virtual impedance $\pha{z}_{\mathrm{v}}$ typically small to avoid a large virtual voltage drop and reactive power consumption \cite{zhang2023current}; purely inductive to test the worst case where the virtual impedance does not provide a resistive damping effect \cite{taoufik2022variable}. In addition, the time constant of the LPFs in the cross-forming regulator is set to $0.01$ or $0.02$\,s, resulting in a bandwidth of around $10$\,Hz \cite{wu2022smallsignal}, separated from the bandwidth of the inner current control loop.

% =======
% FIG
% =======
\begin{figure*}
  \begin{center}
  \includegraphics{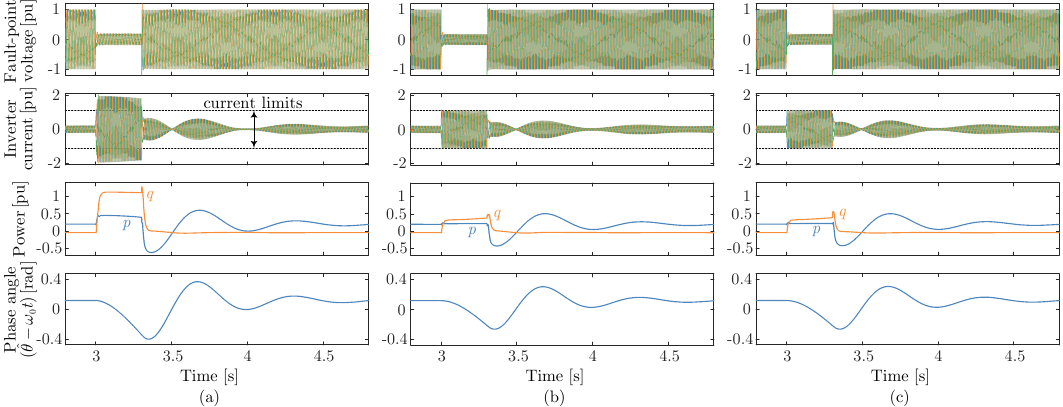}
  \caption{Result of Case Study~I under a symmetrical grid fault: (a) Without using any current-limiting strategies; (b) With the explicit cross-forming control under the grid fault. (c) With the implicit cross-forming control under the grid fault. Note that the post-fault oscillatory recovery is a result of the inertial response of the VSM, and the oscillation can be avoided by employing inertia-less voltage-forming controls, e.g., droop control.}
  \label{fig:case-study-I}
  \end{center}
\end{figure*}

\subsection{Case Study~I: Symmetrical Grid Faults}

The result of Case Study~I is shown in Fig.~\ref{fig:case-study-I}. A symmetrical short-circuit fault occurs at $3$\,s and is cleared at $3.3$\,s. The fault occurs between $\pha z_{\mathrm{g1}}$ and $\pha z_{\mathrm{g2}}$ in the system of Fig.~\ref{fig:case-study-system}(a). The inverter without using any current-limiting strategies immediately suffers from overcurrent, as can be seen in Fig.~\ref{fig:case-study-I}(a). In Fig.~\ref{fig:case-study-I}(b) and \ref{fig:case-study-I}(c), in contrast, both the explicit and implicit cross-forming controls successfully limit the inverter current to the prescribed limit $1.1$\,pu rapidly due to their inherent current magnitude forming capability. Since the current is saturated, the power injection is consequently reduced compared to the unlimited scenario, indicating the actual capability of the inverter to provide FRT services under current limits. It can be seen that the increase in reactive power injection is fast after the grid fault occurs, such that the reactive power/current response can satisfy the fault current specifications in grid codes. This is attributed to the angle forming functionality of the cross-forming control. In other words, the angle of the virtual internal voltage vector is controlled and remains the same as the pre-fault value at the moment of the fault occurrence (and then slowly exhibits an inertial response); see the phase angle $(\hat{\theta} - \omega_0 t)$ in Fig.~\ref{fig:case-study-I}. The reactive current component in the fault current naturally increases at the fault moment, as explained in Fig.~\ref{fig:fault-services}(c). To provide more reactive power/current during the fault, one can reduce the power angle by lowering the active power setpoint. It can also be observed from Fig.~\ref{fig:case-study-I} that the phase angle gradually decreases before the fault clearance, which is because the (Thevenin) equivalent grid voltage phasor undergoes a backward phase jump at the time of the short-circuit fault.

Upon detecting grid voltage recovery, the control reverts to the normal voltage-forming mode. Subsequently, the active power and phase angle exhibit an inertial response before stabilizing at their pre-fault steady-state values. If the fault lasts longer and the power setpoint is not too large, the inverter will synchronize with the faulty grid and settle to a steady state. To obtain transient stability guarantees for either the fault-on stage or the post-fault stage, one can apply the extended method in Section~\ref{sec:transient-stability-results}. This is possible as the proposed cross-forming regulator along with the enhanced voltage-forming reference (VSM in this case) leads to a constant impedance in the equivalent circuit and an equivalent normal form of the closed-loop system.

% =======
% FIG
% =======
\begin{figure}
  \begin{center}
  \includegraphics{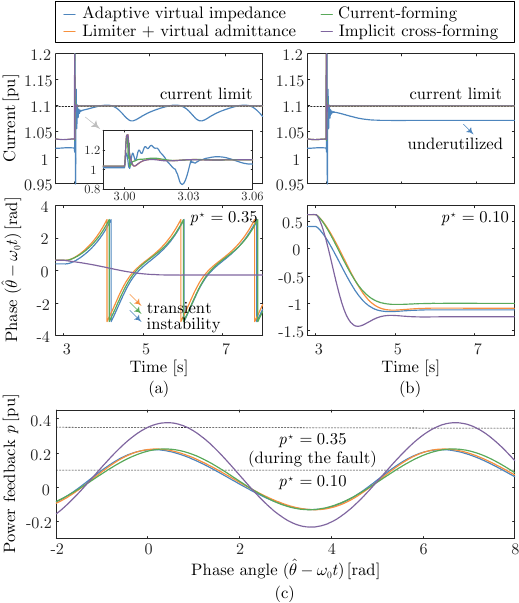}
  \caption{Result of Case Study~II: Current-limiting and stability performance comparison of the proposed cross-forming strategy against three typical strategies. (a) and (b) are the results with $p^{\star} = 0.35$ and $0.10$\,pu during the fault, respectively. (c) Power-angle curves during the fault.}
  \label{fig:case-study-II}
  \end{center}
\end{figure}

\subsection{Case Study~II: Comparison With Existing Strategies}

In Case Study~II, a permanent symmetrical short-circuit fault is simulated at $3$\,s. The fault occurs between $\pha z_{\mathrm{g1}}$ and $\pha z_{\mathrm{g2}}$ in the system of Fig.~\ref{fig:case-study-system}(a). The active power setpoint is set to $1.0$\,pu to test the fault current limitation under a heavily loaded pre-fault condition. However, it is reduced to a low level during the fault to prioritize reactive current injection and facilitate transient stability performance. Two different power setpoint levels for the fault period, $0.35$ and $0.10$\,pu, are considered to compare the differences in transient stability between multiple strategies. The result is displayed in Fig.~\ref{fig:case-study-II}. The benchmarking analysis in Section~\ref{sec:benchmarking} is corroborated with the help of this simulation. More specifically, we observe that the limiter + virtual admittance strategy, the current-forming strategy, and the proposed cross-forming strategy capably rapidly limit the current to the prescribed value. The adaptive virtual impedance is relatively slow as it suffers from the limitation of the voltage control bandwidth. Moreover, due to the use of proportional feedback control, the adaptive virtual impedance cannot fully utilize the overcurrent limit, i.e., the current is limited under $I_{\mathrm{lim}}$, as observed from Fig.~\ref{fig:case-study-II}(a) and (b).

We depict the power-angle curves for the control strategies during the fault in Fig.~\ref{fig:case-study-II}(c), i.e., phase portraits with recorded power and phase angle trajectories under unstable cases. Note that the power-angle concept is valid for the current-forming control since it uses a power droop, instead of a PLL, for synchronization \cite{huang2019transient}. From Fig.~\ref{fig:case-study-II}(c), we can see that the cross-forming control has a larger transient stability margin than the others since the enhancement of the power feedback as in \eqref{eq:power-feedback} results in a higher magnitude of the power-angle curve. Therefore, it can achieve transient stability even with a higher power setpoint $0.35$\,pu. In contrast, the existing strategies exhibit transient instability behaviors due to the loss of equilibrium points in this case, as observed from the divergent phase angle response in Fig.~\ref{fig:case-study-II}(a). With a lower power setpoint $0.10$\,pu during the fault, all the strategies show a stable response in Fig.~\ref{fig:case-study-II}(b). However, for these existing strategies, their inherent transient stability is less robust as the power-angle curve is reduced and thus the allowable range of power-angle variation becomes smaller. Using some transient stability enhancement methods with these strategies (see Table~\ref{tab:comparison-current-limiting-strategies}), it is possible to obtain stability. Nevertheless, the analysis of transient stability for the existing strategies continues to present significant challenges, stemming from the current-dependent equivalent impedance or the switching of control architectures in these traditional strategies. On the other hand, the proposed cross-forming control allows us to readily extend the pre-existing transient stability results. This is a significant advantage that has not yet been established with the existing current-limiting strategies. In short, the proposed cross-forming strategy is easier to implement and analyze.

% =======
% FIG
% =======
\begin{figure*}
  \begin{center}
  \includegraphics{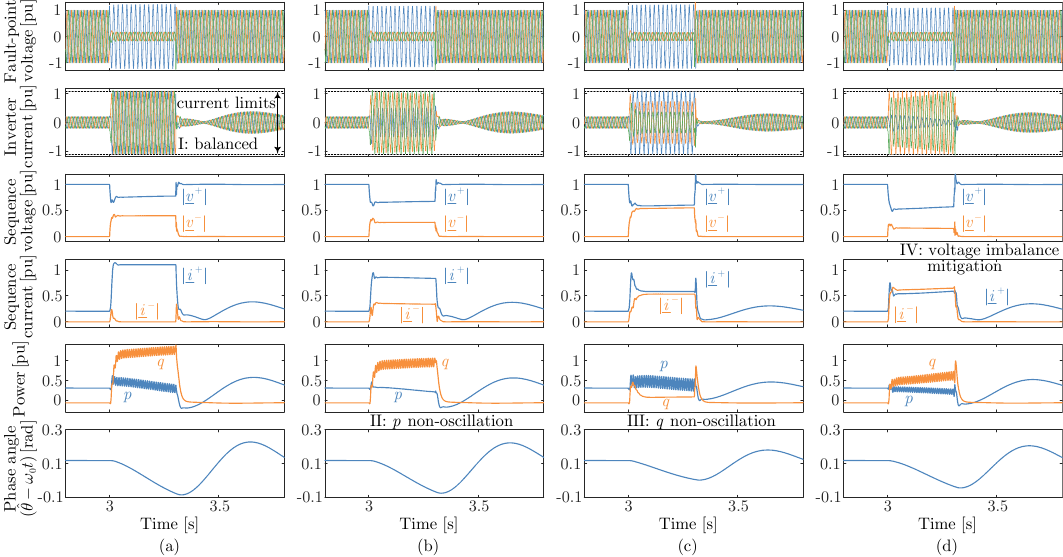}
  \caption{Result of Case Study~III under an asymmetrical grid fault: (a) Mode I -- Balanced current control; (b) Mode II -- Active power oscillation suppression; (c) Mode III -- Reactive power oscillation suppression; (d) Mode IV -- Negative-sequence voltage mitigation.}
  \label{fig:case-study-III}
  \end{center}
\end{figure*}

\subsection{Case Study~III: Asymmetrical Grid Faults}

In Case Study~III, an asymmetrical fault (double line-to-ground fault) occurs at $3$\,s and is cleared at $3.3$\,s. The fault occurs in between $\pha z_{\mathrm{g1}}$ and $\pha z_{\mathrm{g2}}$ in the system shown in Fig.~\ref{fig:case-study-system}(a). The result of Case Study~III is shown in Fig.~\ref{fig:case-study-III}. We observe that the fault-point voltage becomes unbalanced during the fault and the voltages of the faulty phases B and C decrease.

For the control of the negative sequence current, four control modes are presented in Fig.~\ref{fig:case-study-III}. In Fig.~\ref{fig:case-study-III}(a), the result for a balanced current mode is displayed. This control ensures that the inverter current remains balanced. For the result observed in Fig.~\ref{fig:case-study-III}(b)/(c), an active/reactive power oscillation suppression mode (Mode II/III; see Appendix~\ref{appendix:negative-sequence-current}) is employed. We observe that with the help of this control mode, the active/reactive power is non-oscillatory. Finally, in Fig.~\ref{fig:case-study-III}(d), the result for a negative-sequence voltage mitigation mode (Mode IV) is displayed. Using this control mode, the negative-sequence voltage is reduced by intentionally absorbing negative-sequence reactive current through a virtual inductance as in \eqref{eq:negative-sequence-current-spec4}. Thus, the negative-sequence service required by the grid codes is satisfied \cite{ieee2800}. However, the positive-sequence voltage is also reduced compared to the result in Fig.~\ref{fig:case-study-III}(a). This reduction is because the positive-sequence reactive current component is compromised and the sequence networks are coupled, similar to the grid-following case \cite{he2022synchronization}. The results in this case study validate that the proposed cross-forming control successfully limits the maximum phase current magnitude to the prescribed limit for asymmetrical faults and is compatible with any of the four negative-sequence control modes. The control modes are reviewed in more detail in Appendix~\ref{appendix:negative-sequence-current}.

% =======
% FIG
% =======
\begin{figure}
  \begin{center}
  \includegraphics{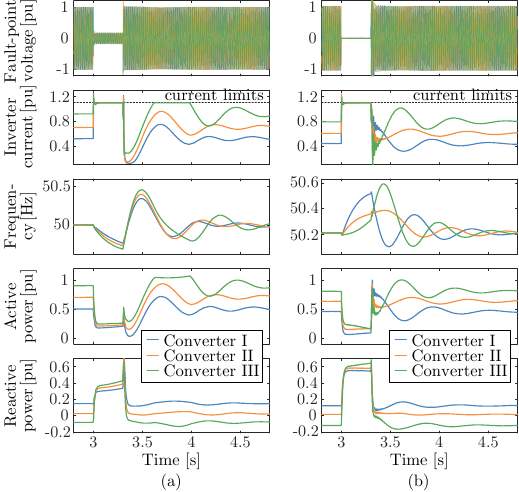}
  \caption{Result of Case Study~IV: (a) Three-inverter grid-connected system in Fig.~\ref{fig:case-study-system}(b); (b) IEEE 9-bus system with three inverters in Fig.~\ref{fig:case-study-system}(c).}
  \label{fig:case-study-IV}
  \end{center}
\end{figure}

\subsection{Case Study~IV: Multi-Inverter Scenarios}

The results for the multi-inverter scenarios are shown in Fig.~\ref{fig:case-study-IV}. Case Study~IVa represents the grid-connected system results and Case Study~IVb represents the IEEE 9-Bus system results.

\subsubsection{Case Study~IVa} A symmetrical short-circuit fault occurs at $3$\,s and is cleared at $3.3$\,s in the grid-connected system shown in Fig.~\ref{fig:case-study-system}(b). The three inverters have the same power ratings but different power setpoints and different line impedance values to the point of common coupling (PCC), as indicated in Table~\ref{tab:system-parameters}(a). The simulation result in Fig.~\ref{fig:case-study-IV} shows that the inverters successfully ride through the grid fault. The inverter currents remain limited at/within the prescribed value during and after the fault, and reactive power is injected during the fault. Since the three-phase currents are slightly distorted due to transients at the time of the occurrence and clearance of the fault, the 2-norm current magnitude expression $\lvert {i_\alpha + j i_\beta} \rvert$ cannot precisely represent the phase current peaks at these time instants. Therefore, even though the transient peak of the 2-norm current magnitude exceeds the limit, the actual phase current magnitudes ($\infty$-norm) remain within the limit. It is important to note that the active power setpoints are reduced to $0.2$\,pu during the fault to prioritize reactive power injection and alleviate active power imbalance to improve transient stability. The inverters achieve transient stability after grid fault recovery even if the current limit may be touched again (see the green waveform) during fault recovery.

\subsubsection{Case Study~IVb} A symmetrical short-circuit fault occurs at $3$\,s and is cleared at $3.3$\,s in the IEEE 9-bus system in Fig.~\ref{fig:case-study-system}(c). The three inverters have different power ratings and different power setpoints, as indicated in Table~\ref{tab:system-parameters}(a). The simulation result shows that the inverters successfully ride through the grid fault and their currents remain limited at the prescribed value during the severe grid fault. Similarly, the active power setpoints are reduced to $0.2$\,pu during the fault to improve post-fault transient stability. It is important to note that all the inverters enter the cross-forming operating mode during the grid fault and, as a result, their voltages concurrently follow the current injection based on the circuit law. Since the IEEE 9-bus system contains only constant impedance loads, the system voltage behavior depends on the current response. This is similar to the islanded operation of a current-forming inverter discussed in \cite{li2022revisiting}. However, if there are voltage-dependent nonlinear loads, the system may require voltage magnitude forming devices to respond to form the voltage magnitude. This is possible in practice if some grid-forming inverters are far from the fault location and thus do not enter current saturation during the fault (i.e., consistently preserving voltage magnitude-and-angle forming).

% =======
% FIG
% =======
\begin{figure}
  \begin{center}
  \includegraphics{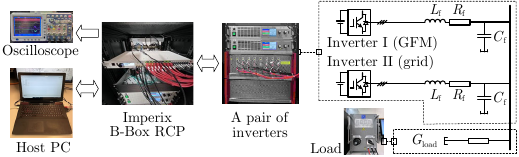}
  \caption{Experimental setup based on prototype inverters, where inverter II emulates grid voltage, grid impedance, and different grid faults.}
  \label{fig:exp-setup}
  \end{center}
\end{figure}

In all the case studies above, once the grid voltage recovers, the inverter exits current saturation based on fault clearance detection through terminal voltage measurement. A brief period of current re-saturation may occur during the transition back to the standard voltage-forming mode. However, the inverter does not need to re-enter the cross-forming mode, which is crucial to avoid repeated switching between the two modes. Following fault recovery, the system's transient stability significantly improves due to the restored voltage. As a result, even under the standard voltage-forming mode, the inverter can endure the brief current re-saturation period and stabilize, provided that the system state is near the steady state \cite{zhang2023current}.

\subsection{Experimental Results}

% =======
% FIG
% =======
\begin{figure}
  \begin{center}
  \includegraphics{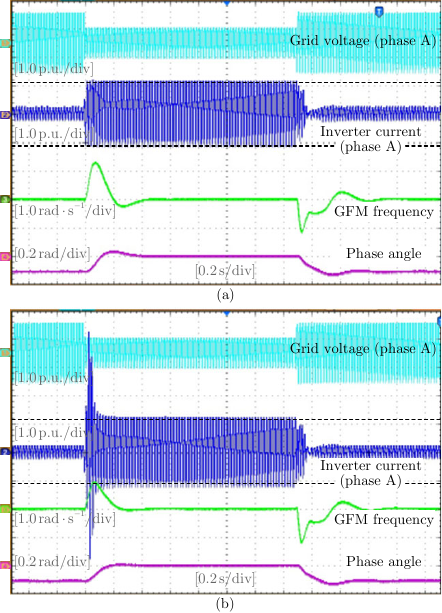}
  \caption{Experimental results under a symmetrical grid fault, where the grid voltage dips to $0.5$\,pu and phase jumps $+15$ degrees. (a) With the proposed current-limiting strategy. (b) Without using current limiting.}
  \label{fig:exp-voltage-dip}
  \end{center}
\end{figure}

% =======
% FIG
% =======
\begin{figure}
  \begin{center}
  \includegraphics{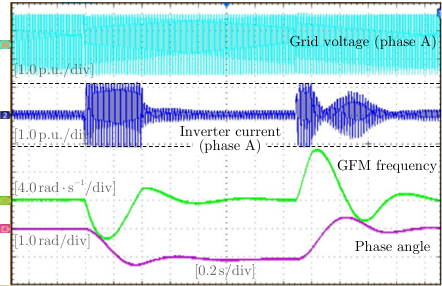}
  \caption{Experimental results under grid phase jump $-60$ degrees. Note again that the recovery is oscillatory due to the inertial response.}
  \label{fig:exp-phase-jump}
  \end{center}
\end{figure}

% =======
% FIG
% =======
\begin{figure*}
  \begin{center}
  \includegraphics{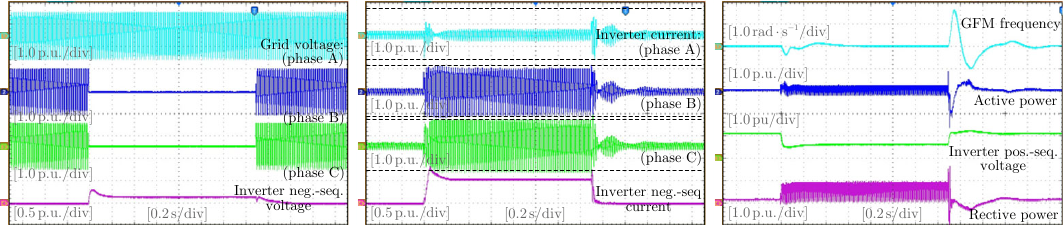}
  \caption{Experimental results under an asymmetrical grid fault, where phase A voltage remains while phase B and C voltages dip to zero.}
  \label{fig:exp-asym}
  \end{center}
\end{figure*}

The experimental setup depicted in Fig.~\ref{fig:exp-setup} is used to further validate the performance of the cross-forming control. The setup involves the interconnection of two parallel inverters, one employing a VSM-based voltage-forming reference and the other emulating grid voltage, grid impedance, and different grid faults. Since the DC voltage sources of both inverters are bi-directional, a resistive load $G_{\mathrm{load}}$ is connected to the AC bus to consume the power delivered by the inverters. Both inverters are controlled by an Imperix B-Box rapid control prototyping system. The parameters used in the experiments are given in Table~\ref{tab:system-parameters}(b).

Fig.~\ref{fig:exp-voltage-dip} shows the experimental results for a symmetrical disturbance. The voltage dips to $0.5$\,pu and the phase angle jumps by $+15$\,degrees. The current of the grid-forming inverter is limited to the prescribed value when using the proposed cross-forming strategy, as demonstrated in Fig.~\ref{fig:exp-voltage-dip}(a). In contrast, when no current limiting is employed, the inverter suffers from an immediate overcurrent peak and subsequent steady-state overcurrent during the fault, as shown in Fig.~\ref{fig:exp-voltage-dip}(b). Since the voltage dip lasts for $1.5$\,s, the cross-forming controlled inverter synchronizes with the grid during the fault. The inverter also achieves transient stability after grid fault recovery.

In Fig.~\ref{fig:exp-phase-jump}, the experimental results under a grid phase jump of $-60$ degrees (without any voltage magnitude changes) are shown. A higher virtual impedance $\pha z_{\mathrm{v}} = 0.6$\,pu is used to guarantee that there exist feasible operating points while operating the inverter in the cross-forming mode. More specifically, as illustrated in Fig.~\ref{fig:fault-services}(c), when the power angle between $\angle \pha v$ and $\angle \pha v_{\mathrm{g}}$ suddenly changes considerably, a large radius of $\vert \Delta \pha v_{z} \vert = \vert (\pha z_{\mathrm{g}} + \pha z_{\mathrm{v}}) \pha i \vert$ is required to ensure that the circle always intersects with the angle direction $\angle \pha v$ during the grid resynchronization. Otherwise, the cross-forming control will collapse due to the non-existence of a feasible operating point. From Fig.~\ref{fig:exp-phase-jump}, it can be seen that the overcurrent is limited during the fault and the inverter also synchronizes with the grid. Throughout the current-limiting process, there are always operating points because of the use of a large virtual impedance. When the grid resynchronization is achieved, the overcurrent state exists, and the inverter returns to the original steady state.

Finally, Fig.~\ref{fig:exp-asym} shows the experimental results under an asymmetrical grid voltage dip. Phase A voltage remains the same while phase B and C voltages dip to zero, emulating a bolted double line-to-ground fault. During the asymmetrical fault, phase C current stays at the limit, phase B current remains close to the limit, and phase A current remains small. We employ the negative-sequence voltage mitigation mode, which leads to a small negative-sequence voltage at the inverter terminal. Similar to the result under the symmetrical fault, synchronization is achieved during the fault stage and after the recovery stage.

\section{Conclusion}
\label{sec:conclusion}

We present the concept of cross-forming inverter control, particularly applicable for current-saturated operation during symmetrical or asymmetrical grid faults. With the philosophy to regulate the voltage angle even when the current output of the inverter is saturated, the cross-forming control consistently provides voltage angle forming services. Therefore, the cross-forming control integrates voltage angle forming and current magnitude forming characteristics, inherently satisfying grid-forming objectives/specifications under grid faults, including angle-/frequency-forming synchronization, FRT services provision, fault current limiting, etc. Based on the proposed control and the resulting equivalent circuit featuring a constant virtual impedance, we establish an equivalent normal form of the system. The equivalent normal form represents structurally identical synchronization dynamics to the normal system, thus enabling the extension of previously established transient stability results from unsaturated to saturated conditions. The extension makes transient stability analysis, assessment, and guarantees under current saturation tractable. It also facilitates the use of consistent modeling and analysis approaches across fault and normal conditions, resembling the practices in conventional power systems.

While we develop and validate two simple and feasible implementations of the cross-forming control, further research is necessary to explore more advanced variants. The explicit and implicit types indicate two technical methodologies. In addition, despite the categorized review and the benchmarking analysis in this work, we believe a new, dedicated, comprehensive study is necessary to fairly compare all available strategies, taking into account factors such as current-limiting performance, transient and small-signal stability, response capabilities to serve for grid code requirements, and implementation complexity. Furthermore, identifying the control capability boundaries of grid-forming inverters under current saturation is crucial for formulating clearer and more reasonable performance requirements for grid-forming inverters.

\appendices

\section{Voltage-Forming Controls: A Review}
\label{appendix:standard-gfm}

A typical control block diagram for a voltage-forming inverter is shown in Fig.~\ref{fig:generic-control-diagram}. In the following, we introduce the state-of-the-art techniques applied in each control module.

% =======
% FIG
% =======
\begin{figure}
  \begin{center}
  \includegraphics{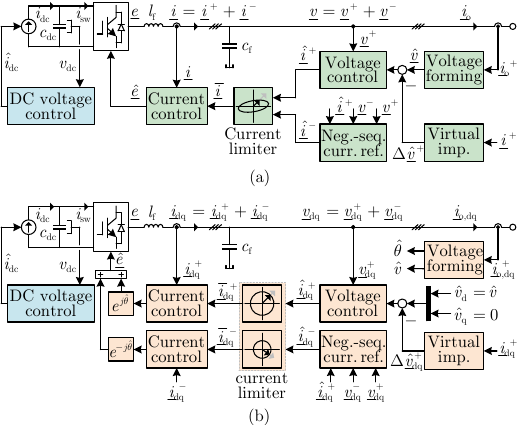}
  \caption{Typical control diagrams of voltage-forming inverters. (a) Implementation in the stationary reference frame. (b) Implementation in the rotational reference frame. Fig.~\ref{fig:neg-seq-curr-ref} displays the negative-sequence component control modes. Figs.~\ref{fig:current-limiter-albe} and \ref{fig:current-limiter-dq} illustrates the current limiters.}
  \label{fig:generic-control-diagram}
  \end{center}
\end{figure}

\subsection{Normal Forms of Voltage-Forming Controls}

Numerous voltage-forming control schemes have been developed in the literature. We categorize them into the following three major types considering similarities in their control architectures \cite{dorfler2023control}. In particular, these controls are termed ``\textit{normal forms}" since their initial designs aim to operate as normal voltage sources without considering current saturation. Moreover, since the voltage-forming controls are mostly applied only in the positive-sequence domain, we use $\pha{\hat{v}}^{+} = \pha{\hat v} = \hat{v} \angle \hat{\theta}$ to indicate the \textit{positive-sequence} voltage-forming reference throughout the paper.

\textit{Single-Input Single-Output Linear Type (Droop Control and Virtual Synchronous Machine):} Droop control and VSMs (and their variants) are prevalent voltage-forming control schemes. Droop control and VSMs are developed with the main consideration of the nominal operating point, where the network's power flows are reasonably approximated through decoupling and linearization. Therefore, both controls are single-input single-output (SISO), and linear in their structures. In particular, in a dominantly inductive network context, the normal form of droop control is given as
\begin{subequations}
    \begin{align}
        \dot {\hat\theta} = \hat \omega &= \omega_0 + m_p (p^{\star} - p) \\
        \hat v &= v^{\star} + m_q (q^{\star} - q)
    \end{align}
\end{subequations}
with power droop gains $m_p \in \mathbb{R}_{\geq 0}$ and $m_q  \in \mathbb{R}_{\geq 0}$, power setpoints $p^{\star}$ and $q^{\star}$, voltage setpoint $v^{\star}$, and nominal frequency $\omega_0$. The power feedback $p$ and $q$ should be taken from positive-sequence components (i.e., calculated with positive-sequence voltage and current components; the same applies hereinafter).

Similar to droop control, VSMs have been introduced to emulate a synchronous machine electromechanical model,
\begin{subequations}
    \begin{align}
        \dot {\hat \theta} &= \hat \omega \\
        T_{J} \dot {\hat \omega} &= -D(\hat \omega - \omega_0) + (p^{\star} - p) \\
        \hat v &= v^{\star} + m_q (q^{\star} - q)
    \end{align}
\end{subequations}
with virtual inertia time constant $T_{J}$ and damping gain $D$.

\textit{Multivariable and Nonlinear Type (Complex Droop Control and Dispatchable Virtual Oscillator Control):} Unlike the decoupled SISO mechanism used in droop control and VSMs, ``complex droop control" is structurally multivariable and nonlinear \cite{he2023quantitative}. Consequently, it can effectively manage the inherent coupling and nonlinearity of the network's active and reactive power flows. Furthermore, it performs well even when operating far from the nominal point. In polar coordinates, complex droop control reads as
\begin{subequations}
\label{eq:complex-droop-polar}
    \begin{align}
    \dot {\hat \theta} = \hat \omega &= \omega_0  + \eta \left( \frac{p^{\star}}{v^{\star 2}} - \frac{p}{\hat v^2} \right),\\
    \frac{\dot {\hat v}}{\hat v} = \hat \varepsilon &= \eta \left(\frac{q^{\star}}{v^{\star 2}} - \frac{q}{\hat v^2} \right) + \eta \alpha \frac{{v^{\star 2} - \hat v^2}}{v^{\star 2}}
    \end{align}
\end{subequations}
with power droop gain $\eta  \in \mathbb{R}_{\geq 0}$ and voltage magnitude droop gain $\alpha  \in \mathbb{R}_{\geq 0}$. The complex number with ${\dot {\hat v}}/{\hat v}$ as the real part and $\dot {\hat \theta}$ as the imaginary part, i.e., ${\dot {\hat v}}/{\hat v} + j\dot {\hat \theta}$, is known as complex frequency \cite{milano2022complex}. The complex frequency represents both the rate of change of the voltage magnitude, $\hat \varepsilon$, and the angular speed, $\hat \omega$. This explains why \eqref{eq:complex-droop-polar} is termed complex droop control (complex-power complex-frequency droop control).

In complex voltage vector coordinates, complex droop control is equivalently rewritten as dispatchable virtual oscillator control (dVOC) \cite{he2023quantitative}, 
\begin{equation}
\label{eq:dvoc}
    \dot{\pha{\hat v}} = j\omega_0\pha{\hat v} + j \eta \left( \frac{p^{\star} - jq^{\star}}{v^{\star 2}} \pha{\hat v} - \pha{i}_\mathrm{o}^{+} \right) + \eta \alpha \frac{{v^{\star 2} - \hat v^2}}{v^{\star 2}} \pha{\hat v},
\end{equation}
where $\pha{\hat v}$ is the voltage reference vector and $\pha{i}_\mathrm{o}^{+}$ the output current (positive-sequence) \cite{he2023quantitative}. It has been shown that complex droop control (i.e., dVOC) guarantees the global asymptotic stability of voltage-forming inverters in both islanded \cite{colombino2019global} and grid-connected scenarios \cite{he2023quantitative}.

\textit{AC-DC Dual-Port Type (Machine Matching and Dual-Port Control):} Most existing voltage-forming controls focus on AC grid forming while neglecting the DC-bus voltage dynamics and regulation. To overcome this limitation, a dual-port voltage-forming control has been developed as \cite{subotic2022power}
\begin{subequations}
\label{eq:dual-port-control}
    \begin{align}
        \label{eq:dual-port-control-1}
        \dot {\hat \theta} = \hat \omega &= \omega_0 + m_p (p^{\star} - p) + m_\mathrm{dc} (v_\mathrm{dc} - v_\mathrm{dc}^{\star})\\
        \hat v &= v^{\star} + m_q (q^{\star} - q),
    \end{align}
\end{subequations}
which regulates the AC frequency, the AC voltage magnitude, and the DC voltage. In a particular case, $m_p = 0$, the control in \eqref{eq:dual-port-control} simplifies to a machine-matching control \cite{huang2017virtual}, which directly links the DC voltage to the AC frequency. This connection exemplifies a widely recognized observation that the DC voltage, similar to the frequency of a synchronous machine, indicates the power imbalance in power inverters. Analogous to \eqref{eq:dual-port-control-1}, a hybrid-angle control is developed in \cite{tayyebi2022grid}, where a nonlinear angle forming term is used instead of the linear active power droop term, rendering a rigorous large-signal stability guarantee.

We remark that all of the voltage-forming controls shown above can be generalized to the case of resistive-inductive networks, where the power feedback should be rotated according to the network impedance angle \cite{juan2009adaptive}.

Aside from the outer voltage-forming module offering a voltage reference, inner control modules are necessary to ensure fast reference tracking, disturbance rejection, and overcurrent protection. As seen from Fig.~\ref{fig:generic-control-diagram}, the inner control loops can be implemented in the stationary or rotational reference frame with proportional-resonant (PR) or proportional-integral (PI) regulators, respectively. We display the control implementation in complex vector coordinates such as $\pha{v}=v_{\alpha} + jv_{\beta}$ in the stationary reference frame. The counterpart in rotational $dq$ coordinates can be obtained analogously.

\subsection{Voltage Tracking or Virtual Admittance Control}

\textit{Voltage Tracking Control:} The primary objective of the voltage control module is voltage reference tracking. A typical voltage tracking control is given as
\begin{equation}
\label{eq:voltage-tracking}
    \pha{\hat i}^{+} = \left(k_\mathrm{p}^{\mathrm{v}} + k_\mathrm{r}^{\mathrm{v}} \frac{2 \omega_{\mathrm{v}}s}{s^{2 }+2\omega_{\mathrm{c}}s+\omega_{0}^{2}}\right)(\pha{\hat v} - \pha{v}^{+}),
\end{equation}
which represents a practical PR regulator with a bandwidth of the resonant filter $\omega_{\mathrm{v}}$ and control gains $k_\mathrm{p}^{\mathrm{v}}$ and $k_\mathrm{r}^{\mathrm{v}}$ \cite{zmood2003stationary}.

\textit{Virtual Admittance Control:} Instead of using PR or PI tracking regulators, another typical voltage control uses a virtual admittance as a proportional-like regulator \cite{fan2022equivalent,zhang2023current,saffar2023impacts}. A typical implementation is as follows
\cite{rosso2021implementation,kkuni2024effects,zhang2023current,fan2022equivalent,saffar2023impacts},
\begin{equation}
\label{eq:virtual-admittance}
    \pha{\hat i}^{+} = \frac{1}{r_\mathrm{v} + l_\mathrm{v}s} \left(\pha{\hat v} - \pha{v}^{+} \right)\ \mathrm{or}\ \pha{\hat i}^{+} = \frac{1}{r_\mathrm{v} + j x_\mathrm{v}} \left(\pha{\hat v} - \pha{v}^{+} \right)
\end{equation}
with virtual resistance $r_\mathrm{v}$ and virtual inductance $l_\mathrm{v}$. The regulator $\frac{1}{r_\mathrm{v} + l_\mathrm{v}s}$ functions as a dynamic virtual admittance. Alternatively, one can choose a static virtual admittance $\frac{1}{r_\mathrm{v} + j x_\mathrm{v}}$ \cite{zhang2023current}, or even a real-valued proportional voltage control \cite{fan2022equivalent}. While the PR or PI regulator tracks the voltage reference with zero errors in a steady state, the virtual admittance control behaves as a proportional gain. The proportional gain allows tracking errors while being free of the integrator windup issue during current saturation, which will be elaborated on later.

\subsection{Current Tracking Control}

With a positive-sequence current reference as given in \eqref{eq:voltage-tracking} and a negative-sequence current reference specified by \eqref{eq:negative-sequence-current-spec-flexible} or \eqref{eq:negative-sequence-current-spec4} (shown in Appendix~\ref{appendix:negative-sequence-current} later), the composite current reference is given as
\begin{equation}
\label{eq:current-reference}
    \pha{\hat i} = \pha{\hat i}^{+} + \pha{\hat i}^{-}.
\end{equation}
Analogous to the voltage tracking control, a current tracking control in the stationary reference frame is given as
\begin{equation}
\label{eq:current-control}
    \pha{\hat e} = \left(k_\mathrm{p}^{\mathrm{c}} + k_\mathrm{r}^{\mathrm{c}} \frac{2 \omega_{\mathrm{c}}s}{s^{2 }+2\omega_{\mathrm{c}}s+\omega_{0}^{2}}\right)(\pha{\overline i} - \pha{i}),
\end{equation}
which uses a practical PR regulator with a bandwidth of the resonant filter $\omega_{\mathrm{c}}$ and control gains $k_\mathrm{p}^{\mathrm{c}}$ and $k_\mathrm{r}^{\mathrm{c}}$ \cite{zmood2003stationary}. In \eqref{eq:current-control}, $\pha{\overline i}$ denotes the current reference saturated by the current limiter (introduced in Appendix~\ref{appendix:current-limiting}). When the current reference is not saturated, $\pha{\overline i}$ equals $\pha{\hat i}$.

\section{Negative-Sequence Current Specifications: \\A Review of Four Modes and Rigorous Proofs}
\label{appendix:negative-sequence-current}

% =======
% FIG
% =======
\begin{figure}
  \begin{center}
  \includegraphics{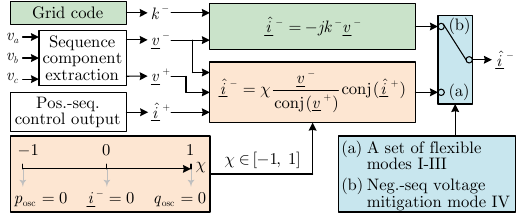}
  \caption{Multiple control modes for negative-sequence current components. (a) A set of flexible control modes to achieve similar objectives as for conventional grid-following devices (e.g., suppressing active power or reactive power oscillation \cite{jia2017review}); (b) Negative-sequence voltage mitigation mode (to satisfy grid code requirements) \cite{nasr2023controlling}.}
  \label{fig:neg-seq-curr-ref}
  \end{center}
\end{figure}

For the specification of negative-sequence current references under unbalanced grid conditions, the following four modes can be typically employed.

\subsection{Mode I: Balanced Current Control} The current of inverters is balanced when negative-sequence current components are absent. To achieve this, the negative-sequence current reference can be specified as zero, i.e.,
\begin{equation}
\label{eq:negative-sequence-current-spec1}
    \pha{\hat i}^{-} = 0.
\end{equation}
Without providing a negative-sequence current component, the negative-sequence circuit on the inverter side is open-circuit.

\subsection{Mode II: Active Power Oscillation Suppression} There is an oscillating component at twice the fundamental frequency in active and reactive power whenever current and voltage contain both positive- and negative-sequence components \cite{zheng2018flexible}. The oscillation in either active or reactive power can be eliminated by properly specifying the negative-sequence current component. Specifically, the active power oscillation will be eliminated if and only if the negative-sequence current is given as follows \cite{zheng2018flexible},
\begin{equation}
\label{eq:negative-sequence-current-spec2}
    \pha{\hat i}^{-} = -\frac{\pha{v}^{-}}{\mathrm{conj}(\pha{v}^{+})} \mathrm{conj}(\pha{\hat i}^{+}),
\end{equation}
where $\mathrm{conj}()$ indicates a conjugate operation. In \eqref{eq:negative-sequence-current-spec2}, the specification of the negative-sequence current reference relies on the positive-sequence current reference. The rigorous proof of \eqref{eq:negative-sequence-current-spec2} is given in Proposition~\ref{prop-pqosc}. 

\subsection{Mode III: Reactive Power Oscillation Suppression} Similarly, the reactive power oscillation will be eliminated if and only if the negative-sequence current is given as \cite{zheng2018flexible},
\begin{equation}
\label{eq:negative-sequence-current-spec3}
    \pha{\hat i}^{-} = \frac{\pha{v}^{-}}{\mathrm{conj}(\pha{v}^{+})} \mathrm{conj}(\pha{\hat i}^{+}).
\end{equation}
The proof of \eqref{eq:negative-sequence-current-spec3} is also provided in Proposition~\ref{prop-pqosc}.

We note that the sum of \eqref{eq:negative-sequence-current-spec2} and \eqref{eq:negative-sequence-current-spec3} is zero, which reduces to \eqref{eq:negative-sequence-current-spec1}. Hence, these modes can be synthesized as a set of \textit{flexible modes},
\begin{equation}
\label{eq:negative-sequence-current-spec-flexible}
    \pha{\hat i}^{-} = \chi \frac{\pha{v}^{-}}{\mathrm{conj}(\pha{v}^{+})} \mathrm{conj}(\pha{\hat i}^{+}), \quad \chi \in [-1,\ 1],
\end{equation}
where $\chi$ denotes a tunable parameter for achieving different control objectives, as illustrated in Fig.~\ref{fig:neg-seq-curr-ref}. We note that the formulation in \eqref{eq:negative-sequence-current-spec-flexible} applies to both $\alpha\beta$ coordinates and $dq$ coordinates; see \cite{zheng2018flexible} for the derivation in $dq$ coordinates.

\subsection{Mode IV: Negative-Sequence Voltage Mitigation} Alternatively, the negative-sequence current component can also be specified to suppress the negative-sequence voltage magnitude and thus improve the voltage unbalance factor (VUF). This is a requirement of grid codes. In particular, IEEE Std. 2800-2022 \cite{ieee2800} requires that all inverter-based resources absorb negative-sequence reactive current in a proportion of the negative-sequence voltage. In respect thereof, the negative-sequence \textit{output} current reference can be specified as
\begin{equation}
\label{eq:negative-sequence-current-spec4}
    \pha{\hat i}^{-} = -j k^{-} \pha{v}^{-},
\end{equation}
where $k^{-}$ is known as $K$-factor \cite{vde2017technical}. The complex coefficient $j k^{-}$ relating the negative-sequence voltage to the \textit{input} current (i.e., $-\pha{\hat i}^{-} = j k^{-} \pha{v}^{-}$) is equivalent to a virtual susceptance in the negative-sequence circuit, or equivalently regarded as a virtual reactance $\frac{1}{j k^{-}} = j(-\omega)\frac{1}{\omega k^{-}}$ with a negative-sequence frequency $-\omega$ and equivalent inductance $\frac{1}{\omega k^{-}}$. Therefore, this control mode contributes to reducing the negative-sequence voltage. A resistance/conductance component or a low-pass/band-pass filter can be incorporated into the coefficient in \eqref{eq:negative-sequence-current-spec4} to enhance dynamic performance.

We note that the control objectives of Modes I to III can be entirely fulfilled even when the positive- and negative-sequence current references are scaled down in the same ratio by a current limiter. For Mode IV, the current scaling ratio will reduce the $K$-factor equivalently. Furthermore, we indicate that the zero-sequence current component is inherently zero since the low-voltage-side zero-sequence circuit remains open due to the $\Delta$-Y configuration of the step-up transformer \cite{goksu2013impact}.

\begin{remark}
\textit{Differences in Negative-Sequence Control Between Grid-Forming and Grid-Following Architectures:}
The previous control objectives have been applied to a grid-following architecture for unbalanced grid conditions, which have been widely documented in the literature; see \cite{jia2017review} for a review. Compared to the grid-following architecture, the negative-sequence current reference specification in the grid-forming architecture is largely different regarding Modes II and III. This is because, in the grid-following architecture, both positive- and negative-sequence current references can be flexibly specified to suppress power oscillation. In contrast, in the grid-forming architecture, only the negative-sequence current reference can be flexibly specified while the positive-sequence current reference is governed in priority by the grid-forming control. It is shown in \cite{pola2023fault} that the power oscillation suppression may also admit grid-following-like schemes, e.g., by manipulating power references and current references multiple times. However, the grid-following-like scheme is rather complicated compared to the scheme in \eqref{eq:negative-sequence-current-spec-flexible}.
\end{remark}

\subsection{A Rigorous Proof of Modes II and III}

We denote positive-sequence and negative-sequence voltage and current components in complex vectors as follows: $\pha{v}^{+} \coloneqq v^{+} e^{j \theta_{\mathrm{v}}^{+}}$, $\pha{v}^{-} \coloneqq v^{-} e^{j \theta_{\mathrm{v}}^{-}}$, $\pha{i}^{+} \coloneqq i^{+} e^{j \theta_{\mathrm{c}}^{+}}$, and $\pha{i}^{-} \coloneqq i^{-} e^{j \theta_{\mathrm{c}}^{-}}$. The complex power $\pha{s}$ is then given as
\begin{equation}
\label{eq:complex-power}
\begin{aligned}
    \pha{s} &= (\pha{v}^{+} + \pha{v}^{-}) \left(\mathrm{conj}(\pha{i}^{+}) + \mathrm{conj}(\pha{i}^{-}) \right)\\
    &= \underbrace{ \pha{v}^{+} \mathrm{conj}(\pha{i}^{+}) + \pha{v}^{-} \mathrm{conj}(\pha{i}^{-}) }_{\pha{s}_{\mathrm{dc}}} + \underbrace{ \pha{v}^{+} \mathrm{conj}(\pha{i}^{-}) + \pha{v}^{-} \mathrm{conj}(\pha{i}^{+}) }_{\pha{s}_{\mathrm{osc}}},
\end{aligned}    
\end{equation}
with a dc component $\pha{s}_{\mathrm{dc}}$ and an oscillating component $\pha{s}_{\mathrm{osc}}$. The oscillating component is of concern, in which the active and reactive power components are represented as
\begin{equation}
\label{eq:pq-osc}
\begin{aligned}
    p_{\mathrm{osc}} &= \underbrace{ \Re \{\pha{v}^{+} \mathrm{conj}(\pha{i}^{-}) \} }_{ v^{+}i^{-} \cos{(\theta_{\mathrm{v}}^{+} - \theta_{\mathrm{c}}^{-})} } + \underbrace{ \Re \{\pha{v}^{-} \mathrm{conj}(\pha{i}^{+}) \} }_{ v^{-}i^{+} \cos{(\theta_{\mathrm{v}}^{-} - \theta_{\mathrm{c}}^{+})}}, \\
    q_{\mathrm{osc}} &= \underbrace{ \Im \{\pha{v}^{+} \mathrm{conj}(\pha{i}^{-}) \} }_{ v^{+}i^{-} \sin{(\theta_{\mathrm{v}}^{+} - \theta_{\mathrm{c}}^{-})} } + \underbrace{ \Im \{\pha{v}^{-} \mathrm{conj}(\pha{i}^{+}) \} }_{v^{-}i^{+} \sin{(\theta_{\mathrm{v}}^{-} - \theta_{\mathrm{c}}^{+})} }. \\
\end{aligned}    
\end{equation}

\begin{lemma}
\label{lem-pqosc}
$p_{\mathrm{osc}} = 0$ holds if and only if it holds that 
\begin{subequations}
\label{eq:posc-zero}
\begin{align}
\label{eq:posc1-zero}
    \Re \{\pha{v}^{+} \mathrm{conj}(\pha{i}^{-}) \} &= - \Re \{\pha{v}^{-} \mathrm{conj}(\pha{i}^{+}) \}, \\
\label{eq:posc2-zero}
    \Im \{\pha{v}^{+} \mathrm{conj}(\pha{i}^{-}) \} &= \Im \{\pha{v}^{-} \mathrm{conj}(\pha{i}^{+}) \}.
\end{align}
\end{subequations}
Similarly, $q_{\mathrm{osc}} = 0$ holds if and only if it holds that 
\begin{subequations}
\label{eq:qosc-zero}
\begin{align}
\label{eq:qosc1-zero}
    \Re \{\pha{v}^{+} \mathrm{conj}(\pha{i}^{-}) \} &= \Re \{\pha{v}^{-} \mathrm{conj}(\pha{i}^{+}) \}, \\
\label{eq:qosc2-zero}
    \Im \{\pha{v}^{+} \mathrm{conj}(\pha{i}^{-}) \} &= -\Im \{\pha{v}^{-} \mathrm{conj}(\pha{i}^{+}) \}.
\end{align}
\end{subequations}
\end{lemma}
\begin{proof}
The sufficiency is self-evident as shown in \eqref{eq:posc1-zero} and \eqref{eq:qosc1-zero}. We show below that $p_{\mathrm{osc}} = 0$ or $q_{\mathrm{osc}} = 0$ will also lead to \eqref{eq:posc2-zero} or \eqref{eq:qosc2-zero}, respectively.

The sine and cosine terms in \eqref{eq:pq-osc} are functions of time $t$ as they contain components of twice the fundamental frequency, i.e., $2 \omega t$. Therefore, $p_{\mathrm{osc}} = 0$, for any time $t$, leads to
\begin{equation}
\label{eq:posc-zero1}
    v^{+}i^{-} = v^{-}i^{+}, \quad \cos{(\theta_{\mathrm{v}}^{+} - \theta_{\mathrm{c}}^{-})} = -\cos{(\theta_{\mathrm{v}}^{-} - \theta_{\mathrm{c}}^{+})}.
\end{equation}
Similarly, $q_{\mathrm{osc}} = 0$, for any time $t$, leads to
\begin{equation}
\label{eq:qosc-zero1}
    v^{+}i^{-} = v^{-}i^{+}, \quad \sin{(\theta_{\mathrm{v}}^{+} - \theta_{\mathrm{c}}^{-})} = -\sin{(\theta_{\mathrm{v}}^{-} - \theta_{\mathrm{c}}^{+})}.
\end{equation}
We further show that $p_{\mathrm{osc}} = 0$ and $q_{\mathrm{osc}} = 0$, more specifically, \eqref{eq:posc-zero1} and \eqref{eq:qosc-zero1}, are mutually exclusive (cannot hold simultaneously). The proof by contradiction is as follows: \eqref{eq:posc-zero1} and \eqref{eq:qosc-zero1} will lead to $\tan{(\theta_{\mathrm{v}}^{+} - \theta_{\mathrm{c}}^{-})} = \tan{(\theta_{\mathrm{v}}^{-} - \theta_{\mathrm{c}}^{+})}$, and further, $\theta_{\mathrm{v}}^{+} - \theta_{\mathrm{c}}^{-} = \theta_{\mathrm{v}}^{-} - \theta_{\mathrm{c}}^{+} + k \pi,\ \forall k \in \mathbb{Z}$. This cannot hold for any time $t$ since the left-hand side is a function of $2 \omega t$ while the right-hand side is a function of $-2 \omega t$. Hence, an accompanying result of \eqref{eq:posc-zero1} is that $\sin{(\theta_{\mathrm{v}}^{+} - \theta_{\mathrm{c}}^{-})} = \sin{(\theta_{\mathrm{v}}^{-} - \theta_{\mathrm{c}}^{+})}$ while an accompanying result of \eqref{eq:qosc-zero1} is that $\cos{(\theta_{\mathrm{v}}^{+} - \theta_{\mathrm{c}}^{-})} = \cos{(\theta_{\mathrm{v}}^{-} - \theta_{\mathrm{c}}^{+})}$. This completes the proof of the necessity by recalling \eqref{eq:pq-osc}.
\end{proof}

\begin{proposition}
\label{prop-pqosc}
\textit{Necessary and Sufficient Conditions for Power Non-Oscillation:}
Given the positive-sequence current component $\pha{i}^{+}$, $p_{\mathrm{osc}} = 0$ holds if and only if it holds that
\begin{equation}
\label{eq:posc-zero-condition}
    \pha{i}^{-} = -\frac{\pha{v}^{-}}{\mathrm{conj}(\pha{v}^{+})} \mathrm{conj}(\pha{i}^{+});
\end{equation}
Moreover, $q_{\mathrm{osc}} = 0$ holds if and only if it holds that
\begin{equation}
\label{eq:qosc-zero-condition}
    \pha{i}^{-} = \frac{\pha{v}^{-}}{\mathrm{conj}(\pha{v}^{+})} \mathrm{conj}(\pha{i}^{+}).
\end{equation}
\end{proposition}

\begin{proof}
Consider $\pha{s}_{\mathrm{osc}} = \pha{v}^{+} \mathrm{conj}(\pha{i}^{-}) + \pha{v}^{-} \mathrm{conj}(\pha{i}^{+}) $ as in \eqref{eq:complex-power}. By leveraging Lemma~\ref{lem-pqosc}, it follows that $p_{\mathrm{osc}} = 0$ holds if and only if $\pha{v}^{+} \mathrm{conj}(\pha{i}^{-})$ and $\pha{v}^{-} \mathrm{conj}(\pha{i}^{+})$ have opposite real parts while the same imaginary parts. This is equivalent to the relationship that $\pha{v}^{+} \mathrm{conj}(\pha{i}^{-}) = -\mathrm{conj} \left[\pha{v}^{-} \mathrm{conj}(\pha{i}^{+})\right]$, which is further equivalent to \eqref{eq:posc-zero-condition}. In a similar vein, $q_{\mathrm{osc}} = 0$ holds if and only if $\pha{v}^{+} \mathrm{conj}(\pha{i}^{-})$ and $\pha{v}^{-} \mathrm{conj}(\pha{i}^{+})$ have the same real parts while opposite imaginary parts, as shown in Lemma~\ref{lem-pqosc}. This is equivalent to $\pha{v}^{+} \mathrm{conj}(\pha{i}^{-}) = \mathrm{conj} \left[\pha{v}^{-} \mathrm{conj}(\pha{i}^{+})\right]$ and further equivalent to \eqref{eq:qosc-zero-condition}.
\end{proof}

\section{Current-Limiting Strategies: \\A Categorized Review}
\label{appendix:current-limiting}

We consider typical existing current-limiting strategies and categorize them into three types to explicitly indicate three different technical routes to practical applications.

\subsection{Type A: Adaptive/Threshold Virtual Impedance Control}
\label{sec:virtual-impedance}

Virtual impedance control is motivated by the requirements of reshaping the network impedance characteristics to improve dynamic performance and power-sharing \cite{vijay2021adaptive} and by the need to limit the fault current during grid faults \cite{paquette2015virtual}. We introduce a fixed and an adaptive virtual impedance in the following.

A virtual impedance control module is explicitly employed to generate a voltage drop based on the current feedback \cite{paquette2015virtual},
\begin{equation}
\label{eq:virtual-impedance}
    \Delta \pha{\hat{v}}^{+} = \left ( r_{\mathrm{v}} + j x_{\mathrm{v}} \right) \pha{i}^{+},
\end{equation}
where $ r_{\mathrm{v}} + j x_{\mathrm{v}}$ denotes a virtual impedance. The voltage drop $\Delta \pha{\hat{v}}^{+}$ is then subtracted from the voltage reference $ \pha{\hat{v}}$.

A fixed virtual impedance cannot adapt to grid fault disturbances of different severity. To overcome this limitation, a current feedback-based adaptive virtual impedance has been proposed \cite{paquette2015virtual,qoria2020current}, which is arranged as follows,
\begin{equation}
\label{eq:adaptive-impedance}
    x_{\mathrm{v}} = \sigma_{\mathrm{vi}} r_{\mathrm{v}}, \  
    r_{\mathrm{v}} = 
    \begin{cases}
        0, & \abs{\pha{i}} \leq I_{\mathrm{th}}, \\
        \kappa_{\mathrm{vi}} \left(\abs{\pha{i}} - I_{\mathrm{th}}\right), & \abs{\pha{i}} > I_{\mathrm{th}},
    \end{cases}
\end{equation}
where $\sigma_{\mathrm{vi}}$ is a desired $X/R$ ratio, $\kappa_{\mathrm{vi}}$ is a proportional feedback gain, and $I_{\mathrm{th}}$ is a current-limiting threshold ($I_{\mathrm{th}} < I_{\lim}$). The choice of $\kappa_{\mathrm{vi}}$ is important to strictly limit the current magnitude under the maximum current. In the worst case, where a three-phase bolted fault is considered, such that $\pha{v} = 0$, the voltage-forming reference should be canceled by the virtual impedance voltage drop. Thus, $\kappa_{\mathrm{vi}}$ should satisfy that \cite{fan2022review}
\begin{equation}
\begin{aligned}
    \abs{\pha{\hat v}} \leq \abs{\pha{i}} \abs{r_{\mathrm{v}} + j x_{\mathrm{v}}} &= \abs{\pha{i}} \kappa_{\mathrm{vi}} \sqrt{\sigma_{\mathrm{vi}}^2 + 1} \bigl \vert \abs{\pha{i}} - I_{\mathrm{th}}\bigr \vert \\
    &\leq I_{\lim} \kappa_{\mathrm{vi}} \sqrt{\sigma_{\mathrm{vi}}^2 + 1} \left(I_{\lim} - I_{\mathrm{th}}\right), \\
    \Rightarrow \kappa_{\mathrm{vi}} &\geq \frac{\abs{\pha{\hat v}}}{I_{\lim} \sqrt{\sigma_{\mathrm{vi}}^2 + 1} \left(I_{\lim} - I_{\mathrm{th}}\right)}.
\end{aligned}
\end{equation}

The steady-state current settles in between the threshold $I_{\mathrm{th}}$ and the maximum current $I_{\lim}$ in most cases where the grid fault is not a bolted fault. This implies that the maximum overcurrent capability is underutilized. To overcome this problem, a separate proportional-integral control loop has been applied in \cite{rosso2021implementation} and \cite{nasr2023controlling} to correct or estimate the ongoing voltage drop across the virtual impedance. However, this necessitates further multi-loop interaction management and parameter tuning. Moreover, it has been indicated in \cite{taoufik2022variable} that the $X/R$ ratio, $\sigma_{\mathrm{vi}}$, may need to be adaptively adjusted to reach a compromise between the damping of the current response and the margin of transient stability.

The virtual impedance strategy has been extended to more general asymmetrical fault conditions \cite{baeckeland2022stationary,zhang2023simultaneous}. In respect thereof, the 2-norm $\abs{\pha{i}}$ cannot represent the maximum phase current magnitude. Instead, the current magnitude per phase should be detected, and the maximum phase current magnitude is used as the current feedback \cite{baeckeland2022stationary,zhang2023simultaneous}.

\subsection{Type B: Current Limiter With Virtual Admittance Control}
\label{sec:current-limiter}

Current limiters offer a more intuitive approach to current limiting than the indirect method of virtual impedance emulation. Various types of current limiters have been reported in the literature, with different implementations in different coordinates, e.g., $abc$ natural reference frame, $\alpha\beta$ stationary reference frame, and $dq$ synchronous
reference frame, as surveyed in \cite{zhang2021gridforming}. We recall the most commonly used current limiter in the following.

\subsubsection{Current Limiter for Balanced Conditions}
When the current reference is balanced, three-phase currents have the same magnitude, which is $\vert {\pha{\hat i}} \vert$. The magnitude-limited current reference, $\pha{\overline i}$, is then determined by a circular limiter as
\begin{gather}
\label{eq:current-limiter-bal}
    \pha{\overline i} =
    \begin{cases}
        \pha{\hat i}, & \lvert \pha{\hat i} \rvert \leq I_{\lim},\\
        \frac{I_{\lim}}{\lvert \pha{\hat i} \rvert} \pha{\hat i}, & \lvert \pha{\hat i} \rvert > I_{\lim}.
    \end{cases}
\end{gather}
The circular limiter is depicted in Fig.~\ref{fig:current-limiter-albe}(a). It also directly applies to a balanced current reference in $dq$ coordinates.

\subsubsection{Current Limiter for Unbalanced Conditions}
The magnitude limiting for an unbalanced current reference is not as direct as the balanced condition. The main difference is that three-phase currents have different magnitudes and the current vector, $\pha {\hat i} = \hat i_{\alpha} + j \hat i_{\beta}$, rotates according to an ellipse rather than a circle, as illustrated in Fig.~\ref{fig:current-limiter-albe}(b). The projection of the ellipse to each phase axis corresponds to the phase current. To avoid the overcurrent of any phase, the ellipse needs to be scaled down such that the maximum current magnitude is within the limit value. To do so, the magnitude of the phase currents needs to be identified, which is given as \cite{awal2023double}
\begin{equation}
\label{eq:phase-current-mag}
    \hat I_{x}^{\mathrm{m}} = \sqrt{\vert \pha{\hat i}^{+} \vert^2 + \vert \pha{\hat i}^{-} \vert^2 + 2 \Re \{\pha{\hat i}^{+} \pha{\hat i}^{-} e^{j2\lambda_x} \} },
\end{equation}
where $x \in \{a,b,c\}$ and $\lambda_x \in \{0, -2\pi/3, 2\pi/3\}$ respectively. Based on the phase current magnitudes, the elliptical current limiter is formulated as \cite{sadeghkhani2017current,zarei2019reniforcing}
\begin{gather}
\label{eq:current-limiter-unbal}
    \pha{\overline i} =
    \begin{cases}
        \pha{\hat i}, & \max \{\hat{I}_{a}^{\mathrm{m}},\hat{I}_{b}^{\mathrm{m}},\hat{I}_{c}^{\mathrm{m}} \} \leq I_{\lim},\\
        \frac{I_{\lim}}{\max \{\hat{I}_{a}^{\mathrm{m}},\hat{I}_{b}^{\mathrm{m}},\hat{I}_{c}^{\mathrm{m}} \}} \pha{\hat i}, & \max \{\hat{I}_{a}^{\mathrm{m}},\hat{I}_{b}^{\mathrm{m}},\hat{I}_{c}^{\mathrm{m}} \} > I_{\lim}.
    \end{cases}
\end{gather}
Since $\pha{\hat i} = \pha{\hat i}^{+} + \pha{\hat i}^{-}$, the positive- and negative-sequence current references are scaled down equally. We notice that the limiter in \eqref{eq:current-limiter-bal}, where $\lvert \pha{\hat i} \rvert = \hat{I}_{a}^{\mathrm{m}} = \hat{I}_{b}^{\mathrm{m}} = \hat{I}_{c}^{\mathrm{m}}$, is a special case of \eqref{eq:current-limiter-unbal}.

% =======
% FIG
% =======
\begin{figure}
  \begin{center}
  \includegraphics{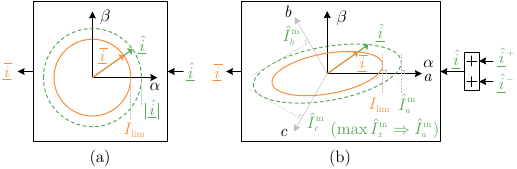}
  \caption{Illustration of current limiters in $\alpha\beta$ coordinates. (a) Circular current limiter for balanced conditions. (b) Elliptical current limiter for unbalanced conditions.}
  \label{fig:current-limiter-albe}
  \end{center}
\end{figure}

% =======
% FIG
% =======
\begin{figure}
  \begin{center}
  \includegraphics{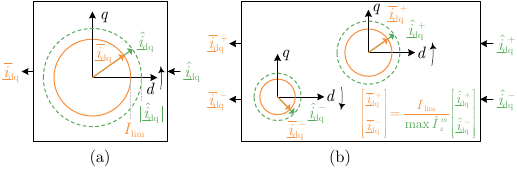}
  \caption{Illustration of current limiters in $dq$ coordinates. (a) Circular current limiter for balanced conditions. (b) Two circular current limiters (with equal scaling factors) for unbalanced conditions.}
  \label{fig:current-limiter-dq}
  \end{center}
\end{figure}

The operation of the current reference limiting in \eqref{eq:phase-current-mag} and \eqref{eq:current-limiter-unbal} can be extended to $dq$ coordinates, i.e.,
\begin{equation}
\label{eq:phase-current-mag-dq}
    \hat I_{x}^{\mathrm{m}} = \sqrt{\vert \pha{\hat i}_\mathrm{dq}^{+} \vert^2 + \vert \pha{\hat i}_\mathrm{dq}^{-} \vert^2 + 2 \Re \{\pha{\hat i}_\mathrm{dq}^{+} \pha{\hat i}_\mathrm{dq}^{-} e^{j2\lambda_x} \} },
\end{equation}
where $\pha{\hat i}_\mathrm{dq}^{+} \coloneqq e^{-j\hat \theta} \pha{\hat i}^{+}$ and $\pha{\hat i}_\mathrm{dq}^{-} \coloneqq e^{j\hat\theta} \pha{\hat i}^{-}$ imply that $\vert \pha{\hat i}^{+} \vert = \vert \pha{\hat i}_\mathrm{dq}^{+} \vert$, $\vert \pha{\hat i}^{-} \vert = \vert \pha{\hat i}_\mathrm{dq}^{-} \vert$, and $\pha{\hat i}_\mathrm{dq}^{+} \pha{\hat i}_\mathrm{dq}^{-} = \pha{\hat i}^{+} \pha{\hat i}^{-}$. Accordingly, the current limiter in $dq$ coordinates for unbalanced cases is given as
\begin{gather}
\label{eq:current-limiter-unbal-dq}
    \begin{bmatrix}
        \pha{\overline i}_\mathrm{dq}^{+} \\
        \pha{\overline i}_\mathrm{dq}^{-}
    \end{bmatrix} =
    \begin{cases}
        \begin{bmatrix}
            \pha{\hat i}_\mathrm{dq}^{+} \\
            \pha{\hat i}_\mathrm{dq}^{-}
        \end{bmatrix}, & \max \{\hat{I}_{a}^{\mathrm{m}},\hat{I}_{b}^{\mathrm{m}},\hat{I}_{c}^{\mathrm{m}} \} \leq I_{\lim},\\
        \frac{I_{\lim}}{\max \{\hat{I}_{a}^{\mathrm{m}},\hat{I}_{b}^{\mathrm{m}},\hat{I}_{c}^{\mathrm{m}} \}}
        \begin{bmatrix}
        \pha{\hat i}_\mathrm{dq}^{+} \\
        \pha{\hat i}_\mathrm{dq}^{-}
        \end{bmatrix}, & \max \{\hat{I}_{a}^{\mathrm{m}},\hat{I}_{b}^{\mathrm{m}},\hat{I}_{c}^{\mathrm{m}} \} > I_{\lim}.
    \end{cases}
\end{gather}
Since $\pha{\hat i}_\mathrm{dq}^{+}$ and $\pha{\hat i}_\mathrm{dq}^{-}$ are expressed in two different reference frames, the current limiter in \eqref{eq:current-limiter-unbal-dq} should be represented by two separate circular limiters, as illustrated in Fig.~\ref{fig:current-limiter-dq}.

\subsubsection{Virtual Admittance Serving for Anti-Windup}

When using integrator-included (PI or PR) voltage controllers, it is necessary to configure anti-windup along with a current limiter to avoid accumulating a significant control error in the integrator during current saturation. Generally, \textit{integrator clamping} (also known as conditional integration) and \textit{back-calculation} (also known as tracking integration) are two standard anti-windup methods used in industry for PI \cite{choi2009anti} and PR regulators \cite{ghoshal2010anti}.
\begin{itemize}
    \item The integrator clamping method disables the integrator whenever the output is saturated \cite{fan2022equivalent}. In addition, for a PR regulator, the resonant integrator should be reset to zero when disabled to avoid leaving a constant offset \cite{ghoshal2010anti}.
    \item The back-calculation method introduces a feedback loop from the portion of the output that exceeds the limiter \cite{ajala2021model,zhang2023active}, reshaping the PI regulator into a lead or lag filter, and the PR regulator a band-pass or band-stop filter.
\end{itemize}
Both anti-windup methods intentionally ``turn off" the integration during saturation, and afterward, the proportional regulator dominates the feedback control. In this sense, we link anti-windup with the virtual admittance control in \eqref{eq:virtual-admittance}, since a virtual admittance can be seen as a proportional-like regulator, which does not suffer from windup. Specifically, the virtual admittance control as in \eqref{eq:virtual-admittance} revises the voltage PR regulator in \eqref{eq:voltage-tracking} into an equivalent admittance, shown below again,
\begin{equation}
\label{eq:virtual-admittance-duplicate}
    \pha{\hat i}^{+} = \frac{1}{r_\mathrm{v} + l_\mathrm{v}s} \left(\pha{\hat v} - \pha{v}^{+} \right)\ \mathrm{or}\ \pha{\hat i}^{+} = \frac{1}{r_\mathrm{v} + j x_\mathrm{v}} \left(\pha{\hat v} - \pha{v}^{+} \right).
\end{equation}
Note that the virtual admittance control in \eqref{eq:virtual-admittance-duplicate} and the virtual impedance control in \eqref{eq:virtual-impedance} are equivalent in the steady state where the current is unsaturated, i.e., $\pha{i}^{+} = \pha{\hat i}^{+}$ and $\Delta \pha{\hat{v}}^{+} = \pha{\hat v} - \pha{v}^{+}$. However, when the current is saturated, both controls lead to different equivalent impedances; see Proposition~\ref{prop:unified-circuit}.

Other specific anti-windup schemes for voltage-forming inverters include: adjusting the outer-loop active and reactive power reference \cite{taul2020current}, limiting the outer-loop voltage and power reference \cite{chen2020use}, applying a virtual impedance to reduce the voltage reference \cite{zarei2019reniforcing}, moving the current limiter to the outer layer of the voltage-forming reference \cite{awal2023double,jiang2024current}, etc.

\subsection{Type C: Saturated-Current-Forming Control}

During current saturation, it is also possible to deactivate or bypass the voltage control and solely maintain the current control. Thus, this type of control falls into the scope of current-forming control, as described in Definition~\ref{def:current-forming} and Table~\ref{tab:gfm-operation-mode}. In this respect, the reference angle of the current vector control can still be generated by the original angle forming (frequency forming) loop, e.g., $\dot \theta = \omega_0 + k_p (p^{\ast}  - p)$, or switched to a PLL. Based on the latter, many recent studies have explored directly specifying current references (in $dq$ coordinates) during current saturation \cite{huang2019transient,rokrok2022transient,liu2023dynamic,li2023transient,wang2023transient}, where various stabilizing remedies have been developed, such as using the q-axis voltage feedback \cite{huang2019transient} and adjusting the current reference angle \cite{rokrok2022transient,liu2023dynamic,li2023transient}. Since this control falls into the current-forming type, the voltage-forming functionality cannot be provided independently. To achieve this, additional outer control loops are needed \cite{xin2021dual,schweizer2022grid}, but they may suffer from limited control bandwidth.

\subsection{Unified Equivalent Circuit for Type-A and -B Strategies}
\label{sec:unified-circuit}

% =======
% FIG
% =======
\begin{figure}
  \begin{center}
  \includegraphics{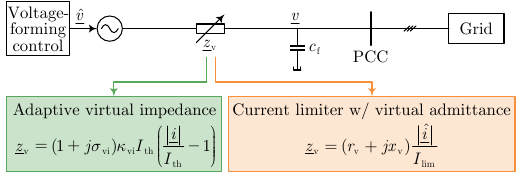}
  \caption{A unified equivalent circuit representation \cite{wu2024design} for both the adaptive virtual impedance control \cite{paquette2015virtual,qoria2020current} and the current limiter with virtual admittance control \cite{rosso2021implementation,kkuni2024effects,zhang2023current,fan2022equivalent,saffar2023impacts}, where the resulting equivalent impedance is current-dependent for both current-limiting strategies.}
  \label{fig:unified-circuit}
  \end{center}
\end{figure}

A unified equivalent circuit has been established to describe the output behavior of voltage-forming inverters with an activated adaptive virtual impedance or an activated current limiter with a virtual admittance \cite{wu2024design}. For a balanced condition, the unified circuit is reformulated precisely in the following Proposition~\ref{prop:unified-circuit}. The result can also be extended to an unbalanced condition, where the circuit will refer to the positive-sequence domain.

\begin{proposition}
\label{prop:unified-circuit}
\textit{Unified and Current-Dependent Equivalent Circuit:}
Consider a voltage-forming inverter under current saturation, where either the adaptive virtual impedance in \eqref{eq:adaptive-impedance} or the current limiter in \eqref{eq:current-limiter-bal} alongside the virtual admittance in \eqref{eq:virtual-admittance-duplicate} is activated. The output behavior of the voltage-forming inverter can be represented by a unified equivalent circuit in Fig.~\ref{fig:unified-circuit}, where the equivalent virtual impedance $\pha z_{\mathrm{v}}$ is dependent on the current $\pha{i}$ or its reference $\pha{\hat i}$, respectively, as follows,
\begin{align}
\label{eq:zv1}
    \pha z_{\mathrm{v}} &= (1 + j\sigma_{\mathrm{vi}}) \kappa_{\mathrm{vi}} I_{\mathrm{th}} \left( \frac{\abs{\pha{i}}}{I_{\mathrm{th}}} - 1\right), \\
\label{eq:zv2}
    \pha z_{\mathrm{v}} &= \left(r_\mathrm{v} + j x_\mathrm{v}\right) \frac{\vert\pha{\hat i} \vert}{I_{\mathrm{\lim}}},
\end{align}
and the currents $\abs{\pha{i}}$ and $\vert\pha{\hat i} \vert$ are dependent on both the state of the outer voltage-forming controller and the exogenous circuit.
\end{proposition}
\begin{proof}
First, we consider the adaptive virtual impedance in \eqref{eq:adaptive-impedance}, which is derived as
\begin{equation}
    \begin{aligned}
        \pha z_{\mathrm{v}} = r_{\mathrm{v}} + jx_{\mathrm{v}} &= (1 + j\sigma_{\mathrm{vi}}) \kappa_{\mathrm{vi}} \left( \abs{\pha{i}} - I_{\mathrm{th}} \right) \\
        &= (1 + j\sigma_{\mathrm{vi}}) \kappa_{\mathrm{vi}} I_{\mathrm{th}} \left( \frac{\abs{\pha{i}}}{I_{\mathrm{th}}} - 1\right).
    \end{aligned}
\end{equation}
This completes the proof of \eqref{eq:zv1}. Next, we consider the current limiter in \eqref{eq:current-limiter-bal} and the virtual admittance in \eqref{eq:virtual-admittance-duplicate} (with a static admittance) for balanced conditions and arrive at
\begin{equation}
    \begin{aligned}
        \pha{\hat v} - \pha{v} = \pha{\hat i} \left(r_\mathrm{v} + j x_\mathrm{v} \right) &= \frac{\lvert \pha{\hat i} \rvert}{I_{\lim}} \left(r_\mathrm{v} + j x_\mathrm{v} \right) \pha{\overline i} \\
        &= \underbrace{\frac{\lvert \pha{\hat i} \rvert}{I_{\lim}} \left(r_\mathrm{v} + j x_\mathrm{v} \right)}_{\pha z_{\mathrm{v}}} \pha{i},
    \end{aligned}
\end{equation}
where it is assumed that the output current $\pha{i}$ tracks the saturated reference $\pha{\overline i}$. The proof of \eqref{eq:zv2} is completed. Finally, since $\abs{\pha{i}}$ in \eqref{eq:zv1} tracks its reference $\vert\pha{\hat i} \vert$, and the current references $\vert\pha{\hat i} \vert$ in both types of strategies are regulated with the voltage-forming reference input $\pha {\hat v}$ and the voltage measurement feedback $\pha v$, the dependence of the currents on the voltage-forming controller state and the exogenous circuit is evident.
\end{proof}

The fact that the resulting virtual impedance for current limiting is current-dependent is reasonable since a larger equivalent impedance is required to limit a higher overcurrent. On the other hand, this implies that the equivalent virtual impedance is unpredictable, depending on the fault current which is influenced by the fault severity. It has been found in \cite{fan2022equivalent,zhang2023current,saffar2023impacts} that the current-dependent equivalent impedance leads to a distorted power-angle relationship. Moreover, the unpredictable virtual impedance brings about significant difficulties in transient stability analysis and insufficient robustness to unforeseen fault disturbances.

The unified equivalent circuit establishes a connection between the two independently developed current-limiting strategies. However, this does not suggest that both strategies provide the same performance. The most significant difference is that the saturated current, while using the adaptive virtual impedance, is in between $I_{\mathrm{th}}$ and $I_{\lim}$, whereas the saturated current while using the current limiter is $I_{\lim}$. More comparisons can be found in Table~\ref{tab:comparison-current-limiting-strategies}.

% ====== REFERENCE SECTION
\bibliographystyle{IEEEtran}
\bibliography{IEEEabrv,Bibliography}

\begin{thebibliography}{10}
\providecommand{\url}[1]{#1}
\csname url@rmstyle\endcsname
\providecommand{\newblock}{\relax}
\providecommand{\bibinfo}[2]{#2}
\providecommand\BIBentrySTDinterwordspacing{\spaceskip=0pt\relax}
\providecommand\BIBentryALTinterwordstretchfactor{4}
\providecommand\BIBentryALTinterwordspacing{\spaceskip=\fontdimen2\font plus
\BIBentryALTinterwordstretchfactor\fontdimen3\font minus \fontdimen4\font\relax}
\providecommand\BIBforeignlanguage[2]{{%
\expandafter\ifx\csname l@#1\endcsname\relax
\typeout{** WARNING: IEEEtran.bst: No hyphenation pattern has been}%
\typeout{** loaded for the language `#1'. Using the pattern for}%
\typeout{** the default language instead.}%
\else
\language=\csname l@#1\endcsname
\fi
#2}}
\renewcommand\BIBentryALTinterwordstretchfactor{4}

\bibitem{gb2024}
\BIBentryALTinterwordspacing
{National Grid ESO}, ``The grid code,'' National Grid Electricity System Operator, Tech. Rep., 2024. [Online]. Available: \url{https://www.neso.energy/document/287271/download}
\BIBentrySTDinterwordspacing

\bibitem{aemo2023voluntary}
\BIBentryALTinterwordspacing
{AEMO}, ``Voluntary specification for grid-forming inverters,'' Australian Energy Market Operator (AEMO), Tech. Rep., 2023. [Online]. Available: \url{https://aemo.com.au/-/media/files/initiatives/primary-frequency-response/2023/gfm-voluntary-spec.pdf}
\BIBentrySTDinterwordspacing

\bibitem{entso2024}
\BIBentryALTinterwordspacing
{ENTSO-E}, ``Grid forming capability of power park modules,'' European network of transmission system operators for electricity (ENTSO-E), Tech. Rep., 2024. [Online]. Available: \url{https://eepublicdownloads.entsoe.eu/clean-documents/Publications/SOC/20240503_First_interim_report_in_technical_requirements.pdf}
\BIBentrySTDinterwordspacing

\bibitem{bahrani2024grid}
B.~Bahrani, M.~H. Ravanji, B.~Kroposki, D.~Ramasubramanian, X.~Guillaud, T.~Prevost, and N.-A. Cutululis, ``Grid-forming inverter-based resource research landscape: Understanding the key assets for renewable-rich power systems,'' \emph{{IEEE} Power Energy Mag.}, vol.~22, no.~2, pp. 18--29, 2024.

\bibitem{ghimire2024functional}
S.~Ghimire, G.~M. Guerreiro, K.~Vatta~Kkuni, E.~D. Guest, K.~H. Jensen, G.~Yang, and X.~Wang, ``Functional specifications and testing requirements of grid-forming offshore wind power plants,'' \emph{Wind Energy Science Discussions}, vol. 2024, pp. 1--21, 2024.

\bibitem{baeckeland2024overcurrent}
N.~Baeckeland, D.~Chatterjee, M.~Lu, B.~Johnson, and G.-S. Seo, ``Overcurrent limiting in grid-forming inverters: A comprehensive review and discussion,'' \emph{{IEEE} Trans. Power Electron.}, vol.~39, no.~11, pp. 14\,493--14\,517, 2024.

\bibitem{ordono2024current}
A.~Ordono, A.~Sanchez-Ruiz, M.~Zubiaga, F.~J. Asensio, and J.~A. Cortajarena, ``Current limiting strategies for grid forming inverters under low voltage ride through,'' \emph{Renew. Sust. Energ. Rev.}, vol. 202, p. 114657, 2024.

\bibitem{pan2020transient}
D.~Pan, X.~Wang, F.~Liu, and R.~Shi, ``Transient stability of voltage-source converters with grid-forming control: A design-oriented study,'' \emph{{IEEE} J. Emerg. Sel. Top. Power Electron.}, vol.~8, no.~2, pp. 1019--1033, 2020.

\bibitem{rosso2021gridforming}
R.~Rosso, X.~Wang, M.~Liserre, X.~Lu, and S.~Engelken, ``Grid-forming converters: Control approaches, grid-synchronization, and future trends—{A} review,'' \emph{{IEEE} Open J. Ind. Appl.}, vol.~2, pp. 93--109, 2021.

\bibitem{zhang2021gridforming}
H.~Zhang, W.~Xiang, W.~Lin, and J.~Wen, ``Grid forming converters in renewable energy sources dominated power grid: Control strategy, stability, application, and challenges,'' \emph{J. Mod. Power Syst. Clean Energy}, vol.~9, no.~6, pp. 1239--1256, 2021.

\bibitem{dorfler2023control}
F.~D{\"o}rfler and D.~Gro{\ss}, ``Control of low-inertia power systems,'' \emph{Annu. Rev. Control Robot. Auton. Syst.}, vol.~6, pp. 415--445, 2023.

\bibitem{paquette2015virtual}
A.~D. Paquette and D.~M. Divan, ``Virtual impedance current limiting for inverters in microgrids with synchronous generators,'' \emph{{IEEE} Trans. Ind. Appl.}, vol.~51, no.~2, pp. 1630--1638, 2015.

\bibitem{qoria2020current}
T.~Qoria, F.~Gruson, F.~Colas, X.~Kestelyn, and X.~Guillaud, ``Current limiting algorithms and transient stability analysis of grid-forming {VSCs},'' \emph{Electr. Power Syst. Res.}, vol. 189, p. 106726, 2020.

\bibitem{rosso2021implementation}
R.~Rosso, S.~Engelken, and M.~Liserre, ``On the implementation of an {FRT} strategy for grid-forming converters under symmetrical and asymmetrical grid faults,'' \emph{{IEEE} Trans. Ind. Appl.}, vol.~57, no.~5, pp. 4385--4397, 2021.

\bibitem{kkuni2024effects}
K.~V. Kkuni and G.~Yang, ``Effects of current limit for grid forming converters on transient stability: analysis and solution,'' \emph{Int. J. Electr. Power Energy Syst.}, vol. 158, p. 109919, 2024.

\bibitem{fan2022equivalent}
B.~Fan and X.~Wang, ``Equivalent circuit model of grid-forming converters with circular current limiter for transient stability analysis,'' \emph{{IEEE} Trans. Power Syst.}, vol.~37, no.~4, pp. 3141--3144, 2022.

\bibitem{zhang2023current}
Y.~Zhang, C.~Zhang, R.~Yang, M.~Molinas, and X.~Cai, ``Current-constrained power-angle characterization method for transient stability analysis of grid-forming voltage source converters,'' \emph{{IEEE} Trans. Energy Convers.}, vol.~38, no.~2, pp. 1338--1349, 2023.

\bibitem{saffar2023impacts}
K.~G. Saffar, S.~Driss, and F.~B. Ajaei, ``Impacts of current limiting on the transient stability of the virtual synchronous generator,'' \emph{{IEEE} Trans. Power Electron.}, vol.~38, no.~2, pp. 1509--1521, 2023.

\bibitem{huang2019transient}
L.~Huang, H.~Xin, Z.~Wang, L.~Zhang, K.~Wu, and J.~Hu, ``Transient stability analysis and control design of droop-controlled voltage source converters considering current limitation,'' \emph{{IEEE} Trans. Smart Grid}, vol.~10, no.~1, pp. 578--591, 2019.

\bibitem{rokrok2022transient}
E.~Rokrok, T.~Qoria, A.~Bruyere, B.~Francois, and X.~Guillaud, ``Transient stability assessment and enhancement of grid-forming converters embedding current reference saturation as current limiting strategy,'' \emph{{IEEE} Trans. Power Syst.}, vol.~37, no.~2, pp. 1519--1531, 2022.

\bibitem{liu2023dynamic}
Y.~Liu, H.~Geng, M.~Huang, and X.~Zha, ``Dynamic current limiting of grid-forming converters for transient synchronization stability enhancement,'' \emph{{IEEE} Trans. Ind. Appl.}, vol.~60, no.~2, pp. 2238--2248, 2024.

\bibitem{li2023transient}
Y.~Li, Y.~Lu, J.~Yang, X.~Yuan, R.~Yang, S.~Yang, H.~Ye, and Z.~Du, ``Transient stability of power synchronization loop based grid forming converter,'' \emph{{IEEE} Trans. Energy Convers.}, vol.~38, no.~4, pp. 2843--2859, 2023.

\bibitem{wang2023transient}
G.~Wang, L.~Fu, Q.~Hu, C.~Liu, and Y.~Ma, ``Transient synchronization stability of grid-forming converter during grid fault considering transient switched operation mode,'' \emph{{IEEE} Trans. Sustain. Energy}, vol.~14, no.~3, pp. 1504--1515, 2023.

\bibitem{xin2021dual}
H.~Xin, K.~Zhuang, P.~Hu, Y.~Gu, and P.~Ju, ``Dual synchronous generator: Inertial current source based grid-forming solution for {VSC},'' \emph{arXiv preprint arXiv:2107.01805}, 2021.

\bibitem{schweizer2022grid}
M.~Schweizer, S.~Alm{\'e}r, S.~Pettersson, A.~Merkert, V.~Bergemann, and L.~Harnefors, ``Grid-forming vector current control,'' \emph{{IEEE} Trans. Power Electron.}, vol.~37, no.~11, pp. 13\,091--13\,106, 2022.

\bibitem{liu2022current}
T.~Liu, X.~Wang, F.~Liu, K.~Xin, and Y.~Liu, ``A current limiting method for single-loop voltage-magnitude controlled grid-forming converters during symmetrical faults,'' \emph{{IEEE} Trans. Power Electron.}, vol.~37, no.~4, pp. 4751--4763, 2022.

\bibitem{qoria2020critical}
T.~Qoria, F.~Gruson, F.~Colas, G.~Denis, T.~Prevost, and X.~Guillaud, ``Critical clearing time determination and enhancement of grid-forming converters embedding virtual impedance as current limitation algorithm,'' \emph{{IEEE} J. Emerg. Sel. Top. Power Electron.}, vol.~8, no.~2, pp. 1050--1061, 2020.

\bibitem{li2022revisiting}
Y.~Li, Y.~Gu, and T.~C. Green, ``Revisiting grid-forming and grid-following inverters: A duality theory,'' \emph{{IEEE} Trans. Power Syst.}, vol.~37, no.~6, pp. 4541--4554, 2022.

\bibitem{zhang2024control}
Y.~Zhang, C.~Zhang, M.~Molinas, and X.~Cai, ``Control of virtual synchronous generator with improved transient angle stability under symmetric and asymmetric short circuit fault,'' \emph{{IEEE} Trans. Energy Convers.}, pp. 1--18, 2024.

\bibitem{awal2023double}
M.~A. Awal, M.~R.~K. Rachi, H.~Yu, I.~Husain, and S.~Lukic, ``Double synchronous unified virtual oscillator control for asymmetrical fault ride-through in grid-forming voltage source converters,'' \emph{{IEEE} Trans. Power Electron.}, vol.~38, no.~6, pp. 6759--6763, 2023.

\bibitem{zheng2018flexible}
T.~Zheng, L.~Chen, Y.~Guo, and S.~Mei, ``Flexible unbalanced control with peak current limitation for virtual synchronous generator under voltage sags,'' \emph{J. Mod. Power Syst. Clean Energy}, vol.~6, no.~1, pp. 61--72, 2018.

\bibitem{nasr2023controlling}
M.-A. Nasr and A.~Hooshyar, ``Controlling grid-forming inverters to meet the negative-sequence current requirements of the {IEEE} {Standard} 2800-2022,'' \emph{{IEEE} Trans. Power Del.}, vol.~38, no.~4, pp. 2541--2555, 2023.

\bibitem{ieee2800}
{IEEE Standards Association}, ``{IEEE} standard for interconnection and interoperability of inverter-based resources {(IBRs)} interconnecting with associated transmission electric power systems,'' \emph{IEEE Std. 2800-2022}, pp. 1--180, 2022.

\bibitem{he2019resynchronization}
X.~He, H.~Geng, J.~Xi, and J.~M. Guerrero, ``Resynchronization analysis and improvement of grid-connected {VSCs} during grid faults,'' \emph{{IEEE} J. Emerg. Sel. Top. Power Electron.}, vol.~9, no.~1, pp. 438--450, 2021.

\bibitem{jia2017review}
J.~Jia, G.~Yang, and A.~H. Nielsen, ``A review on grid-connected converter control for short-circuit power provision under grid unbalanced faults,'' \emph{{IEEE} Trans. Power Del.}, vol.~33, no.~2, pp. 649--661, 2018.

\bibitem{rezazadeh2023single}
H.~Rezazadeh, M.~Monfared, M.~Fazeli, and S.~Golestan, ``Single-phase grid-forming inverters: A review,'' in \emph{2023 international conference on computing, electronics \& communications Engineering (ICCECE)}.\hskip 1em plus 0.5em minus 0.4em\relax IEEE, 2023, pp. 7--10.

\bibitem{sadeghkhani2017current}
I.~Sadeghkhani, M.~E. Hamedani~Golshan, J.~M. Guerrero, and A.~Mehrizi-Sani, ``A current limiting strategy to improve fault ride-through of inverter interfaced autonomous microgrids,'' \emph{{IEEE} Trans. Smart Grid}, vol.~8, no.~5, pp. 2138--2148, 2017.

\bibitem{he2022synchronization}
X.~He, C.~He, S.~Pan, H.~Geng, and F.~Liu, ``Synchronization instability of inverter-based generation during asymmetrical grid faults,'' \emph{{IEEE} Trans. Power Syst.}, vol.~37, no.~2, pp. 1018--1031, 2022.

\bibitem{wu2022smallsignal}
H.~Wu and X.~Wang, ``Small-signal modeling and controller parameters tuning of grid-forming {VSCs} with adaptive virtual impedance-based current limitation,'' \emph{{IEEE} Trans. Power Electron.}, vol.~37, no.~6, pp. 7185--7199, 2022.

\bibitem{taoufik2022variable}
T.~Qoria, H.~Wu, X.~Wang, and I.~Colak, ``Variable virtual impedance-based overcurrent protection for grid-forming inverters: Small-signal, large-signal analysis and improvement,'' \emph{{IEEE} Trans. Smart Grid}, vol.~14, no.~5, pp. 3324--3336, 2023.

\bibitem{alipoor2014power}
J.~Alipoor, Y.~Miura, and T.~Ise, ``Power system stabilization using virtual synchronous generator with alternating moment of inertia,'' \emph{{IEEE} J. Emerg. Sel. Top. Power Electron.}, vol.~3, no.~2, pp. 451--458, 2015.

\bibitem{wu2020mode}
H.~Wu and X.~Wang, ``A mode-adaptive power-angle control method for transient stability enhancement of virtual synchronous generators,'' \emph{{IEEE} J. Emerg. Sel. Top. Power Electron.}, vol.~8, no.~2, pp. 1034--1049, 2020.

\bibitem{taul2020current}
M.~G. Taul, X.~Wang, P.~Davari, and F.~Blaabjerg, ``Current limiting control with enhanced dynamics of grid-forming converters during fault conditions,'' \emph{{IEEE} J. Emerg. Sel. Top. Power Electron.}, vol.~8, no.~2, pp. 1062--1073, 2020.

\bibitem{wang2024active}
J.~Wang and X.~Zhang, ``Active power and voltage cooperative control for improving fault ride-through capability of grid-forming converters,'' \emph{{IEEE} Trans. Ind. Electron.}, vol.~71, no.~10, pp. 12\,301--12\,311, 2024.

\bibitem{he2023quantitative}
X.~He, L.~Huang, I.~Subotić, V.~Häberle, and F.~Dörfler, ``Quantitative stability conditions for grid-forming converters with complex droop control,'' \emph{{IEEE} Trans. Power Electron.}, vol.~39, no.~9, pp. 10\,834--10\,852, 2024.

\bibitem{he2024passivity}
X.~He and F.~D{\"o}rfler, ``Passivity and decentralized stability conditions for grid-forming converters,'' \emph{{IEEE} Trans. Power Syst.}, vol.~39, no.~3, pp. 5447--5450, 2024.

\bibitem{colombino2019global}
M.~Colombino, D.~Groß, J.-S. Brouillon, and F.~Dörfler, ``Global phase and magnitude synchronization of coupled oscillators with application to the control of grid-forming power inverters,'' \emph{{IEEE} Trans. Autom. Control}, vol.~64, no.~11, pp. 4496--4511, 2019.

\bibitem{schiffer2014conditions}
J.~Schiffer, R.~Ortega, A.~Astolfi, J.~Raisch, and T.~Sezi, ``Conditions for stability of droop-controlled inverter-based microgrids,'' \emph{Automatica}, vol.~50, no.~10, pp. 2457--2469, 2014.

\bibitem{choopani2020newmulti}
M.~Choopani, S.~H. Hosseinian, and B.~Vahidi, ``New transient stability and {LVRT} improvement of multi-{VSG} grids using the frequency of the center of inertia,'' \emph{{IEEE} Trans. Power Syst.}, vol.~35, no.~1, pp. 527--538, 2020.

\bibitem{shuai2019transient}
Z.~Shuai, C.~Shen, X.~Liu, Z.~Li, and Z.~J. Shen, ``Transient angle stability of virtual synchronous generators using {Lyapunov}’s direct method,'' \emph{{IEEE} Trans. Smart Grid}, vol.~10, no.~4, pp. 4648--4661, 2019.

\bibitem{kabalan2017largesignal}
M.~Kabalan, P.~Singh, and D.~Niebur, ``Large signal {Lyapunov}-based stability studies in microgrids: A review,'' \emph{{IEEE} Trans. Smart Grid}, vol.~8, no.~5, pp. 2287--2295, 2017.

\bibitem{fu2021large}
X.~Fu, J.~Sun, M.~Huang, Z.~Tian, H.~Yan, H.~H.-C. Iu, P.~Hu, and X.~Zha, ``Large-signal stability of grid-forming and grid-following controls in voltage source converter: A comparative study,'' \emph{{IEEE} Trans. Power Electron.}, vol.~36, no.~7, pp. 7832--7840, 2021.

\bibitem{kundur1994power}
P.~Kundur, \emph{Power System Stability and Control}.\hskip 1em plus 0.5em minus 0.4em\relax New York, NY, USA: McGraw Hill, 1994.

\bibitem{li2023iterative}
X.~Li, Z.~Tian, X.~Zha, P.~Sun, Y.~Hu, and M.~Huang, ``An iterative equal area criterion for transient stability analysis of grid-tied converter systems with varying damping,'' \emph{{IEEE} Trans. Power Syst.}, vol.~39, no.~1, pp. 1771--1784, 2024.

\bibitem{moon2000estimating}
Y.-H. Moon, B.-K. Choi, and T.-H. Roh, ``Estimating the domain of attraction for power systems via a group of damping-reflected energy functions,'' \emph{Automatica}, vol.~36, no.~3, pp. 419--425, 2000.

\bibitem{desai2023saturation}
M.~A. Desai, X.~He, L.~Huang, and F.~D\"{o}rfler, ``Saturation-informed current-limiting control for grid-forming converters,'' \emph{Electr. Power Syst. Res.}, vol. 234, p. 110746, 2024.

\bibitem{milano2022complex}
F.~Milano, ``Complex frequency,'' \emph{{IEEE} Trans. Power Syst.}, vol.~37, no.~2, pp. 1230--1240, 2022.

\bibitem{subotic2022power}
I.~Suboti{\'c} and D.~Gro{\ss}, ``Power-balancing dual-port grid-forming power converter control for renewable integration and hybrid {AC/DC} power systems,'' \emph{{IEEE} Trans. Control Netw. Syst.}, vol.~9, no.~4, pp. 1949--1961, 2022.

\bibitem{huang2017virtual}
L.~Huang, H.~Xin, Z.~Wang, K.~Wu, H.~Wang, J.~Hu, and C.~Lu, ``A virtual synchronous control for voltage-source converters utilizing dynamics of {DC}-link capacitor to realize self-synchronization,'' \emph{{IEEE} J. Emerg. Sel. Top. Power Electron.}, vol.~5, no.~4, pp. 1565--1577, 2017.

\bibitem{tayyebi2022grid}
A.~Tayyebi, A.~Anta, and F.~Dörfler, ``Grid-forming hybrid angle control and almost global stability of the {DC-AC} power converter,'' \emph{{IEEE} Trans. Autom. Control}, vol.~68, no.~7, pp. 3842--3857, 2023.

\bibitem{juan2009adaptive}
J.~C. Vasquez, J.~M. Guerrero, A.~Luna, P.~Rodriguez, and R.~Teodorescu, ``Adaptive droop control applied to voltage-source inverters operating in grid-connected and islanded modes,'' \emph{{IEEE} Trans. Ind. Electron.}, vol.~56, no.~10, pp. 4088--4096, 2009.

\bibitem{zmood2003stationary}
D.~Zmood and D.~Holmes, ``Stationary frame current regulation of {PWM} inverters with zero steady-state error,'' \emph{{IEEE} Trans. Power Electron.}, vol.~18, no.~3, pp. 814--822, 2003.

\bibitem{vde2017technical}
VDE/FNN, ``Technical requirements for the connection and operation of customer installations to the high voltage network ({TAR} high voltage),'' 2017.

\bibitem{goksu2013impact}
{\"O}.~G{\"o}ksu, R.~Teodorescu, C.~L. Bak, F.~Iov, and P.~Carne~Kj{\ae}r, ``Impact of wind power plant reactive current injection during asymmetrical grid faults,'' \emph{{IET} Renew. Power Gener.}, vol.~7, no.~5, pp. 484--492, 2013.

\bibitem{pola2023fault}
S.~Pola, M.~Azzouz, A.~S. Awad, and H.~Sindi, ``Fault ride-through strategies for synchronverter-interfaced energy resources under asymmetrical grid faults,'' \emph{{IEEE} Trans. Sustain. Energy}, vol.~14, no.~4, pp. 2391--2405, 2023.

\bibitem{vijay2021adaptive}
A.~S. Vijay, N.~Parth, S.~Doolla, and M.~C. Chandorkar, ``An adaptive virtual impedance control for improving power sharing among inverters in islanded ac microgrids,'' \emph{{IEEE} Trans. Smart Grid}, vol.~12, no.~4, pp. 2991--3003, 2021.

\bibitem{fan2022review}
B.~Fan, T.~Liu, F.~Zhao, H.~Wu, and X.~Wang, ``A review of current-limiting control of grid-forming inverters under symmetrical disturbances,'' \emph{{IEEE} Open J. Power Electron.}, vol.~3, pp. 955--969, 2022.

\bibitem{baeckeland2022stationary}
N.~Baeckeland, D.~Venkatramanan, M.~Kleemann, and S.~Dhople, ``Stationary-frame grid-forming inverter control architectures for unbalanced fault-current limiting,'' \emph{{IEEE} Trans. Energy Convers.}, vol.~37, no.~4, pp. 2813--2825, 2022.

\bibitem{zhang2023simultaneous}
H.~Zhang, R.~Liu, C.~Xue, and Y.~Li, ``Simultaneous overvoltage and overcurrent mitigation strategy of grid-forming inverters under a single-line-to-ground fault,'' \emph{{IEEE} Trans. Ind. Electron.}, vol.~71, no.~9, pp. 10\,818--10\,830, 2024.

\bibitem{zarei2019reniforcing}
S.~F. Zarei, H.~Mokhtari, M.~A. Ghasemi, and F.~Blaabjerg, ``Reinforcing fault ride through capability of grid forming voltage source converters using an enhanced voltage control scheme,'' \emph{{IEEE} Trans. Power Del.}, vol.~34, no.~5, pp. 1827--1842, 2019.

\bibitem{choi2009anti}
J.-W. Choi and S.-C. Lee, ``Antiwindup strategy for {PI}-type speed controller,'' \emph{{IEEE} Trans. Ind. Electron.}, vol.~56, no.~6, pp. 2039--2046, 2009.

\bibitem{ghoshal2010anti}
A.~Ghoshal and V.~John, ``Anti-windup schemes for proportional integral and proportional resonant controller.'' in \emph{National Power Electronic Conference}, 2010.

\bibitem{ajala2021model}
O.~Ajala, M.~Lu, B.~Johnson, S.~V. Dhople, and A.~Dom{\'\i}nguez-Garc{\'\i}a, ``Model reduction for inverters with current limiting and dispatchable virtual oscillator control,'' \emph{{IEEE} Trans. Energy Convers.}, vol.~37, no.~4, pp. 2250--2259, 2022.

\bibitem{zhang2023active}
H.~Zhang, R.~Liu, C.~Xue, and Y.~Li, ``Active power enhancement control strategy of grid-forming inverters under asymmetrical grid faults,'' \emph{{IEEE} Trans. Power Electron.}, vol.~39, no.~1, pp. 1447--1459, 2023.

\bibitem{chen2020use}
J.~Chen, F.~Prystupczuk, and T.~O'Donnell, ``Use of voltage limits for current limitations in grid-forming converters,'' \emph{{CSEE} J. Power Energy Syst.}, vol.~6, no.~2, pp. 259--269, 2020.

\bibitem{jiang2024current}
S.~Jiang, Y.~Zhu, T.~Xu, and G.~Konstantinou, ``Current-synchronization control of grid-forming converters for fault current limiting and enhanced synchronization stability,'' \emph{{IEEE} Trans. Power Electron.}, vol.~39, no.~5, pp. 5271--5285, 2024.

\bibitem{wu2024design}
H.~Wu, X.~Wang, and L.~Zhao, ``Design considerations of current-limiting control for grid-forming capability enhancement of {VSCs} under large grid disturbances,'' \emph{{IEEE} Trans. Power Electron.}, vol.~39, no.~10, pp. 12\,081--12\,085, 2024.

\end{thebibliography}

\vspace{-8mm}
\begin{IEEEbiography}
[{\includegraphics[width=1in,height=1.25in,clip,keepaspectratio]{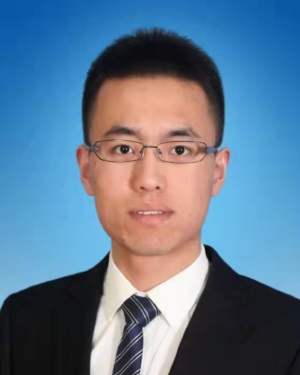}}]{Xiuqiang He (Member, IEEE)}
received the B.S. and Ph.D. degrees in control science and engineering from Tsinghua University, Beijing, China, in 2016 and 2021, respectively. From 2021 to 2024, he was a Postdoctoral Researcher with the Automatic Control Laboratory, ETH Zürich, Zürich, Switzerland, where he is currently a Senior Scientist. His research interests include power system dynamics, stability, and control, involving multidisciplinary expertise in automatic control, power systems, power electronics, and renewable energy generation. Dr. He was the recipient of the Beijing Outstanding Graduates Award and the Outstanding Doctoral Dissertation Award from Tsinghua University.
\end{IEEEbiography}

\vspace{-8mm}

\begin{IEEEbiography}[{\includegraphics[width=1in,height=1.25in,clip,keepaspectratio]{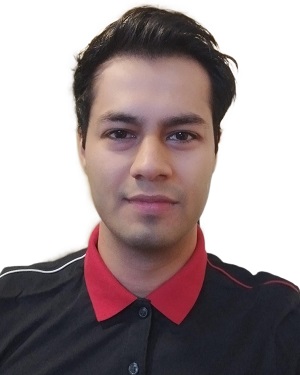}}]{Maitraya Avadhut Desai (Graduate Student Member, IEEE)} received the B.Tech. degree in electrical engineering from Veermata Jijabai Technological Institute, Mumbai, India, in 2021. He received the M.Sc. degree in electrical engineering and information technology from the Swiss Federal Institute of Technology (ETH), Zürich, Switzerland, where he has been working towards the Ph.D. degree in electrical engineering with the Power Systems Laboratory since February 2024. His research interests include modeling, control, and optimization of inverter-based power systems.
\end{IEEEbiography}

\vspace{-8mm}

\begin{IEEEbiography}
[{\includegraphics[width=1in,height=1.25in,clip,keepaspectratio]{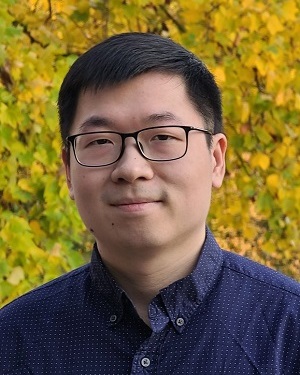}}]{Linbin Huang (Member, IEEE)}
received the B.Eng. and Ph.D. degrees in electrical engineering from Zhejiang University, Hangzhou, China, in 2015 and 2020, respectively. From 2020 to 2024, he was a Postdoctoral Researcher and then a Senior Scientist with the Automatic Control Laboratory at ETH Zürich, Zürich, Switzerland. He is currently a Professor with the College of Electrical Engineering, Zhejiang University, Hangzhou, China. His research interests include power system stability, optimal control of power electronics, and data-driven control.
\end{IEEEbiography}

\vspace{-8mm}

\begin{IEEEbiography}
[{\includegraphics[width=1in,height=1.25in,clip,keepaspectratio]{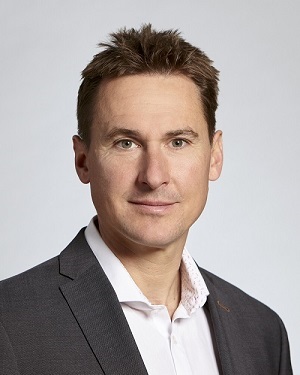}}]{Florian Dörfler (Senior Member, IEEE)}
received the Diploma degree in engineering cybernetics from the University of Stuttgart, Stuttgart, Germany, in 2008, and the Ph.D. degree in mechanical engineering from the University of California at Santa Barbara, Santa Barbara, CA, USA, in 2013. From 2013 to 2014, he was an Assistant Professor with the University of California Los Angeles. He has been serving as the Associate Head of the ETH Zürich Department of Information Technology and Electrical Engineering from 2021 until 2022. He is currently a Professor with the Automatic Control Laboratory, ETH Zürich, Zürich, Switzerland. His research interests include automatic control, system theory, and optimization. His particular foci are on network systems, data-driven settings, and applications to power systems. Dr. Dörfler was the recipient of the distinguished young research awards by IFAC (Manfred Thoma Medal 2020) and EUCA (European Control Award 2020). His students were winners or finalists for Best Student Paper awards at the European Control Conference (2013, 2019), the American Control Conference (2016,2024), the Conference on Decision and Control (2020), the PES General Meeting (2020), the PES PowerTech Conference (2017), the International Conference on Intelligent Transportation Systems (2021), the IEEE CSS Swiss Chapter Young Author Best Journal Paper Award (2022,2024), and the IFAC Conference on Nonlinear Model Predictive Control (2024). He is furthermore a recipient of the 2010 ACC Student Best Paper Award, the 2011 O. Hugo Schuck Best Paper Award, the 2012-2014 Automatica Best Paper Award, the 2016 IEEE Circuits and Systems Guillemin-Cauer Best Paper Award, the 2022 IEEE Transactions on Power Electronics Prize Paper Award, the 2024 Control Systems Magazine Outstanding Paper Award, and the 2015 UCSB ME Best PhD award. He is currently serving on the council of the European Control Association and as a Senior Editor of Automatica.
\end{IEEEbiography}

\end{document}